\documentclass[a4paper,11pt]{article}
\usepackage{jheppub} % for details on the use of the package, please
\usepackage[T1]{fontenc} % if needed
\usepackage{mathrsfs}
\usepackage{arydshln}
\usepackage{pdflscape}
\usepackage{subcaption}
\usepackage[normalem]{ulem}
\usepackage{slashed}
\usepackage{tikz-cd}
\usepackage{graphicx}
\usepackage{caption}
\definecolor{PineGreen}{rgb}{0.0,0.47,0.44}
\definecolor{MidnightBlue}{rgb}{0.1,0.1,0.44}
\definecolor{magenta}{rgb}{1.0,0.0,1.0}
\definecolor{org1}{rgb}{.92,.39,.21}
\definecolor{pur2}{rgb}{.53,.47,.7}
\definecolor{darkTeal}{HTML}{0A4F4D}
\definecolor{brickRed}{HTML}{AB2229}
\usepackage{hyperref}
\hypersetup{
	colorlinks=true,
	linkcolor=blue,
	citecolor=PineGreen,
	filecolor=magenta,
	urlcolor=blue
}
\usepackage{amsthm}
\newtheorem{theorem}{Theorem}
\numberwithin{theorem}{section}
\newtheorem{proposition}[theorem]{Proposition}

\newtheorem{corollary}[theorem]{Corollary}
\theoremstyle{definition}

\newtheorem{definition}[theorem]{Definition}

\theoremstyle{remark}
\newtheorem{remark}[theorem]{Remark}
\newtheorem{example}[theorem]{Example}

\newcommand{\Gr}{\mathbf{Gr}}

\newcommand{\RR}{\mathbb{R}}

\newcommand{\QQ}{\mathbb{Q}}
\newcommand{\PP}{\mathbb{P}}

\newcommand{\CC}{\mathbb{C}}
\newcommand{\ZZ}{\mathbb{Z}}
\newcommand{\NN}{\mathbb{N}}

\newcommand{\jac}{\mathscr{J}}
\newcommand{\U}{\mathcal{U}}
\newcommand{\F}{\mathcal{F}}
\newcommand{\G}{\mathcal{G}}

\newcommand{\conv}{\mathrm{conv}}
\newcommand{\vol}{\mathrm{vol}}
\newcommand{\newt}{\mathrm{Newt}}
\newcommand{\supp}{\mathrm{Supp}}
\newcommand{\rank}{\mathrm{rank}}
\newcommand{\cI}{\mathcal{I}}
\newcommand{\cZ}{\mathcal{Z}}

\usepackage{tikz}    
\usetikzlibrary{graphs,graphs.standard,calc,automata}

\usepackage{tikz-feynman}
\tikzfeynmanset{compat=1.1.0}

\newcommand{\mCm}{\mathcal{Y}}
\newcommand{\cY}{\mathcal{Y}}

\newcommand{\Gm}{G}
\newcommand{\cQ}{\mathcal{Q}}

\newcommand{\be}{\begin{equation}}
\newcommand{\ee}{\end{equation}}
\newcommand{\ma}{\texttt{Mathematica}}

\title{\boldmath Symbol Alphabets from the Landau Singular Locus}

\author[a]{Christoph Dlapa,}
\author[b]{Martin Helmer,}
\author[a]{Georgios Papathanasiou\,}
\author[a]{and Felix Tellander\,}

\affiliation[a]{Deutsches Elektronen-Synchrotron DESY,\\ Notkestr.~85, 22607 Hamburg, Germany.}
\affiliation[b]{Department of Mathematics,\\ North Carolina State University, Raleigh, USA.}

\emailAdd{christoph.dlapa@desy.de}
\emailAdd{mhelmer@ncsu.edu}
\emailAdd{georgios.papathanasiou@desy.de}
\emailAdd{felix@tellander.se}

\preprint{DESY-23-048}

\abstract{We provide evidence through two loops, that rational letters of polylogarithmic Feynman integrals are captured by the Landau equations, when the latter are recast as a polynomial of the kinematic variables of the integral, known as the principal $A$-determinant. Focusing on one loop, we further show that all square-root letters may also be obtained, by re-factorizing the principal $A$-determinant with the help of Jacobi identities. We verify our findings by explicitly constructing canonical differential equations for the one-loop integrals in both odd and even dimensions of loop momenta, also finding agreement with earlier results in the literature for the latter case. We provide a computer implementation of our results for the principal $A$-determinants, symbol alphabets and canonical differential equations in an accompanying \ma\ file. Finally, we study the question of when a one-loop integral satisfies the Cohen-Macaulay property and show that for almost all choices of kinematics the Cohen-Macaulay property holds. Throughout, in our approach to Feynman integrals, we make extensive use of the Gel'fand, Graev, Kapranov and Zelevinski\u{\i} theory on what are now commonly called  GKZ-hypergeometric systems whose singularities are described by the principal $A$-determinant.

}

\begin{document} 
\maketitle
\flushbottom

\newpage

\renewcommand{\labelenumii}{\arabic{enumi}.\arabic{enumii}}

%%%%%%%%%%%%%%%%%%%%%%%%%%%%%%%%%%%%%%%%%%%%%%%%%%%%%%%%%%%%%%%%%%%%%%%%%%%%%%%%%%%%%%%%
\section{Introduction}
Feynman integrals are central objects in theoretical physics, for example, their evaluation is central for the calculation of any scattering amplitude in high-energy physics \cite{Weinberg:1995mt}. This not only includes experiments using the Large Hadron Collider at CERN but also amplitudes in gravitational wave physics \cite{Bern:2019nnu,Dlapa:2022lmu} or the critical exponent in statistical field theory \cite{zinn2021quantum}.

Evaluation of these integrals is a challenging problem which has fostered two main lines of development: advanced numerical schemes have been developed for fast and accurate direct evaluation (see e.g.~\cite{Borowka:2017idc,Smirnov:2021rhf} or \cite{Borinsky:2020rqs,Borinsky:2023jdv}) and  sophisticated  analytical { computer tools have been developed for either direct evaluation or for understanding the analytic structure, see for example \cite{Lee:2013mka,Smirnov:2019qkx,vonManteuffel:2012np,Maierhofer:2018gpa,Panzer:2014caa,Duhr:2019tlz,Prausa:2017ltv,Gituliar:2017vzm,Meyer:2017joq,Dlapa:2020cwj,Lee:2020zfb,Belitsky:2022gba,Jantzen:2012mw,Peraro:2019svx}}. Focusing on the latter, the topic is classical and dates back to the seminal work of Landau \cite{Landau:1959fi}, Cutkosky \cite{Cutkosky:1960sp} and the $S$-matrix program of the 1960's \cite{Eden:1966dnq}. At that time it was understood by Regge that every Feynman integral is a solution to a system of partial differential equations (PDEs) of ``hypergeometric type'' \cite{Regge1968}, and a strong connection with the Japanese $D$-module school of Sato and Kashiwara was established \cite{Kashiwara1977,Kashiwara1978}. 

Later in the 1980's, a completely combinatorial description of a vast family of $D$-modules, to which all Feynman integrals belong, was given by Gel'fand and collaborators \cite{Gelfand1989}.  $D$-modules in this family are now commonly referred to as {\em GKZ-hypergeometric systems}. The singularities of these GKZ-hypergeometric systems are described by a polynomial known as the {\em principal $A$-determinant}. When the GKZ-hypergeometric system under consideration arises from a Feynman integral, the principal $A$-determinant then describes the kinematic singularities of the Feynman integral, { namely the values of the kinematic parameters of the integral, for which it may become singular.  The zero set of the principal $A$-determinant is commonly referred to in physics as the {\em Landau singular locus}, in other words it is} the solution to the Landau equations \cite{Landau:1959fi}, where the presence of kinematic singularities is formulated as a condition for the contour of integration to become trapped.  One important property of GKZ-hypergeometric systems is the {\em Cohen-Macaulay property}; when a GKZ-hypergeometric system has this property its {\em rank} is given by a simple combinatorial formula and series solutions to the system may be obtained in a much more straightforward manner.  
In the context of Feynman integrals, the rank of the GKZ-hypergeometric system bounds the number of {\em master integrals}, to be defined below.   
These connections have been rediscovered in recent years attracting a lot of interest, see e.g.~\cite{delaCruz:2019skx,Klausen:2019hrg,Klemm:2019dbm,Ananthanarayan:2022ntm}.

Parallel to this, analytic evaluation approaches using partial differential equations were also developed natively within the physics community \cite{Kotikov:1990kg,Remiddi:1997ny,Gehrmann:1999as}, culminating in what is now called \emph{canonical differential equations} \cite{Henn:2013pwa}. When a Feynman integral can be represented like this it can be expressed as a Chen iterated integral \cite{Chen1977} which, when the kernels are rational, further reduces to the well-studied class of multiple polylogarithms (MPLs) \cite{goncharov2001multiple,Goncharov2005}.

More concretely, these approaches are based on the solution of integration by parts identities (IBPs) \cite{Chetyrkin:1981qh}, namely linear relations between any set of Feynman integrals with integer propagator powers, for example those contributing to a given process. The master integrals alluded to before are precisely a finite basis $\vec{g}$ in this linear space, obtained by solving the identities in question. Derivatives of the master integrals may then be re-expressed in terms of this basis. Regularizing the infrared and ultraviolet divergences of the integrals by setting the dimension of the loop momenta to $D=D_0-2\epsilon$, the canonical transformation that greatly facilitates their solution is then a change of basis such that
\begin{equation}\label{eq:DE}
 d\vec{g}=\epsilon\, d\widetilde{M}\, \vec{g},
\end{equation}
where we have grouped all partial derivatives with respect to the independent kinematic variables of the integrals, $v_{i}$, into the total differential $d=\sum_i  dv_{i} \partial_{v_{i}}$. The matrix $\widetilde{M}$ is independent of $\epsilon$ and takes the simple dlog-form,
\begin{equation}
\label{eq:can-ansatz}
 \widetilde{M}\equiv \sum_{i} \tilde{a}_i\log W_i,
\end{equation}
where the $\tilde{a}_i$ are constant matrices and the $W_i$ carry all kinematic dependence and are called \emph{letters}. The set of all letters is called the \emph{alphabet}. Evidence suggests that a transformation to the first equality is always possible, however the particular form of the $\widetilde{M}$ matrix given by the second equality is expected to exist only when the master integrals live in the aforementioned class of MPLs. The latter class of integrals will be the focus of this paper.

\iffalse
centered around solving integration by parts identities (IBPs) \cite{Chetyrkin:1981qh} and
One can then derive the following set of PDEs for $\vec{f}$:
\begin{equation}
 \frac{\partial}{\partial v_{i}}\vec{f}=M_{v_{i}}\vec{f},
\end{equation}
where we used dimensional regularization with $D=D_0-2\epsilon$ and $v_i$ are independent kinematic variables.

A common way of solving the PDEs \cite{Henn:2013pwa} is then to transform to a basis of \emph{pure} integrals $\vec{g}=T\vec{f}$ which satisfies
\begin{equation}
\label{eq:PDEs-canonical}
 \frac{\partial}{\partial v_{i}}\vec{g}=\epsilon\widetilde{M}_{v_{i}}\vec{g},
\end{equation}
where the matrices $\widetilde{M}_{v_{i}}$ are independent of $\epsilon$ and are free of double or higher-order poles. One can unify the above PDEs by writing
\begin{equation}
 \mathrm{d}\vec{g}=\epsilon\,\mathrm{d}\widetilde{M}\, \vec{g},
\end{equation}
where $\widetilde{M}$ is identified by $\widetilde{M}_{v_{i}}=\frac{\partial}{\partial v_{i}}\widetilde{M}$. Since the differential equation matrices only have simply poles it is possible to express them in terms of logarithms
\begin{equation}
\label{eq:can-ansatz}
 \widetilde{M}=\sum_{i=1}^n \tilde{a}_i\log W_i,
\end{equation}
where the $\tilde{a}_i$ are constant matrices and the $W_i$ carry all kinematic dependence and are called \emph{letters}. The set of all letters is called the \emph{alphabet}.
\fi 

The merit of the canonical differential equations~\eqref{eq:DE}-\eqref{eq:can-ansatz} for the integrals $\vec{g}$ is that they can now be easily solved as an expansion\footnote{We assume that the integrals are normalized such that they have uniform transcendental weight zero.} in $\epsilon$
\begin{equation}
 \vec{g}=\sum_{k=0}^{\infty}\epsilon^k\vec{g}^{(k)},
\end{equation}
where at each order $\vec{g}^{(k)}$ the \emph{symbol} $\mathcal{S}$~\cite{Goncharov:2010jf}, capturing the full answer up to transcendental constants, is given by
\begin{equation}
\label{eq:symbol}
 \mathcal{S}(\vec{g}^{(k)})=\sum_{i_1,\ldots,i_k=1}^n \tilde{a}_{i_k}\cdot\tilde{a}_{i_{k-1}}\cdots\tilde{a}_{i_1}\cdot\vec{g}^{(0)}\ W_{i_1}\otimes\cdots\otimes W_{i_k}.
\end{equation}

In practical applications, two major bottlenecks that one often encounters in the above procedure are solving the IBPs analytically so as to find a basis of master integrals, and finding the transformation that relates this to a new, canonical basis~\eqref{eq:can-ansatz}.
However, with knowledge of the letters $W_i$ it would be possible to avoid doing these two steps analytically by using the latter equation as an ansatz. The unknown coefficient matrices $\tilde{a}_i$ can then be fixed by matching the partial derivatives of the ansatz to multiple numerical evaluations of \eqref{eq:DE} derived through numeric IBP identities (over finite fields if necessary). Similar approaches have, for example, been used in~\cite{Abreu:2020jxa}. 

Motivated by the great potential benefits of this alternative route, in this paper we will open a new door to obtaining the symbol alphabet from the Landau equations, when recast in the form of the aforementioned principal $A$-determinant, before attempting to analytically evaluate the integrals. Many crucial results in theoretical physics have been obtained by analyzing the Landau equations, whose study has recently received renewed interest, see for example~\cite{Dennen:2015bet,Prlina:2018ukf,Collins:2020euz,Mizera:2021fap,Klausen:2021yrt,Mizera:2021icv,Hannesdottir:2021kpd,Correia:2021etg,Hannesdottir:2022xki,Lippstreu:2022bib,Berghoff:2022mqu} for an incomplete list. From eqs.~\eqref{eq:DE}-\eqref{eq:can-ansatz} it is evident that values of the kinematic variables where the letters $W_i$ vanish are potential branch points of Feynman integrals, and it is well known that these values are indeed captured by the Landau equations. However this information is in general not enough for fixing the entire functional form of the letters.

{ Quite remarkably, here we observe that the principal $A$-determinant of a Feynman integral, in the natural factorization it is endowed with as a function of its kinematic variables, coincides with the product of rational letters of the integral in question!} We will provide precise definitions in the following Sections, but let us give an idea of this identification for the one-loop `two-mass easy' box with all internal masses being zero, and additionally $p_2^2=p_4^2=0$, $p_1^2,p_3^2\neq 0$: 
\begin{equation*}
    \centering
     \begin{tikzpicture}[baseline=-\the\dimexpr\fontdimen22\textfont2\relax]
    \begin{feynman}
    \vertex (a);
    \vertex [right = of a] (b);
    \vertex [below = of b] (c);
    \vertex [below = of a] (d);
    \vertex [above left = of a] (x){\(p_1\)};
    \vertex [above right = of b] (y){\(p_4\)};
    \vertex [below right = of c] (z){\(p_3\)};
    \vertex [below left = of d] (w){\(p_2\)};
    \diagram*{
        (x)--[fermion](a), (y)--[fermion,color=gray](b), (z)--[fermion](c), (w)--[fermion,color=gray](d), (a) -- [edge label=\(\textcolor{black}{x_1}\),color=gray] (b) -- [edge label=\(\textcolor{black}{x_4}\),color=gray] (c) --[edge label=\(\textcolor{black}{x_3}\),color=gray] (d) -- [edge label=\(\textcolor{black}{x_2}\),color=gray](a),
    };
    \end{feynman}
    \end{tikzpicture}
\end{equation*}
The principal $A$-determinant of this Feynman integral is
\begin{equation}\label{eq:2meEA}
    \widetilde{E_A}=\underbrace{(p_1^2p_3^2-st)}_{\mathrm{type-I}}\underbrace{p_1^2p_3^2st(p_1^2+p_3^2-s-t)(p_3^2-t)(p_3^2-s)(p_1^2-t)(p_1^2-s)}_{\mathrm{type-II}}.
\end{equation}
where $s=(p_1+p_2)^2$ and $t=(p_1+p_4)^2$ are Mandelstam invariants. The 10 factors above in fact coincide with the 10 letters in the symbol alphabet of this diagram. In eq.~\eqref{eq:2meEA} we have also indicated whether each factor describes a type-I and type-II Landau singularity, associated to the entrapment of the integration contour at finite or infinite values of the loop momentum, respectively (as also reviewed in Section \ref{sec: Singularities and Principal A-determinant}). In much of the relevant literature there has been a tendency to focus on type-I singularities, however {here we wish to emphasize}  that for symbol alphabets the type-II singularities in general cannot be neglected.

We will also provide two-loop evidence of the connection between principal $A$~-determinants and rational letters, but in this paper we will mainly focus on further extracting letters containing square roots from them, which is well known that already appear in one-loop integrals. For these integrals, we will also prove that the Cohen-Macaulay property holds.

The paper is organized as follows. In Section~ \ref{sec:FeynmanGGKZ} we provide the necessary background for Feynman integrals and generalized hypergeometric systems. In particular we connect the Landau singularities to the principal $A$-determinant. In Section~\ref{sec: symbol alphabets} we restrict our focus to one-loop graphs and provide the full principal $A$-determinant, symbol alphabet and canonical differential equations for all graphs with generic kinematics. The process of {starting with the case of generic kinematics and} taking limits {to obtain} non-generic kinematics, as well as the connection to previous work, are also discussed. Explicit examples are provided in Section~\ref{sec: examples}. In Section~\ref{sec: normality} we prove that the Cohen-Macaulay property holds for one-loop graphs for almost all choices of kinematics (in fact we prove a stronger sufficient condition referred to as {\em normality}) and discuss the generalized permutohedron property. Finally in Section~\ref{sec: outlook} we provide conclusions and outlook. Our results on the principal $A$-determinant, symbol alphabet and differential equations have also been implemented in the \texttt{Mathematica} notebook \texttt{LandauAlphabetDE.nb} accompanying the version of this paper in the \texttt{arXiv}.

\medskip

\noindent {\bf Note added:} While this project was in the process of writing up, we became aware of the recent preprint~\cite{Jiang:2023qnl}, which overlaps in part with the results presented in subsection~\ref{subsec:SymbolLettersFormula}.

%%%%%%%%%%%%%%%%%%%%%%%%%%%%%%%%%%%%%%%%%%%%%%%%%%%%%%%%%%%%%%%%%%%%%%%%%%%
\section{Feynman integrals, Landau singularities and GKZ systems}\label{sec:FeynmanGGKZ}
In this Section we establish our conventions on Feynman scalar integrals, and review how they can be interpreted as GKZ hypergeometric systems when expressed in the Lee-Pomeransky representation. We also recall how a natural object in this framework, the principal $A$-determinant, captures the Landau singularities of these integrals, mostly following~\cite{Klausen:2021yrt}.

\subsection{Feynman integrals in the Lee-Pomeransky representation}\label{section:LeePomRep}
In this paper we consider one-particle irreducible Feynman graphs $G:=(E,V)$ with internal edge set $E$, vertex set $V$ and loop number $L=|E|-|V|+1$. Every edge $e\in E$ is assigned an arbitrary direction with which we define the \emph{incidence matrix} $\eta_{ve}$ of $G$ to satisfy $\eta_{ve}=1$ if $e$ ends at $v$, $-1$ if $e$ starts at $v$, and $0$ otherwise. The vertex set $V$ has the disjoint partition $V=V_{\mathrm{ext}}\sqcup V_{\mathrm{int}}$ where each vertex $v\in V_{\mathrm{ext}}$ is assigned an external incoming $d$-dimensional momenta  $p_v\in\RR^{1,d-1}$ with the mostly minus convention $(p_v)^2=(p_v^0)^2-(p_v^1)^2-\cdots$ and we put $p_v=0$ for all $v\in V_{\mathrm{int}}$. Using Feynman's causal $i\varepsilon$ prescription, scalar Feynman rules assigns the following integral to $G$:
\begin{equation}\label{eq: FeynIntProp}
    \cI=\left(\frac{e^{\gamma_E\epsilon}}{i\pi^{D/2}}\right)^L\lim_{\varepsilon\to 0^+}\int\prod_{e\in E}d^Dq_e\left(\frac{-1}{q_e^2-m_e^2+i\varepsilon}\right)^{\nu_e}\prod_{v\in V\setminus \{v_0\}}\delta^{(D)}\left(p_v+\sum_{e\in E}\eta_{ve}q_e\right)
\end{equation}
where $\gamma_E=-\Gamma'(1)\simeq 0.577$ is the Euler-Mascheroni constant, $\nu_e\in\ZZ$ are generalized propagator powers and $q_e$ is the total internal momenta flowing through the edge $e$. We are also employing dimensional regularization with $D:=D_0-2\epsilon$, and while the external and (integer part of) the internal momenta dimensions are usually taken to coincide, $D_0=d$, here we will distinguish between the two. Physically this can be thought of as restricting one set of momenta to lie in a subspace of the other, and is further justified by the alternative parametric representations of Feynman integrals, which we will get to momentarily. Momentum is conserved at each vertex $v\in V$, but only $|V|-1$ of these constraints are independent, we therefore remove an arbitrary vertex $v_0$ from $V$ in (\ref{eq: FeynIntProp}) to avoid imposing $\delta^{(D)}\left(\sum_{v\in V}p_v\right)$ explicitly. 

To evaluate the integrals we rewrite them in parametric form 
\begin{equation}\label{eq: FeynIntWL}
     \cI=e^{L\gamma_E\epsilon}\Gamma(\omega)\lim_{\varepsilon\to 0^+}\int_0^\infty\prod_{e\in E}\left(\frac{x^{\nu_e}dx_e}{x_e\Gamma(\nu_e)}\right)\frac{\delta(1-H(x))}{\U^{D/2}}\left(\frac{1}{\F/\U-i\varepsilon\sum_{e\in E}x_e}\right)^\omega
\end{equation}
where $\omega:=\sum_{e\in E}\nu_e-LD/2$ is the superficial degree of divergence, $H:\RR^{{|E|}}\to\RR_+$ is a homogeneous function of degree one and $x_e,\ e\in E$, are the \emph{Schwinger/Feynman parameters}. These are defined by the $\Gamma$-function identity
\begin{equation}
    \left(\frac{i}{q_e^2-m_e^2+i\varepsilon}\right)^{\nu_e}=\frac{1}{\Gamma(\nu_e)}\int_0^\infty\frac{dx_e}{x_e}\ x_e^{\nu_e}\exp\left[ix_e(q_e^2-m_e^2+i\varepsilon)\right]
\end{equation}
which is absolutely convergent when $\varepsilon>0$ and the real part of $\nu_e$ is positive, $\Re(\nu_e)>0$, but can be analytically continued to all $\nu_e\in\CC$. In the representation \eqref{eq: FeynIntWL} the information on the Feynman graph $G$ has been encoded in two homogeneous \emph{Symanzik polynomials} $\U$ and $\F$, of degree $L$ and $L+1$ in the integration variables, respectively. As is reviewed in e.g.~\cite{Weinzierl:2022eaz}, by virtue of the matrix tree theorem these are equal to
\begin{align}
    \mathcal{U}&=\sum_{\substack{~~~\,T {\rm \;a \; spanning} \\ {\rm tree \; of \; }G}}\;\;\;\prod_{e\not\in T}x_e,\label{eq:U_Symanzik}\\
    \mathcal{F}&=\F_m+\F_0= \mathcal{U}\sum_{e\in E}m_e^2x_e- \sum_{\substack{F {\rm \;a \; spanning} \\ {\rm 2-forest \; of \; }G}}\;\;\; p(F)^2\prod_{e\not\in F}x_e,\label{eq:F_Symanzik}
\end{align} 
where a spanning tree is a connected subgraph of $G$ which contains all of its vertices but no loops, and the spanning two-forest is defined similarly, but now has two connected components. For each spanning two-forest $F=(T,T')$ of $G$ we let $p(F)=\sum_{v\in T\cap V_{\rm ext}}p(v)$ denote the total momentum flowing through cut.

In this paper we will consistently think of the Feynman integral of a Feynman diagram $G$ in their parametric representation due to Lee and Pomeransky \cite{Lee2013}, that is we consider integrals:  
\begin{equation}
\label{eq:I-def}
   \cI=e^{L\gamma_E\epsilon}\frac{\Gamma(D/2)}{\Gamma(D/2-\omega)}\int_0^\infty\prod_{e \in E}\left(\frac{x_e^{\nu_e}dx_e}{x_e\Gamma(\nu_e)}\right)\frac{1}{\G^{D/2}}
\end{equation}
where
\begin{equation}
 \mathcal{G}=\mathcal{U}+\mathcal{F}   
\end{equation}
and the dependence on $i\varepsilon$ has been suppressed as it will not play a role in the rest of the paper.  Going from \eqref{eq:I-def} back to \eqref{eq: FeynIntWL} is done by inserting $1=\int_0^\infty\delta(t-H(x))dt$, re-scaling the variables $x_e\to tx_e,\ t>0$ and performing the $t$-integral.

When a Feynman integral is written in Lee-Pomeransky form it is a generalized hypergeometric integral \cite{delaCruz:2019skx,Klausen:2019hrg} of the form studied by Passare and collaborators \cite{Nilsson2013,Berkesch2014}. As a consequence it is also a solution to a generalized hypergeometric system of partial differential equations in the sense of Gel'fand, Graev, Kapranov and Zelevinski\u{\i} (GGKZ, commonly shortened to GKZ) \cite{Gelfand1986,Gelfand1987,Gelfand1989,Gelfand1990,Gelfand1993}. The singularities of these hypergeometric systems are described by the {\em principal $A$-determinant}, see \cite[\S3]{adolphson1994hypergeometric} and \cite[Chapter 9]{gelfand2008discriminants} or \cite[Theorem~1.36]{cattani2006three} or \cite[\S3]{Klausen:2021yrt}. We will define the principal $A$-determinant in Section \ref{sec: Singularities and Principal A-determinant} below; in the context of Feynman integrals { the zero set of the principal $A$-determinant contains all kinematic points} where the Feynman integral fails to be an analytic function. 

Using multi-index notation we may write the Lee-Pomeransky polynomial as $\G=\sum_{i=1}^rc_ix^{\alpha_i}$ with $c_i\neq 0$ and $\alpha_i\in \ZZ_{\ge 0}^{{|E|}}$ for all $i=1,\ldots,r$. We define the two matrices
\begin{align}
    A&:=\{1\}\times A_-=\begin{pmatrix}
        1&1&\cdots&1,\\
        \alpha_1&\alpha_2&\cdots&\alpha_r
    \end{pmatrix}\in\ZZ_{\ge 0}^{({|E|}+1)\times r}, \; {\rm and}\label{eq:AMatrix}\\
    \beta&:=\begin{pmatrix}
        -D/2, & -\nu_1 &, \ldots &, -\nu_{{|E|}}
    \end{pmatrix}^T\in\RR^{{|E|}+1},\label{eq:BetaVector}
\end{align}
where $A_-:=\mathrm{Supp}(\G)$ is the monomial {\em support} of $\G$, that is the matrix whose columns are the exponent vectors of the monomials appearing with non-zero coefficients in $\G$. From these two matrices the GKZ hypergeometric system can be defined as a left-ideal $H_A(\beta)$ in the Weyl algebra $W:=\QQ(\beta)[c_1,\ldots,c_r]\langle\partial_1,\ldots,\partial_r\rangle$ where $\partial_i$ denotes the partial differential operator associated to $c_i$ (cf. \cite{saito2013grobner}). The hypergeometric ideal $H_A(\beta)$ can be written as the sum of the two ideals
\begin{align}
    I_A&:=\left\langle\partial^u-\partial^v\ |\ u,v\in\ZZ_{\ge 0}^{r}\ \mathrm{s.t.}\ Au=Av\right\rangle, \;\;{\rm and}\label{eq:toricIdeal}\\
    Z_A(\beta)&:=\left\langle \Theta_i(c, \partial)\;|\; \Theta= A\cdot\begin{pmatrix}
        c_1\partial_1\\
        \vdots\\
        c_r\partial_r
    \end{pmatrix}-\beta\right\rangle,
\end{align}
where $\Theta$ is a vector containing ${|E|}+1$ polynomials. Note that a Feynman integral $\cI$ is annihilated by all polynomials in the left-ideal $H_A(\beta):=I_A+Z_A(\beta)$, i.e. $H_A(\beta)\bullet\cI=0$, hence, from an  analytic perspective,  $H_A(\beta)$ is a system of partial differential equations. We also note that by definition the ideal $I_A$ is in fact an ideal in the {\em commutative polynomial ring} $\QQ[\partial_1, \cdots,\partial_r]$,  (which we consider as a left ideal in the Weyl algebra) and has a finite generating set $I_A=\langle h_1, \dots h_\ell  \rangle$ with $h_i\in \QQ[\partial_1, \cdots,\partial_r]$. This ideal $I_A$ is a prime ideal whose finite generating set consists of binomials and is often referred to as a {\em toric ideal} as it gives the defining equations of the projective {\em toric variety} $$
X_A=\{ z\in \PP^{r-1} \; |\; h_1(z)=\cdots= h_\ell(z)=0 \}
$$
associated to the matrix $A$, see e.g.~\cite{eisenbud1996binomial}, \cite[II, Chapter~5]{gelfand2008discriminants}.

%%%%%%%%%%%%%%%%%%%%%%%%%%%%%%%%%%%%%%%%%%%%%%%%%%%%%%%%%%%%%%%%%%%%%%%%%%%%%%%%%%%
\subsection{Landau singularities and the principal $A$-determinant}\label{sec: Singularities and Principal A-determinant}
In going from the momentum space representation \eqref{eq: FeynIntProp} to the parametric representation \eqref{eq: FeynIntWL} one has the intermediate step
\begin{equation}
    \cI=\int \prod_{l=1}^L\frac{d^Dk_l}{i\pi^{D/2}}\int_0^\infty\prod_{e\in E}dx_e\frac{\delta(1-H(x))}{\cQ^{\sum_e\nu_e}}
\end{equation}
where the momentum integrals is over the $L$ independent loop momenta and
\begin{equation}
    \cQ=\sum_{e\in E}x_e(-q_e^2+m_e^2).
\end{equation}
The original Landau analysis \cite{Landau:1959fi} involves finding not only when $\cQ$ is zero but when it has a stationary point so that the integration contour becomes pinched between poles of the integrand. These conditions are expressed in the \emph{Landau equations}
\begin{equation}
    \begin{cases}
        \cQ=0\\
        \frac{\partial}{\partial k_l}\cQ=0,\ \forall l=1,\ldots,L.
    \end{cases}
\end{equation}
If all loop momenta are kept finite the solutions are called type-I singularities while if all loop momenta are infinite it is referred to as a type-II singularity (first observed by Cutkosky \cite{Cutkosky:1960sp}). In general some loop momenta can be finite while some pinch at infinity giving a \emph{mixed} type-II singularity. However, at one-loop there are no such mixed singularities as there is only one loop momentum. The Landau equations have also been studied in the parametric representation \eqref{eq: FeynIntWL}, see for example \cite{Polkinghorne:1960a,Polkinghorne:1960b}. After introducing some machinery from the GKZ-formalism, we will see that solutions of the Landau equations in the Lee-Pomeransky formalism are associated with the vanishing of the principal $A$-determinant (cf. Definition~\ref{def: principal A-determinant}).

Following the notation in \cite{gelfand2008discriminants} we let $A=\{a_1,\ldots,a_n\}\subset\ZZ^{k-1}$ be a set of lattice points that generates $\ZZ^{k-1}$ and let $\CC^A$ denote the finite dimensional $\CC$-vector space of all Laurent polynomials with support $A$, meaning all Laurent polynomials which can be formed from the momomials $x^{a_1}, \dots, x^{a_n}$, i.e.~$\CC^A:=\{\sum_{i=1}^nc_ix^{a_i}\;|\;c_i\in \CC\}$.

Let $\cZ_0(A)\subset(\CC^A)^k$ be the set of polynomials $(f_1(x),\ldots,f_k(x))$ for which there is $x$ in the algebraic torus $(\CC^*)^{k-1}$ satisfying $f_1(x)=\cdots f_k(x)=0$, i.e.
\begin{equation}
    \cZ_0(A):=\left\{(f_1,\ldots,f_k)\in(\CC^A)^k\,|\, \mathbf{V}\left(f_1,\ldots,f_k\right)\neq\emptyset\ \mathrm{in}\ (\CC^*)^{k-1}\right\}.
\end{equation}
The closure $\cZ(A)$ of $\cZ_0(A)$ is an irreducible hypersurface in $(\CC^A)^{k}$ over the rational numbers.
\begin{definition}[$A$-resultant]
    Since $\cZ(A)$ is an irreducible hypersurface there is an irreducible polynomial $R_A$ in the coefficients of $f_1,\ldots,f_k$ with integer coefficients that is unique up to sign. This polynomial is called the $A$-\emph{resultant}.
\end{definition}

A special case of the $A$-resultant is when $f_k=f$ and $f_i=x_i\partial f/\partial x_i$ for $i=1,\ldots,k-1$:
\begin{definition}[Principal $A$-determinant]\label{def: principal A-determinant}
    The special $A$-resultant
    \begin{equation}
        E_A(f):=R_A\left(x_1\frac{\partial f}{\partial x_1},\ldots,x_{k-1}\frac{\partial f}{\partial x_{k-1}},f\right)
    \end{equation}
    is called the \emph{principal $A$-determinant}.
\end{definition}
A major result in this field \cite[Chapter~10]{gelfand2008discriminants} is that the principal $A$-determinant can be written as a product of $A$\emph{-discriminants}.
Let $\nabla_0\subset\CC^A$ be the set defined as
\begin{equation}
    \nabla_0:=\left\{f\in\CC^A\ |\ \mathbf{V}\left(f,\frac{\partial f}{\partial x_1},\ldots,\frac{\partial f}{\partial x_k}\right)\neq\emptyset\ \mathrm{for}\ x\in(\CC^*)^k\right\}
\end{equation}
and denote by $\nabla_A$ the Zariski closure of $\nabla_0$.
\begin{definition}[$A$-discriminant]
If $\nabla_A\subset\CC^A$ has codimension 1, then the $A$-\emph{discriminant} is the irreducible polynomial $\Delta_A(f)$ in the coefficients $c_i$ of $f$ that vanishes on $\nabla_A$. If codim$\nabla_A>1$ we put $\Delta_A=1$.
\end{definition}

Again writing $X_A$ for the toric variety associated to $A$, we have that the projectivization of $\nabla_A$ is the projective dual of $X_A$, see e.g.~\cite{gelfand2008discriminants}. The faces $\Gamma\subset\conv(A)$ induce a stratification of $X_A$ with strata $X(\Gamma)$ and projective duals $\nabla_{A\cap\Gamma}$; by $A\cap \Gamma$ we mean the matrix consisting of all columns of $A$ which are also contained in the face $\Gamma$. For a variety $X(\Gamma)\subset X_A$ we denote the multiplicity of $X_A$ along $X(\Gamma)$ as $\mathrm{mult}_{X(\Gamma)}X_A$, \cite[Definition 3.15, Chapter 5]{gelfand2008discriminants}, and we have the following factorization theorem.
\begin{theorem}[Prime factorization, Theorem 1.2, Chapter 10 of \cite{gelfand2008discriminants}]\label{thm:EAfact}
Let $Q=\conv(A)$, then the principal $A$-determinant is the polynomial
\begin{equation}\label{eq:EaFact}
    E_A(f)=\pm\prod_{\Gamma\subseteq Q}\Delta_{A\cap\Gamma}(f)^{\mathrm{mult}_{X(\Gamma)}X_A}
\end{equation}
where $\Delta_{A\cap\Gamma}(f):=\Delta_{A\cap\Gamma}(\left.f\right|_\Gamma)$ and $\left.f\right|_\Gamma$ is the coordinate restriction of $f$ supported on $\Gamma$.
\end{theorem}
Calculating the  principal $A$-determinant $E_A(f)$ now comes down to three steps: calculating all the faces $\Gamma$ of $Q=\conv(A)$, calculating the multiplicities $\mathrm{mult}_{X(\Gamma)}X_A$, which are lattice indices \cite[Theorem 3.16, Chapter 5]{gelfand2008discriminants} and can be computed via integer linear algebra \cite[
Remark~2.2]{helmer2018nearest}, and finally calculating the $A$-discriminants $\Delta_{A\cap\Gamma}(f)$ which can be obtained via elimination (which can be accomplished using Gr\"{o}bner basis, see e.g.~\cite{CLO}).

In our discussion we are primarily interested in the zero set of the principal $A$-determinant $E_A(f)$, hence we may neglect the exponents appearing in  \eqref{eq:EaFact}, as these do not change the zero set. To this end we make the following definition of a {\em reduced principal $A$-determinant}, which is the unique polynomial (up to constant) that corresponds to the zero set of $E_A(f)$.\footnote{Mathematically speaking, the exponents in eq.~\eqref{eq:EaFact} do contain important information; in the context of toric geometry this e.g. pertains to the local geometry in the neighbourhood of a singular point on the toric variety. The physical meaning of these exponents in the Feynman integral context is something that certainly merits further exploraton. As a first step in this direction, we have computed these exponents for several examples of generic 1-loop graphs, finding that they are always equal to 1. We defer a more detailed investigation to future work.}

\begin{definition}[Reduced Principal $A$-determinant]Let $Q=\conv(A)$, then the reduced principal $A$-determinant is the polynomial
\begin{equation}
            \widetilde{E_A}(f)=\prod_{\Gamma\subseteq Q}\Delta_{A\cap\Gamma}(f)
\end{equation}where $\Delta_{A\cap\Gamma}(f):=\Delta_{A\cap\Gamma}(\left.f\right|_\Gamma)$ and $\left.f\right|_\Gamma$ is the coordinate restriction of $f$ supported on $\Gamma$.\label{def:reducedEA}
\end{definition}
Using the {\em homogenized} Lee-Pomeransky polynomial ${\G}_h:=\U x_0+\F$ and $A=\supp(\G_h )$\,\footnote{This is equivalent to using $A=\{1\}\times\supp(\U+\F)$, see Section \ref{sec: normality}.} in the definition of $E_A(\G_h)$ we can understand $E_A(\G_h)$ as a polynomial in the coefficients $c_i$ whose zeros correspond to coefficients such that
\begin{equation}
    \G_h=0,\ \mathrm{and\ either}\ x_i=0\ \mathrm{or}\ \frac{\partial \G_h}{\partial x_i}=0\ \forall\ i=0,\ldots,{|E|}\ \mathrm{in}\ (\CC^*)^{{|E|}+1}.
\end{equation}
Written in this way this is the ``third representation'' of the Landau equations in \cite[\S2.2]{Eden:1966dnq} with some important differences. In \cite{Eden:1966dnq} they define the singularities only in terms of the $\F$-polynomial. Using the full Lee-Pomeransky polynomial, not only do we get all the Landau singularities from the $\F$-polynomial (the type-I singularities), but also the singularities only depending on external kinematics (type-II singularities) and the mixed type-II singularities. For a recent discussion on the relation between discriminants and Landau singularities see e.g.~\cite{Klausen:2021yrt,Mizera:2021icv}. In short we get all possible singularities by using the full Lee-Pomeransky polynomial.

Using the prime factorization theorem it is easy to associate each of the non-trivial discriminants appearing in the factorization to certain type of singularities, at least for generic kinematics at one-loop where contracting edges of the graph $G$ is the equivalent to going to faces of $\newt(\G)$, as also discussed in~\cite{Klausen:2021yrt}:
\begin{itemize}
    \item $\Delta_A(\G)$ is the type-II singularity for the full graph.
    \item $\Delta_{A\cap x_i=0}(\G|_{x_i=0})$ is the type-II singularity for the sub-graph with edge $i$ contracted.
    \item $\Delta_{A\cap\newt(\F)}(\F)$ is the type-I singularity for the full graph (i.e.~the leading Landau singularity).
    \item $\Delta_{A\cap\newt(\F)\cap x_i=0}(\F|_{x_i=0})$ is the type-I singularity for the sub-graph with edge $i$ contracted.
    \item $\Delta_{A\cap\Gamma}(\G|_\Gamma)$ with $\Gamma$ having vertices both in $\supp(\U)$ and $\supp(\F)$ are mixed singularities.
\end{itemize} 
Multiple edges can be contracted to get singularities for even smaller sub-graphs. Subtleties with these identifications can appear for specific kinematic configurations, as we will illustrate in subsection~\ref{subsec:Boxes} in the concrete example of the box with three offshell external and all massless internal legs. Nevertheless, in subsection~\ref{sec: limiting procedure} we will also show that the principal $A$-determinant of a non-generic graph may be obtained as a limit of that of a generic graph, such that subtleties of this sort ultimately do not matter.

One of the main themes of this paper is to show how the principal $A$-determinant can be used to determine symbol letters. The argument for why this holds follows from the fact that every GKZ-hypergeometric system can be written as a system of first order differential equations, in this context this system is called the \emph{Pfaffian system} and can be calculated using Gr\"{o}bner basis methods \cite{saito2013grobner}. The coefficients of this system depend rationally on the kinematic variables but polynomially on the parameters in the $\beta$-vector, and hence also on the dimensional regulator $\epsilon$. The singular locus of the Pfaffian system is the product of its denominators, which may a priori be larger than the principal $A$-determinant. However if it is possible to find a transformation that brings it into an $\epsilon$-factorized form \eqref{eq:DE}, then the singular locus of the Pfaffian system coincides with the principal $A$-determinant. As shown in~\cite{Lee:2014ioa,Lee:2017oca}, this is because the monodromy group has been normalized. Especially in the case when the denominators of the canonical differential equation factorize into linear factors, these factors correspond to the symbol letters and are also the factors of the principal $A$-determinant.

For our purposes the principal $A$-determinant is important since it describes the singularities of a GKZ system $H_\beta(A)$, but before concluding this Section let us also mention its relation to the triangulations of the polytope $\conv(A)=\newt(f)$ and the Chow form of the toric variety $X_A$. The relation between these three objects manifests as the following three polytopes coinciding
\begin{equation}\label{eq: a determinant secondary chow}
    \newt(E_A(f))\simeq\Sigma(A)\simeq\mathrm{Ch}(X_A).
\end{equation}
Here $\Sigma(A)$ denotes the \emph{secondary polytope} of $A$, that is, the polytope that encodes the regular subdivisions of $A$, and whose  vertices particularly correspond to the regular triangulations of $A$. Two vertices of $\Sigma(A)$ are connected by an edge if and only if the two triangulations are related by a modification along a circuit in $A$. The similarity between $\Sigma(A)$ and the exchange graph (more precisely, the cluster polytope) of a \emph{cluster algebra} \cite{Fomin1} is too striking to ignore. Specifically, if $A\subset\RR^2$ is the set of vertices of a convex $n$-gon, then the secondary polytope $\Sigma(A)$ is the $n-$th associahedron, which is indeed isomorphic to the cluster polytope of the $A_{n-3}$ cluster algebra. Since $\newt(E_A)\simeq\Sigma(A)$,  our approach for extracting symbol alphabets from the principal $A$-determinant $E_A(\mathcal{G})$ offers promise for providing a first-principle derivation and extension of the intriguing cluster-algebraic structures observed in a wealth of different Feynman integrals~\cite{Chicherin:2020umh,He:2021eec,He:2022tph}, following similar observations in the context of scattering amplitudes in $\mathcal{N}=4$ super Yang-Mills theory~\cite{Golden:2013xva,Drummond:2017ssj}, see also the recent review~\cite{Papathanasiou:2022lan}.

Finally, the third polytope in eq.~(\ref{eq: a determinant secondary chow}), $\mathrm{Ch}(X_A)$, is the Chow polytope. This is the weight polytope of the Chow form $R_{X_A}$, which describes all the $(n-k-1)-$dimensional projective subspaces in $\PP^{n-1}$ that intersect $X_A$. Here $X_A$ is a toric variety of dimension $k-1$ and degree $d$. Let $\mathcal{B}=\bigoplus\mathcal{B}_m$ be the homogeneous coordinate ring of $\Gr(n-k,n)$, then $R_{X_A}\in\mathcal{B}_d$. Again this seems to point towards cluster algebras since the homogeneous coordinate ring of a Grassmanian comes with a natural cluster algebra structure \cite{Scott}.
%%%%%%%%%%%%%%%%%%%%%%%%%%%%%%%%%%%%%%%%%%%%%%%%%%%%%%%%%%%%%%%%%%%%%%%%%%

\section{One-loop principal $A$-determinants and symbol letters}\label{sec: symbol alphabets}
The goal of this Section is to present a formula to compute the symbol alphabet of one-loop Feynman graphs from their principal $A$-determinant. At this loop order the relevant principal  $A$-determinant is in turn computed via determinant calculations as described in subsection \ref{subsec:OneLoopPricADet} below. The resulting formulas for the symbol letters of generic Feynman integrals, where all masses and momenta squared are nonzero and different from each other, are then given in subsection \ref{subsec:SymbolLettersFormula}. These formulas are then verified by directly constructing the corresponding canonical differential equations in subsection~\ref{sec:DE}, and by comparing with earlier results in the literature, whenever available. Finally, in subsection~\ref{sec: limiting procedure} we provide evidence that the principal $A$-determinant and symbol alphabets of non-generic graphs may be obtained from the generic ones by a limiting process.

\subsection{Matrix representation of one-loop principal $A$-determinants}\label{subsec:OneLoopPricADet}
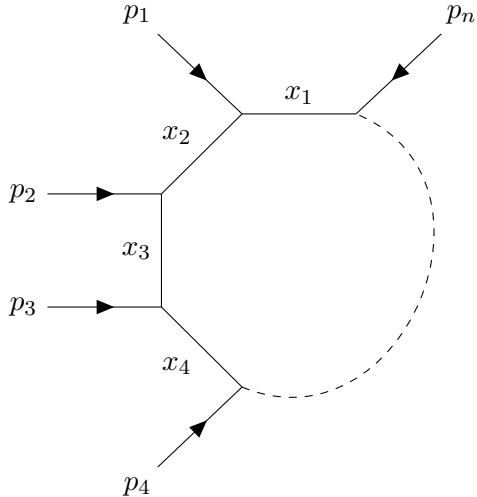
\begin{figure}
    \centering
    \begin{tikzpicture}[baseline=-\the\dimexpr\fontdimen22\textfont2\relax]
    \begin{feynman}
    \vertex (a);
    \vertex [below left = of a] (b);
    \vertex [below  = of b] (c);
    \vertex [below right = of c](d);
    \vertex [right = of a] (e);
    %Vertices for external momenta
    \vertex [above left = of a](p1){\(p_1\)};
    \vertex [left = of b](p2){\(p_2\)};
    \vertex [left = of c](p3){\(p_3\)};
    \vertex [below left = of d](p4){\(p_4\)};
    \vertex [above right = of e](pe){\(p_{{n}}\)};
    \diagram*{
        (a)--[edge label'=\(x_2\)](b)--[edge label'=\(x_3\)](c)--[edge label'=\(x_4\)](d),(a)--[edge label=\(x_1\)](e), (d)--[scalar, half right](e),
        (p1)--[fermion](a),
        (p2)--[fermion](b),
        (p3)--[fermion](c),
        (p4)--[fermion](d),
        (pe)--[fermion](e),
    };
    \end{feynman}
    \end{tikzpicture}
    \caption{Generic one-loop Feynman diagram with ${n}$ external legs.}
    \label{fig:one-loop feynman}
\end{figure}

Let us start by further specializing the mostly general discussion of  Section \ref{sec: Singularities and Principal A-determinant} to one-loop graphs as shown and labeled in Figure~\ref{fig:one-loop feynman}. In this case the number of external legs coincides with the number of internal legs and with the number of vertices, and for simplicity from this point onwards we will denote this number with
\begin{equation}
n=|E|=|V|\,.
\end{equation}

For one-loop graphs $\U$ has degree one and $\F$ has degree two, so the homogenized Lee-Pomeransky polynomial $\G_h=\U x_0+\F$ has degree two. For degree two polynomials the discriminant and hence also the principal $A$-determinant calculation, which we will consider in this Section, can be made very simple.

Let $f$ be a homogeneous polynomial of degree two. The homogeneity means that the vanishing of all partial derivatives implies the vanishing of $f$, i.e.
\begin{equation}
    \frac{\partial f}{\partial x_1}=\cdots=\frac{\partial f}{\partial x_k}=0\Longrightarrow f=0.
\end{equation}
And since $f$ has degree two, all partial derivatives are linear homogeneous functions. The definition of $\nabla_0$ now reduces to requiring that the vanishing locus of the partial derivatives be non-empty with $x$ in the torus $(\CC^*)^k$. 

Writing the zero set of the partial derivatives as
\begin{equation}\label{eq:Jacobian_Coefficent_Mat}
    \begin{pmatrix}
    \frac{\partial f}{\partial x_1}\\ \vdots \\ \frac{\partial f}{\partial x_k}
    \end{pmatrix}=:
    \jac(f)\begin{pmatrix}
    x_1\\ \vdots\\ x_k
    \end{pmatrix}=\mathbf{0},
\end{equation}
where $\jac(f)$ is the coefficient matrix associated to the Jacobian of $f$, we get that
\begin{equation}\label{eq:discrimDet_punc}
    \overline{\mathbf{V}\left(\frac{\partial f}{\partial x_1},\ldots,\frac{\partial f}{\partial x_k}\right)\neq\emptyset\ \mathrm{for}\ x\in\CC^k\setminus\{\mathbf{0}\}}\Longleftrightarrow\ \det(\jac(f))=0.
\end{equation}
Note that this is an equivalence for $x$ in the punctured affine space $\CC^k\setminus\{\mathbf{0}\}$ and not for $x$ in the torus $(\CC^*)^k$, we will return to this important point shortly.

For the Lee-Pomeransky polynomial $\G$ the coefficient matrix $\jac(\G)$ coincides with the \emph{modified Cayley matrix} as defined by Melrose in \cite{Melrose:1965kb}. To begin with, we define the \emph{Cayley matrix} $Y$ to be the ${n}\times {n}$-matrix with elements
\begin{equation}\label{eq:CayleyElements}
    Y_{ii}=2m_i^2\,,\qquad Y_{ij}=m_i^2+m_j^2-p(F_{ij})^2
\end{equation}
where 
\begin{equation}
p(F_{ij})^2 =(p_i+\ldots + p_{j-1})^2\equiv s_{ij-1}=s_{ji-1}\,,    
\end{equation}
is the total momenta flowing through the two-forest $F_{ij}$ obtained from $G$ by removing edges $i$ and $j$. We have also expressed the latter in terms of the Mandelstam invariants in the labelling of Figure~\ref{fig:one-loop feynman}, where cyclicity of the indices modulo $n$ is implied. For concreteness, in what follows we will choose conventions where these $n(n-1)/2$ cyclic Mandelstam invariants do not contain $p_n$. In terms of the coefficient matrix of the Jacobian,  $Y=\jac(\F)$. From the Cayley matrix the \emph{modified Cayley matrix} $
\cY$ is simply obtained by decorating it with a row and column  in the following manner:\footnote{{The modified Cayley matrix also appears naturally in the embedding space formalism~\cite{Dirac:1936fq} as applied to one-loop integrals, see for example~\cite{Abreu:2017ptx}. The compactification of the integration domain, which can be easily carried out in this formalism, also treats type I and II Landau singularities of one-loop integrals on the same footing.}}
\begin{equation}\label{eq:mCayley}
    \cY:=\begin{pmatrix}
        0&1&1&\cdots &1\\
        1&Y_{11}&Y_{12}&\cdots &Y_{1{n}}\\
        1&Y_{12} & Y_{22}&\cdots &Y_{2{n}}\\
        \vdots&\vdots&\vdots &&\vdots \\
        1&Y_{1{n}}&Y_{2{n}} &\cdots&Y_{{n}{n}}
    \end{pmatrix}.
\end{equation}
The determinant of both the Cayley and modified Cayley matrix can be understood as \emph{Gram determinants}, defined in general as
\begin{equation}\label{eq:Gram_general}
 G(k_1,\ldots,k_m;l_1,\ldots,l_m)\equiv \det\left(k_i \cdot l_j\right)=  \det\begin{pmatrix}
        k_1\cdot l_1 & k_1\cdot l_2&\cdots&k_1\cdot l_{{m}}\\
        k_1\cdot l_2 & k_2\cdot l_2 &\cdots&k_2\cdot l_{{m}}\\
        \vdots & \vdots && \vdots\\
        k_1\cdot l_{{m}}&k_2 \cdot l_{{m}}&\cdots&k_{{m}}\cdot l_m 
    \end{pmatrix}\,,
\end{equation}
and further abbreviated as
\begin{equation}\label{eq:Gram_general_kk}
G(k_1,\ldots,k_m)\equiv  G(k_1,\ldots,k_m;k_1,\ldots,k_m)\,,   
\end{equation}
when the two sets of momenta coincide. The determinant of the Cayley matrix is up to a proportionality factor the Gram determinant of the internal momenta restricted to be on their mass shell $q_i^2=m_i^2$, 
\begin{equation}
    \det(Y)=(-2)^{{n}} G(q_1,\ldots q_n)\,.
\end{equation}
Similarly, the determinant of the modified Cayley matrix is proportional to the Gram determinant of any ${n}-1$ of the external momenta, for example
\begin{equation}\label{eq:cYtoG}
    \det(\cY)=-2^{{n}-1}G(p_1,\ldots,p_{n-1})\,,
\end{equation}
namely it is independent of the internal masses. In what follows, we will call this particular determinant the Gram determinant of a Feynman graph/integral, and use the term general Gram determinant otherwise.

The relation between $\det(Y)$, $\det(\cY)$ and general $m\times m$ Gram determinants is helpful since  the latter can be written as certain $(m+2)\times(m+2)$ \emph{Cayley-Menger} determinants according to\small
\begin{equation}
    G(k_1,\ldots,k_m)=\frac{(-1)^{m+1}}{2^m}\det\begin{pmatrix}
        0&1&1&1&\cdots&1 &1\\
        1&0&(k_1-k_2)^2&(k_1-k_3)^2&\cdots&(k_1-k_m)^2 &k_1^2\\
        1&(k_1-k_2)^2&0&(k_2-k_3)^2&\cdots&(k_2-k_m)^2 &k_2^2\\
        1&(k_1-k_3)^2&(k_2-k_3)^2&0&\cdots&(k_3-k_m)^2 &k_3^2\\
        \vdots&\vdots&\vdots&\vdots&&\vdots&\vdots\\
        1&(k_1-k_m)^2&(k_2-k_m)^2&(k_3-k_m)^2&\cdots&0 &k_m^2\\
        1&k_1^2&k_2^2&k_3^2&\cdots&k_m^2&0
    \end{pmatrix}.
\end{equation}\normalsize
Cayley-Menger determinants have the important property that they are irreducible for $n\ge 3$, see \cite{DAndrea}. This means that for a homogeneous polynomial $f$ of degree two a sufficient condition for
\begin{equation}
    \Delta_A(f)=\det(\jac(f))
\end{equation}
is that $\det(\jac(f))$ can be understood as a Cayley-Menger determinant with $n\ge 3$. The irreducibility of $\det(\jac)$ removes the need for  the distinction between the algebraic torus and punctured affine space meaning that the determinant is equal to the discriminant; that is in this case we can replace $x\in\CC^k\setminus\{\mathbf{0}\}$ with $x\in (\CC^*)^k$ in \eqref{eq:discrimDet_punc}. An example illustrating how this identification can fail when the determinant is reducible is given in Example \ref{ex:discriminant_not_determinant}.

%\mcomment{I think Example 5.2 should be here.}

For massive one-loop graphs with generic kinematics the result in \cite{DAndrea} means that the Cayley determinant for ${n}\ge 3$ and the Gram determinant for ${n}\ge 4$ are all irreducible and thus coincide with the expected discriminant. For the few cases not covered by the theorem, explicit calculations confirm that the discriminant and determinant coincide. 

We will not only be interested in the determinant of the modified Cayley matrix but also its minors. Let $1\le k<{n}+1$ be an integer, then the general minor of order ${n}+1-k$ is denoted as
\begin{equation}\label{eq:MinorNotation}
    \cY\begin{bmatrix}
        i_1&i_2&\cdots&i_k\\
        j_1&j_2&\cdots&j_k
    \end{bmatrix},\qquad 1\le i_1<i_2<\cdots i_k\le {n},\ 1\le j_1<j_2<\cdots j_k\le {n}
\end{equation}
where $I=\{i_1,i_2,\cdots,i_k\}$ and $J=\{j_1,j_2,\ldots,j_k\}$ denotes the rows resp, columns removed from $\cY$. If both $I$ and $J$ are empty we recover the full determinant $\det(\cY)=\cY\begin{bmatrix}
    \cdot\\ \cdot
\end{bmatrix}$. It will sometimes also be convenient to introduce the complementary notation where we index the rows and columns kept in the minor, this is signified by a parenthesis $(\cdot)$ instead of square brackets $[\cdot]$. Let $\mathcal{E}=\{1,2,\ldots,{n}+1\}$, then
\begin{equation*}
    \cY\begin{bmatrix}
        I\\J
    \end{bmatrix}=\cY\begin{pmatrix}
        \mathcal{E}\setminus I\\ 
        \mathcal{E}\setminus J
    \end{pmatrix}.
\end{equation*}

The identification between discriminants and subgraphs in Section \ref{sec: Singularities and Principal A-determinant} can now be carried out for determinants as well:
\begin{itemize}
    \item $\Delta_A(\G)=\det(\cY)$,
    \item $\Delta_{A\cap x_i=0}(\G)=\cY\begin{bmatrix}
        i+1\\i+1
    \end{bmatrix}$,
    \item $\Delta_{A\cap\newt(\F)}(\F)=\det(Y)=\cY\begin{bmatrix}
        1\\1
    \end{bmatrix}$,
    \item $\Delta_{A\cap\newt(\F)\cap x_i=0}=Y\begin{bmatrix}
        i\\i
    \end{bmatrix}=\cY\begin{bmatrix}
        1&i+1\\
        1&i+1
    \end{bmatrix}$,
\end{itemize}
For one-loop graphs with massive and generic kinematics this means that all non-trivial discriminants can be expressed as determinants, making the calculation much easier. Especially, the reduced principal $A$-determinant now has a simple expression as a product of minors of $\cY$:
\begin{equation}\label{eq:1LoopEA}
\widetilde{E_A}(\G)=\cY\begin{bmatrix}
    \cdot\\ \cdot
\end{bmatrix} \prod_{i=1}^{{n}+1}\mCm\left[\begin{array}{c}
i\\
i
\end{array}\right]
\ldots
\prod_{i_{{n}-1}>\ldots>i_1=1}^{{n}+1}\mCm\left[\begin{array}{c}
i_1\ldots i_{{n}-1}\\
i_1\ldots i_{{n}-1}
\end{array}\right]
\prod_{i=2}^{{n}+1}\mCm_{ii}\,.
\end{equation}
In other words, $\widetilde{E_A}$ is equal to the product of all diagonal $k$-dimensional minors of $\mCm$ with $k=1,\ldots,n+1$ (the largest one corresponding to the determinant of the entire matrix), except for the $\mCm_{11}$ element, which is zero. Given that minors of $\mCm$ where the first row and column has (not) been removed correspond to the Cayley (Gram) determinant of the Feynman graph, or its subgraphs where certain edges have been contracted, the above formula also has the following interpretation: The principal $A$-determinant of a generic 1-loop $n$-point graph is the product of the Cayley and Gram determinants of the graph and all of its subgraphs. From these considerations, we may also easily derive that the total number of factors in eq.\eqref{eq:1LoopEA}, i.e. the total number of kinematic-dependent Cayley  and Gram determinants of the Feynman graph and all of its subgraphs, is
 \begin{equation}
2^{n+1}-n-2\,,     
 \end{equation}
namely 1, 4, 11, 26, 57 and 120 factors for $n=1,\ldots,6$. In this counting we exclude not only the element $\mCm_{11}=0$ but also all two-dimensional minors including the first row and column, which always yield $-1$.

\subsection{One-loop symbol alphabets}\label{subsec:SymbolLettersFormula}

In this subsection we will show how to calculate the symbol alphabet of a one-loop graph by appropriately re-factorizing its principal $A$-determinant~\eqref{eq:1LoopEA}. Note that we are going to distinguish the two cases for the sum of loop integration dimension and number of external legs $D_0+{n}$ being odd, respectively, even.

\paragraph{Jacobi identities.} As has been seen in the example presented in the introduction, individual factors in the natural factorization of the principal $A$-determinant gives us all the rational letters for a Feynman graph. It is well known, however, that letters containing square roots quickly become unavoidable, even at one-loop. In order to construct the letters containing square roots we will use Jacobi determinant identities to re-factorize the factors of the principal $A$-determinant~\eqref{eq:1LoopEA}, namely of minors of the modified Cayley matrix $\cY$, in pairs. Jacobi determinant identities have played an important role in many areas of mathematics and physics, including the computation of volumes of spherical simplices~\cite{aomoto_1977}, as well as the solution theory of integrable systems~\cite{hirota_2004}.

There will be two types of identities that we need; identities containing the full determinant $\cY\left[\begin{smallmatrix} 
    \cdot\\ \cdot
\end{smallmatrix}\right]$ (of a given graph and of its subgraphs) and those containing the determinant of the Cayley sub-matrix $\cY\left[\begin{smallmatrix}
    1\\ 1
\end{smallmatrix}\right]$. The identities containing the full determinant are
\begin{equation}\label{eq:JacobiIdOdd}
\begin{array}{l}
    \cY\begin{bmatrix}
        \cdot\\ \cdot
    \end{bmatrix}\cY\begin{bmatrix}
        1&i\\1&i
    \end{bmatrix}=\cY\begin{bmatrix}
        i\\i
    \end{bmatrix}\cY\begin{bmatrix}
        1\\1
    \end{bmatrix}-\cY\begin{bmatrix}
        i\\1
    \end{bmatrix}^2\,,\\
     \cY\begin{bmatrix}
        \cdot\\ \cdot
    \end{bmatrix}\cY\begin{bmatrix}
        i&j\\i&j
    \end{bmatrix}=\cY\begin{bmatrix}
        i\\i
    \end{bmatrix}\cY\begin{bmatrix}
        j\\j
    \end{bmatrix}-\cY\begin{bmatrix}
        i\\j
    \end{bmatrix}^2\,,\quad i\ge 2\,,
    \end{array}
\end{equation}
where in the first line we have simply separated out the $i=1$ case of the second line for later convenience. As we will get back to shortly, these identities are relevant in the case when $D_0+{n}$ is odd as it is the modified Cayley determinant that contributes to the leading singularity, { see also eq.~\eqref{eq:leading-sing} and the discussion around it in the next section.}

The identities containing the determinant of the Cayley submatrix are
\begin{equation}\label{eq:JacobiIdEven}
    \begin{array}{l}
        \cY\begin{bmatrix}
        \cdot\\ \cdot
    \end{bmatrix}\cY\begin{bmatrix}
        1&i\\1&i
    \end{bmatrix}=\cY\begin{bmatrix}
        i\\i
    \end{bmatrix}\cY\begin{bmatrix}
        1\\1
    \end{bmatrix}-\cY\begin{bmatrix}
        i\\1
    \end{bmatrix}^2\,,\\
     \cY\begin{bmatrix}
        1\\ 1
    \end{bmatrix}\cY\begin{bmatrix}
        1&i&j\\1&i&j
    \end{bmatrix}=\cY\begin{bmatrix}
        1&j\\1&j
    \end{bmatrix}\cY\begin{bmatrix}
        1&i\\1&i
    \end{bmatrix}-\cY\begin{bmatrix}
        1&j\\1&i
    \end{bmatrix}^2\,,
    \end{array}
\end{equation}
and these will in turn be relevant when $D_0+{n}$ is even, as it is the Cayley determinant that yields the leading singularity.   Note that the first identity in each case is the same.
\paragraph{Symbol letters.} 
%\label{subsubsec:symbol-letters}
%\textbf{Symbol letters.} 
{ We now introduce a procedure for obtaining not only the rational but also the square-root letters of one-loop Feynman integrals as follows: We assume that the latter are produced  by applying Jacobi determinant identities of the form
\begin{equation}\label{eq:JacobiSqRt}
    p\cdot q=f^2-g=(f-\sqrt{g})(f+\sqrt{g})\,,
\end{equation}
where
\begin{enumerate}
    \item $p$ and $q$ are both factors of the principal $A$-determinant, i.e.\ rational letters given by symmetric minors of the modified Cayley matrix.
    \item The square-root letters ${f\pm\sqrt{g}}$ thus obtained contain the leading singularity of the Feynman integral considered in their second term.
\end{enumerate}
 This procedure is motivated by the interpretation of one-loop integrals as volumes of spherical simplices~\cite{Davydychev:1998fk}, and by the role Jacobi identities have played in their computation. Particularly the second assumption adopts a pattern that has been observed not only in one-, but also many two-loop computations. In the next subsection, we will validate the correctness of these assumptions by explicitly constructing the canonical differential equations of one-loop integrals, as well as compare with the existing results in the literature. While the precise identities and assumptions may differ beyond one loop, we expect that a similar re-factorization methodology should still apply.

Given that the rational letters at the very left of eq.~\eqref{eq:JacobiSqRt} are already contained in the rational alphabet, the genuinely new square-root letter that will arise from the procedure we have described will be the ratio of the factors at the very right of the same equation,
\begin{equation}
    \frac{f-\sqrt{g}}{f+\sqrt{g}}\,.
\end{equation}
Furthermore, the subset of Jacobi identities chosen by the assumptions we have stated above will be precisely eqs.~\eqref{eq:JacobiIdOdd} and \eqref{eq:JacobiIdEven} for $D_0+{n}$ odd and even, respectively. The first and second line of these equations then yields  $n$ letters of the first type,}
\begin{equation}\label{eq: letter one-contracted}
    W_{1,\ldots,(i-1),\ldots,{n}}=
    \begin{cases}
        \dfrac{\cY\begin{bmatrix}
        i\\1
    \end{bmatrix}-\sqrt{-\cY\begin{bmatrix}
            \cdot\\ \cdot
        \end{bmatrix}\cY\begin{bmatrix}
        1&i\\1&i
    \end{bmatrix}}}{\cY\begin{bmatrix}
        i\\1
    \end{bmatrix}+\sqrt{-\cY\begin{bmatrix}
            \cdot\\ \cdot
        \end{bmatrix}\cY\begin{bmatrix}
       1&i\\1&i
    \end{bmatrix}}},\qquad &D_0+{n}\ \mathrm{odd,}\\{}\\
    \dfrac{\cY\begin{bmatrix}
        i\\1
    \end{bmatrix}-\sqrt{\cY\begin{bmatrix}
        i\\i
    \end{bmatrix}\cY\begin{bmatrix}
        1\\1
    \end{bmatrix}}}{\cY\begin{bmatrix}
        i\\1
    \end{bmatrix}+\sqrt{\cY\begin{bmatrix}
        i\\i
    \end{bmatrix}\cY\begin{bmatrix}
        1\\1
    \end{bmatrix}}},\qquad &D_0+{n}\ \mathrm{even}.
    \end{cases}
\end{equation}
and $n(n-1)/2$ letters of the second type,
\begin{equation}\label{eq: letter two-contracted}
    W_{1,\ldots,(i-1),\ldots,(j-1),\ldots,{n}}=
    \begin{cases}
    \dfrac{\cY\begin{bmatrix}
        i\\j
    \end{bmatrix}-\sqrt{-\cY\begin{bmatrix}
        \cdot\\ \cdot
    \end{bmatrix}\cY\begin{bmatrix}
        i&j\\i&j
    \end{bmatrix}}}{\cY\begin{bmatrix}
        i\\j
    \end{bmatrix}+\sqrt{-\cY\begin{bmatrix}
        \cdot\\ \cdot
    \end{bmatrix}\cY\begin{bmatrix}
        i&j\\i&j
    \end{bmatrix}}},\qquad &D_0+{n}\ \mathrm{odd},\\{}\\
    \dfrac{\cY\begin{bmatrix}
        1&j\\1&i
    \end{bmatrix}-\sqrt{-\cY\begin{bmatrix}
        1\\1
    \end{bmatrix}\cY\begin{bmatrix}
        1&i&j\\1&i&j
    \end{bmatrix}}}{\cY\begin{bmatrix}
        1&j\\1&i
    \end{bmatrix}+\sqrt{-\cY\begin{bmatrix}
        1\\1
    \end{bmatrix}\cY\begin{bmatrix}
        1&i&j\\1&i&j
    \end{bmatrix}}},\qquad &D_0+{n}\ \mathrm{even},
    \end{cases}
\end{equation}
respectively, where we remind the reader that $\cY$ denotes the modified Cayley matrix~\eqref{eq:mCayley} of the $n$-point one-loop integral, and that minors of the latter, obtained by removing some of its rows and columns, have been defined in eq.~\eqref{eq:MinorNotation}.
 
Finally, the full determinant and the minor $\cY\begin{bmatrix}
    1\\1
\end{bmatrix}$ should also in principle appear as rational letters of the  $n$-point graph. However, as we will come back to in the next subsection, these in fact always come as a rational combination, 
\begin{equation}\label{eq: rational letter}
        W_{1,2,\ldots,{n}}=\dfrac{\cY\begin{bmatrix}
        \cdot\\ \cdot
        \end{bmatrix}}{\cY\begin{bmatrix}
            1\\1
        \end{bmatrix}}\,.
\end{equation}
\begin{remark}
    It is interesting to note that this rational function has an intrinsic meaning in the GKZ approach to Feynman integrals: It is the value the Lee-Pomeransky polynomial attains on its critical point. If we solve $\partial\G/\partial x_i=0$ for all $i=1,\ldots,n$ and evaluate $\G$ at this unique point we get this fraction up to a numerical factor. This is a special case of a general property of the principal $A$-determinant \cite[Theorem 1.17, Chapter 10]{gelfand2008discriminants}.   
\end{remark}

So far, we have obtained $1+n(n+1)/2$ letters of the $n$-point graph for general $n$. To obtain the remaining letters, we apply the above expressions to each of the subgraphs of the $n$-point graph, obtained by contracting some of its edges. That is, we replace $\cY$ with the modified Cayley matrix of the subgraph in question on the right-hand side of eqs.\eqref{eq: letter one-contracted}-\eqref{eq: rational letter}, and we only keep the indices labeling its uncontracted edges on the left-hand side\footnote{Our labeling conventions for the letters also reveal dihedral relations among them, e.g. $W_{1,2}$ and $W_{2,3}$ are related by a cyclic shift. Note however that due to our choice of conventions, these relations may come with additional minus signs or inversions.}.  The correct $\cY$-matrix for a subgraph is obtained from the original one by removing the rows and columns corresponding to the contracted edges. This process terminates once the letters for the tadpoles have been calculated. 

In the ancillary \texttt{Mathematica} file accompanying the version of this paper on the \texttt{arXiv}, we provide code for generating the complete $n$-point alphabet in principle for any $n$. As a benchmark, the runtime for $n=6$ is at the order of a minute on a laptop computer.
\begin{remark}
    In the $n=1$ case, or equivalently the tadpole graph, the modified Cayley matrix is too small to allow for any Jacobi identities, and hence only the rational letter of type \eqref{eq: rational letter} is present. Similarly, for $n=2$ or the bubble graph there are no letters of the type $W_{(i),(j)}$. 
\end{remark}
\begin{remark}\label{remark: 2 kallen triangle}
    For $n=3$ or the triangle graph and even loop integration dimension $D_0$, one would expect $\binom{3}{2}=3$ letters of the type $W_{(i),(j),k}$:
    \begin{equation}
    \frac{a-b+c-\sqrt{\lambda}}{a-b+c+\sqrt{\lambda}},\qquad \frac{a+b-c-\sqrt{\lambda}}{a+b-c+\sqrt{\lambda}},\qquad \frac{a-b-c-\sqrt{\lambda}}{a-b-c+\sqrt{\lambda}}
\end{equation}
    where $\lambda:=\lambda(a,b,c)$ is the fully symmetric Källén function\,, 
\be\label{eq:Kaellen}
\lambda(a,b,c)=a^2+b^2+c^2-2ab-2bc-2ac\,.
\ee
However, these have the multiplicative relation
    \begin{equation}
    \frac{a-b+c-\sqrt{\lambda}}{a-b+c+\sqrt{\lambda}}\cdot \frac{a+b-c-\sqrt{\lambda}}{a+b-c+\sqrt{\lambda}}=\frac{a-b-c-\sqrt{\lambda}}{a-b-c+\sqrt{\lambda}}\,.
\end{equation}
    Hence one need only include two out of the three so as to obtain a multiplicatively independent set, and here and in the attached ancillary file we will in particular choose them to be $W_{(i),(j),k}$ and  $W_{i,(j),(k)}$. Note that the special cases of letters we have discussed pertain also  to tadpole, bubble (and for $D_0$ even triangle) subgraphs of $n\ge4$ graphs, hence the use of generic letter indices. 
\end{remark}
From the above formulas, we can also obtain a closed formula for the total number of letters of an $n$-point graph: the latter has $\binom{n}{m}$ $m$-point subgraphs, $m=1,\ldots n$, and each of them yields
\begin{equation}
    \frac{m(m+1)}{2}+\delta_{m\ge 4}+\delta_{m,3}\delta_{D_0,\text{odd}}\,,
\end{equation}
letters, where $\delta$ denotes the Kronecker delta function, with some abuse of notation to describe the cases where it equals 1 when its index is greater than a given integer value, or odd (and zero otherwise). Therefore in total the number of letters is
\begin{equation}\label{eq:LetterCounts}
|W|= \begin{cases}
2^{n-3}\left(n^2+3n+8\right)-\frac{1}{6}\left(n^3+5n+6\right)\,,    \qquad &D_0\ \mathrm{even},\\{}\\
2^{n-3}\left(n^2+3n+8\right)-\frac{1}{2}\left(n^2+n+2\right)\,,\qquad &D_0\ \mathrm{odd},,
    \end{cases}   
\end{equation}
that is, $|W|=1,5,18,57,166,454,1184$ and $|W|=1,5,19,61, 176,474,1219$ for $n=1,\ldots,7$ when $D_0$ even and odd, respectively. These counts correspond to the total number of multiplicatively independent letters of the generic $n$-point 1-loop integral shown in Figure~\ref{fig:one-loop feynman}, provided that $n\le d+1$, e.g.~$n\le 5$ when the external momenta live in $d=4$ dimensions. This restriction ensures that all Mandelstam invariants appearing the letters may be treated as independent variables, given that no more than $d$ vectors can be linearly independent in $d$ dimensions.

For $n> d+1$, after choosing the first $d$ momenta of the $n$-point integral as our basis in the space of external kinematics, expressing the remaining momenta in terms of them, and dotting these linear relations with the basis vectors, we may express them as polynomial relations between the Mandelstam invariants,
\begin{equation}\label{eq:GramConstraintsLargen}
G(p_1,\ldots,p_d,p_i)=0\,,\quad i=d+1,\ldots,n-1\,,
\end{equation}
where $G$ is the Gram determinant, defined in eq.~\eqref{eq:Gram_general}. In subsection~\ref{sec: limiting procedure} we will provide evidence that the principal $A$-determinant and hence also the alphabet of any $n$-point 1-loop graph, obeying additional restrictions such as~\eqref{eq:GramConstraintsLargen}, or such that any of the masses and Mandelstam invariants become equal to each other or vanish, may be obtained as limits of the generic cases, eqs.~\eqref{eq:1LoopEA} and \eqref{eq: letter one-contracted}-\eqref{eq: rational letter}, respectively. In these cases the limits will introduce additional multiplicative dependence among the letters, hence our generic formulas will provide a spanning set for the alphabet, and the counts~\eqref{eq:GramConstraintsLargen} will correspond to upper bounds. As we will also see in the examples presented in Section~\ref{sec: examples}, a basis within this spanning set may be found immediately in any kinematic parametrization that rationalizes all resulting letters, or e.g. with the help of \texttt{SymBuild}~\cite{Mitev:2018kie} even in the presence of square roots.

%%%%%%%%%%%%%%%%%%%%%%%%%%%%%%%%%%%%%%%%%%%%%%%%%%%%%%%%%%%%%%%%%%%%%%%%%%%%%%%%%%%%%%%%%%%%%%%%%%
\subsection{Verification through differential equations and comparison with literature}\label{sec:DE}
In this Section we show how the letters actually appear in the canonical differential equations. For our basis $\vec{g}$ of pure master integrals as defined in eq.~\eqref{eq:DE}, we use the fact that any integral in $D-2$ dimensions of loop momenta may be expressed as a linear combination of integrals in $D$ dimensions and vice-versa, with the help of dimensional recurrence relations \cite{Tarasov:1996br,Lee:2009dh}. As these relations are merely a change of basis, this implies that different cases of spaces of master integrals are simply distinguished by (the integer part of) $D$ being even or odd, and that within each case we may choose our basis $\vec{g}$ to consist of integrals with different $D$. To be more concrete, we need to introduce some further notation: We write the integrals of \eqref{eq:I-def} in $D=D_0-2\epsilon$ dimensions and for $a_i=1$ as $\mathcal{I}_E^{(D_0)}$, where the set $E$ indicates the edges of the corresponding graph. E.g.~$\mathcal{I}_{134}^{(2)}$ denotes the triangle integral in $D=2-2\epsilon$ dimensions that is obtained when the second propagator of the box integral $\mathcal{I}_{1234}^{(2)}$ is removed.

We then take the basis $\vec{g}$ to consist of the following canonical integrals:
\begin{equation}
\label{eq:can-int-notation}
 \begin{aligned}
     \mathcal{J}_{i_1\ldots i_k}=
     \begin{cases}
      \dfrac{\epsilon^{\left\lfloor\frac{k}{2}\right\rfloor}\mathcal{I}^{(k)}_{i_1\ldots i_k}}{j_{i_1\ldots i_k}} & \text{for } k+D_0 \text{ even},\\
      & \\
      \dfrac{\epsilon^{\left\lfloor\frac{k+1}{2}\right\rfloor}\mathcal{I}^{(k+1)}_{i_1\ldots i_k}}{j_{i_1\ldots i_k}} & \text{for } k+D_0 \text{ odd}\,,
     \end{cases}
 \end{aligned}
\end{equation}
where the leading singularities are
\begin{equation}
\label{eq:leading-sing}
 j_{i_1\cdots i_k}=\begin{cases}
  2^{-\frac{k}{2}+1}\Bigg[(-1)^{\left\lfloor\frac{k}{2}\right\rfloor}\cY\begin{pmatrix}
        i_1+1& i_2+1&\cdots &i_k+1\\
        i_1+1& i_2+1&\cdots &i_k+1
    \end{pmatrix}\Bigg]^{-1/2},& \text{for } k+D_0 \text{ even}\,,\\
  2^{-\frac{k+1}{2}+1}\Bigg[(-1)^{\left\lfloor\frac{k+1}{2}\right\rfloor}\cY\begin{pmatrix}
        1&i_1+1& i_2+1&\cdots &i_k+1\\
        1&i_1+1& i_2+1&\cdots &i_k+1
    \end{pmatrix}\Bigg]^{-1/2},& \text{for } k+D_0 \text{ odd}\,.
 \end{cases}
\end{equation}
 These integrals have been observed to be pure integrals~\cite{Abreu:2017mtm,Chen:2022fyw} for $D_0$ even, see also~\cite{Spradlin:2011wp,Arkani-Hamed:2017ahv,Bourjaily:2019exo} for earlier results with $\epsilon=0$. Note that both the overall sign of the above equation, as well as the choice of branch for the square root, are a matter of choice of convention. As is discussed in e.g.~\cite{Bourjaily:2019exo}, the former stems from the fact that as multidimensional residues, leading singularities are intrinsically dependent on the orientation of the integration contour; whereas the latter stems from the fact that while scalar Feynman integrals are positive definite in the Euclidean region, their pure counterparts need not be. In eq.~\eqref{eq:leading-sing}, we have fixed this freedom, together with the overall kinematic independent normalization that we are also free to choose, such that the differential equations take a convenient form. This choice of pure basis also specifies how square roots of products of modified Cayley (sub)determinants should be replaced by the product of the square roots of the (sub)determinants in question in our basis of letters~\eqref{eq: letter one-contracted}-\eqref{eq: rational letter}.

\iffalse 
Note that the overall kinematic independent normalization is a choice and we picked one that puts the differential equations into a particularly nice form. On the other hand, the choice of sign under the square root will be reflected in how the two factors under the square root appearing in the letters of Section \ref{subsubsec:symbol-letters} are split into two separate square roots.
\fi

{ By explicit computation up to $n=10$, we observe that} the differential equations for the canonical integrals are then as follows: for even ${n}+D_0$ we have\footnote{If we are interested in $D=D_0-2\epsilon$, with $D_0$ being an odd integer, the case of even and odd ${n}$ is exchanged compared to the even-dimensional case. This can easily be seen in Baikov representation~\cite{Baikov:1996iu} where the integrals on the maximal cut are roughly equal to $I\sim G^{({n}-D)/2}C^{(D-1-{n})/2}$, where $G$ is the Gram determinant, and $C$ is the Cayley determinant which is obtained from setting all integration variables to zero in the Baikov polynomial~\cite{Flieger:2022xyq}.}
\begin{equation}\label{eq: CDE D0+E even}
    \begin{aligned}
        d\mathcal{J}_{1\ldots{n}}=&\epsilon\ d\log W_{1\ldots{n}}\ \mathcal{J}_{1\ldots{n}}\\
        &+\epsilon\sum_{1\leq i\leq {n}}(-1)^{i+\left\lfloor\frac{{n}}{2}\right\rfloor}d\log W_{1\ldots(i)\ldots{n}}\ \mathcal{J}_{1\ldots\widehat{i}\ldots{n}}\\
    &+\epsilon\sum_{1\leq i<j\leq {n}}(-1)^{i+j+\left\lfloor\frac{{n}}{2}\right\rfloor}d\log W_{1\ldots(i)\ldots(j)\ldots{n}}\ \mathcal{J}_{1\ldots\widehat{i}\ldots\widehat{j}\ldots{n}},
    \end{aligned}
\end{equation}
and for odd ${n}+D_0$\,,
\begin{equation}\label{eq: CDE D0+E odd}
    \begin{aligned}
        d\mathcal{J}_{1\ldots{n}}=&\epsilon\ d\log W_{1\ldots{n}}\ \mathcal{J}_{1\ldots{n}}\\
        &+\epsilon\sum_{1\leq i\leq {n}}(-1)^{i+\left\lfloor\frac{{n}+1}{2}\right\rfloor}d\log W_{1\ldots(i)\ldots{n}}\ \mathcal{J}_{1\ldots\widehat{i}\ldots{n}}\\
    &+\epsilon\sum_{1\leq i<j\leq {n}}(-1)^{i+j+\left\lfloor\frac{{n}+1}{2}\right\rfloor}d\log W_{1\ldots(i)\ldots(j)\ldots{n}}\ \mathcal{J}_{1\ldots\widehat{i}\ldots\widehat{j}\ldots{n}},
    \end{aligned}
\end{equation}
where the hat indicates that the index is omitted, and $\lfloor x\rfloor$ is the floor function. Note that we chose to present our formulas in such a way that the latter is irrelevant for to case of even $D_0$. The matrix $\widetilde{M}$ encoding the above differential equations according to eq.~\eqref{eq:can-ansatz}, together with our choice of basis integrals~\eqref{eq:can-int-notation} and associated leading singularities~\eqref{eq:leading-sing}, may also be easily generated in principle for any $n$ with the code provided in the attached ancillary file.

As mentioned in the introduction, the knowledge of the letters allows us to derive the differential equations without the need of any analytic computation. The formulas in \eqref{eq: CDE D0+E even} and \eqref{eq: CDE D0+E odd} are based on following this procedure up to $n=10$.\footnote{In practice, to derive the differential equations, one needs to use dimensional recurrence relations to translate all integrals of different dimensions into a common dimension. We stress however, that this choice does not affect the result for the differential equations.} This confirms our prediction for the alphabet of one-loop integrals up to this number of external legs. The IBP reduction was done with a combination of FIRE6 \cite{Smirnov:2019qkx} and LiteRed \cite{Lee:2013mka} by choosing values in a finite field for the $v_i$.

For generic kinematics and masses, it is also easy to determine the leading boundary vector $\vec{g}^{(0)}$ in eq.~\eqref{eq:symbol}: The integrals in \eqref{eq:can-int-notation}  are normalized by powers of $\epsilon$ to be of uniform transcendental weight zero. However, it is easy to see that all integrals except $\mathcal{I}_1^{(2)}$ are finite in integer dimensions. Therefore, only the tadpoles have vanishing weight zero contribution.\footnote{We thank the referee for pointing this out.} In summary, we find with our normalization
\begin{equation}
 \vec{g}^{(0)}=(\underbrace{1,\ldots,1}_{{n}},\underbrace{0,\ldots,0}_{2^{{n}}-1-{n}})^T.
\end{equation}
This in principle allows us to compute the symbol at any order in $\epsilon$ from eq.~\eqref{eq:symbol}.

%Using e.g.~FIESTA \cite{Smirnov:2021rhf}, we also find the leading boundary vector
%\begin{equation}
% \vec{g}^{(0)}=(\underbrace{1,\ldots,1}_{{n}},\underbrace{0,\ldots,0}_{2^{{n}}-1-{n}})^T,
%\end{equation}
%up to ${n}=10${ and for generic kinematics and masses}. This in principle allows us to compute the symbol at any order in $\epsilon$ from eq.~\eqref{eq:symbol}.

Finally, let us compare our findings with earlier results in the literature. For $D_0$ even, { explicit expressions for} the canonical differential equations and symbol alphabets of finite $n$-point one-loop graphs were first derived in~\cite{Abreu:2017mtm}, based on the diagrammatic coaction~\cite{Abreu:2017enx}. The latter decomposes  { any} one-loop Feynman integral into simpler building blocks, mirroring the coaction of the multiple polylogarithmic functions, that these integrals evaluate to, but conjecturally holds to all orders in the dimensional regularization parameter $\epsilon$. In this decomposition also cut integrals appear, where some of the propagators have been placed on their mass shell, and very interestingly it was found that these are restricted to a small subset where all, or all but one, or all but two propagators have been cut. 

Based on the Baikov representation of Feynman integrals, { a similar analysis of canonical differential equations and symbol alphabets, also working out the divergent cases in more detail, was carried out in~\cite{Chen:2022fyw}}. The later and original analyses agree on the form of the letters associated to the maximal cut, and in order to avoid redundant letters coming from individual determinant factors of the principal $A$-determinant, for our rational letters~\eqref{eq: rational letter} we have chosen those ratios that coincide with them (up to immaterial constant normalization factors). For the square-root letters, we will compare with~\cite{Chen:2022fyw}, as the apparent presence of up to five different square roots in some of the letters of~\cite{Abreu:2017mtm} renders this comparison more complicated. In the former reference, and in the orientation where the uncut propagators are those with momenta ($p_{n-1}$ and) $p_{n}$ for the letters associated to the (next-to-)next-to-maximal cut, their building blocks are the following general Gram determinants as defined in eqs.~\eqref{eq:Gram_general}-\eqref{eq:Gram_general_kk},
\be
\begin{aligned}
\mathcal{K}_n&\equiv G(p_1,\ldots,p_{n-1})\,,\quad \tilde G_n\equiv G(l,p_1\ldots,p_{n-1})\,,\\
\tilde B_n&\equiv G(l,p_1,\ldots,p_{n-2};p_{n-1},p_1,\ldots,p_{n-2})\,,\\C_n&\equiv G(p_1,\ldots p_{n-2};p_1,\ldots,p_{n-3},p_{n-2}+p_{n-1})\,,\\
D_n&\equiv G(l,p_1,\ldots p_{n-2};l,p_1,\ldots p_{n-3},p_{n-2}+p_{n-1})\,,
\end{aligned}
\ee
where the loop momentum
\be\label{eq:CMY_BuildingBlocks}
\begin{aligned}
l^2&=    m_1^2\,,\\
l\cdot p_i&=\frac{m_{i+1}^2-p_i^2-m_i^2}{2}-\sum_{j=1}^{i-1}p_i\cdot p_j\,,
\end{aligned}
\ee
is evaluated as a function of the Baikov variables when the latter are set to zero. As a consequence of~\eqref{eq:cYtoG},
\be
\cY\begin{bmatrix}
        \cdot\\ \cdot
    \end{bmatrix}=-2^{n-1}\mathcal{K}_n\,,
\ee
whereas we find that the remaining determinants in~\eqref{eq:CMY_BuildingBlocks} are related to minors of our modified Cayley matrix $\cY$ as follows,
\begin{equation}
\begin{aligned}
    \cY\begin{bmatrix}
        1\\ 1
    \end{bmatrix}&=\det(Y)=2^n \tilde G_n\,,
&    \cY\begin{bmatrix}
        1\\ n+1
    \end{bmatrix}&=-(-2)^{n-1} \tilde B_n\,,\\
    \cY\begin{bmatrix}
        n\\ n+1
    \end{bmatrix}&=-2^{n-2} \tilde C_n\,,&    \cY\begin{bmatrix}
        1&n\\1&n+1
    \end{bmatrix}&=2^{n-1} \tilde D_n\,.\\
    \end{aligned}
\end{equation}
From these identifications, it follows straightforwardly that up to immaterial overall signs and inversions, for $D_0$ even the next-to-maximal and next-to-next-to-maximal cut letters in question correspond to our $W_{1,\ldots,n-1,\ldots,(n)}$ and $W_{1,\ldots n-2,(n-1),\ldots,(n)}$, respectively, and similarly the differential equations~\eqref{eq: CDE D0+E even}-\eqref{eq: CDE D0+E odd} agree with those of~\cite{Abreu:2017mtm,Chen:2022fyw}. To the best of our knowledge, the odd $D_0$ case has not appeared before. In any case, we find it pleasing that our formulas are expressed in terms of a single quantity associated to each one-loop graph, its modified Cayley matrix, making it easy to keep track of both its Landau singularities and contributions of subgraphs from different minors of the matrix.
%%%%%%%%%%%%%%%%%%%%%%%%%%%%%%%%%%%%%%%%%%%%%%%%%%%%%%%%%%%%%%%%%%%%%%%%%%%%%%%%%%%%%%%%%%%%%%%%%
\subsection{Limits of principal $A$-determinants and alphabets}\label{sec: limiting procedure}

So far, we have only considered the generic one-loop $n$-point Feynman integrals shown in Figure~\ref{fig:one-loop feynman}, where all $m_i^2,p^2_j$ are nonvanishing and different from each other. In this subsection, we will provide strong evidence that the principal $A$-determinant and symbol alphabet of any one-loop integral, i.e. also including configurations where different scales are equal to each other or set to zero, may be obtained as limits of the generic ones, eqs.~\eqref{eq:1LoopEA} and~\eqref{eq: letter one-contracted}-\eqref{eq: rational letter}, respectively.

To this end, we will first prove that the principal $A$-determinant $\widetilde{E_A}$ has a well-defined limit when any $m_i^2,p^2_j\to 0$,  namely that this limit is unique regardless of the order or relative rate with which we send these parameters to zero. As is argued in~\cite{Eden:1966dnq}, see also~\cite{Klausen:2021yrt}, this limit of $\widetilde{E_A}$ is defined by removing any of its factors that vanishes as its parameters approach their prescribed values. This reflects the fact that Feynman integrals converge in the Euclidean region for certain choices of propagator powers and dimension of loop integration, and hence they cannot be singular for all values of their kinematic parameters. In other words, assuming a single parameter $x$ of $\widetilde{E_A}$ takes the limiting value $a$, the  limit of the former is
defined as
\begin{equation}\label{eq:EA_limit}
\lim_{x\to a}\widetilde{E_A}=\left.\frac{\partial^l \widetilde{E_A}}{\partial x^l}\right|_{x=a}\ne 0\,,\,\,\text{with }\,\,\left.\frac{\partial^{l'} \widetilde{E_A}}{\partial x^{l'}}\right|_{x=a}= 0\,\, \text{for}\,\, l'=0,\ldots,l-1\,, 
\end{equation}
also including the possibility $l=0$, and with an obvious generalization to the multivariate case. However the fact that this limit is uniquely defined in multivariate limits is highly nontrivial, as can be seen in the following example of the triangle Cayley determinant with elements as shown in eq.~\eqref{eq:CayleyElements} for $n=3$, in the  limit $p_1^2\to 0,p_2^2\to 0,s_{12}=p_3^2\to 0$:  Denoting this Cayley determinant as $Y_3$, its Taylor expansion in the limit is
\be\label{eq:Y3p1p2p3}
Y_3=0+ 2\sum_{i=1}^3 p_i^2(m_i^2-m_{i-1}^2)(m_{i+1}^2-m_{i-1}^2)+\mathcal{O}(p_j^2p_k^2)\,,
\ee
with $j,k=1,2,3$, $m_{j+3}\equiv m_j$. Clearly, the value of $Y_3$ in the limit \emph{does} depend on the order with which we set the three momenta-squared to zero. The upshot of our analysis will be that \emph{while the limits of individual factors in~$\widetilde{E_A}$ do depend on the rate of approach to the limit, the limit of~$\widetilde{E_A}$ as a whole does not, since different rates of approach produce factors that are already contained in it}.

We start by noting that the only Gram (sub)determinant of a graph with $n\ne 3$ that vanishes as $p^2_i\to 0$ is that of a bubble (sub)graph with momentum $p_i$, $\Gm(p_i)= p_i^2$. By the above definition, the latter gives no nontrivial contribution in the limit. Particularly when $n= 3$ we also have a vanishing $\Gm(p_1^2,p_2^2)\propto\lambda(p_1^2,p_2^2,p_3^2)$, where $\lambda$ is the fully symmetric Källén function defined in eq.~\eqref{eq:Kaellen}. In particular we have $\lambda(0,0,p^2)=p^4$, and therefore this does not give any nontrivial contribution to the limit~\eqref{eq:EA_limit} either. For $n> 3$ we have sums of external momenta as arguments of the triangle subgraph, such that the corresponding Källén function does not vanish. Taking into account the fact that Gram determinants are independent of the masses, this exhausts the analysis of their potential ambiguities in the limit we are considering.

By the same token, it is possible to show that the only Cayley (sub)determinants of any graph that vanish as $m_i^2\to 0$ are the diagonal $m_i^2$ elements of the matrix, also giving no contribution to the limit. Finally, we consider the behavior of the Cayley (sub)determinants when both $m_i^2, p^2_i$ vanish. The only vanishing (sub)determinants in this case are the Cayley determinants of bubble and triangle (sub)graphs. The former are proportional to $\lambda(m_i^2,m_{i+1}^2,p_i^2)$, and by the previous analysis we deduce that they do not affect the limit. For the triangle Cayley determinant $Y_3$, we have already shown in~\eqref{eq:Y3p1p2p3} that it vanishes as $p_1^2,p_2^2,p_3^3\to 0$, and by examining all possibilities we similarly find that up to dihedral images, the only other minimum codimension limit for which the same is true is the limit $p_1^2,m_1^2,m_2^2\to0$. The Taylor expansion around the latter yields
\be\label{eq:Y3p1m1m2}
Y_3=0+ 2p_1^2\prod_{i=1}^2(m_3^2-p_{i+1}^2)+2(p_2^2-p_3^2)\sum_{i=1}^2 (-1)^i m^2_i (m_3^2-p^2_{i+1})+\mathcal{O}(m_j^2m_k^2)+\mathcal{O}(m_j^2p_1^2)+\mathcal{O}(p_1^4)\,,
\ee
with $j,k=1,2$, respectively. We see that in both cases the three nonvanishing derivatives depend on the remaining variables and are not equal to each other, such that the value of $Y_3$ depends on the relative rate with which we approach the limit, and causing a potential ambiguity for the limit of the principal $A$-determinant as a whole. Very interestingly, however, the factor that $Y_3$ contributes in all different rates of approach is already contained in $\widetilde{E_A}$ in the limit. Specifically,
\be
\lim_{p_1^2\to0}G(p_1,p_2)\to -\frac{1}{4}\lambda(p_2^2,p_3^2,0)=-\frac{1}{4}\left(p_2^2-p_3^2\right)^2\,,
\ee
accounts for one of the Taylor expansion coefficients in~\eqref{eq:Y3p1m1m2}, and two-dimensional Cayley subdeterminants $\lambda(p_i^2,m_3^2,0)$ of $Y_3$ account for the rest. Similarly,  the Cayley subdeterminants $\lambda(m_i^2,m_j^2,0)$ take care of all the Taylor expansion coefficients in~\eqref{eq:Y3p1p2p3}.

Along the same lines, inspecting all higher-codimension $m_i^2, p^2_i\to 0$ limits where $Y_3$ vanishes reveals that they are all contained in the two aforementioned codimension-3 limits. Thus they are covered by the previous analysis, and this concludes the proof that the principal $A$-determinant has a well-defined limit for any subset of $m_i^2, p^2_i\to 0$.

The analysis we have carried out may of course also be repeated in any other multivariate limit, including restrictions on the dimension of external kinematics of the type~\eqref{eq:GramConstraintsLargen}. For example, $\widetilde{E_A}$ remains well-defined for any codimension-two limit where the Cayley determinant of any bubble subgraph of an $n$-point graph vanishes. As we will also discuss in the next Section, in all cases we have considered this unambiguous limit of the principal $A$-determinant of the generic one-loop graph to any specific kinematic configuration also matches the direct computation of the principal $A$-determinant of the non-generic graph. These findings suggest that, quite remarkably, the space of solutions of the Landau equations for one-loop Feynman integrals smoothly interpolates between different kinematic configurations related by limits, even though the differential equations these integrals obey, be it of hypergeometric or canonical type, diverge and may not have well-defined such limits.

This in turn supports the expectation that also the symbol alphabet of non-generic Feynman integral may be correctly obtained as a limit of the generic one, based on the re-factorization we have exhibited, of the principal $A$-determinant in terms of the alphabet. In the next Section we will explicitly confirm this expectation in several examples, further corroborating earlier observations made in~\cite{Chen:2022fyw}. { These are also in line with the observation that the diagrammatic coaction~\cite{Abreu:2017enx,Abreu:2017mtm} reduces correctly in the limits of generic to non-generic and possibly divergent graphs. The conjectured equivalence of this coaction with the coaction on the polylogarithms these integrals evaluate to, is then also an implicit conjecture about the relation of their alphabets.}

%%%%%%%%%%%%%%%%%%%%%%%%%%%%%%%%%%%%%%%%%%%%%%%%%%%%%%%%%%%%%%%%%%%%%%%%%%%%%%%%%%%%%%%%%%%%%%%%%%%
\section{Examples}\label{sec: examples}
In this Section we present a series of examples illustrating the formulas for the principal $A$-determinant, symbol alphabet and canonical differential equations of a generic $n$-point one-loop Feynman integral presented in Section \ref{sec: symbol alphabets}, as well as the limiting procedure for obtaining the first two for any non-generic integral,  for various values of $n$. We will consider several of these integrals around both even and odd dimensions of loop momenta $D_0$. 
\subsection{Bubble graph}
Let us start with the $n=2$ or bubble integral illustrated below. 
\begin{equation*}
     \centering
    \begin{tikzpicture}[baseline=-\the\dimexpr\fontdimen22\textfont2\relax]
    \begin{feynman}
    \vertex (a);
    \vertex [right = of a] (b);
    \vertex [right = of b] (c);
    \vertex [right = of c] (d);
    \diagram* {
	    (a) --[fermion, edge label=\(p\)] (b) -- [half left,edge label=\({x_1}\)](c) -- [anti fermion, edge label=\(-p\)](d), (c) -- [half left,edge label=\({x_2}\)](b),
};
\end{feynman}
\end{tikzpicture}
\end{equation*}
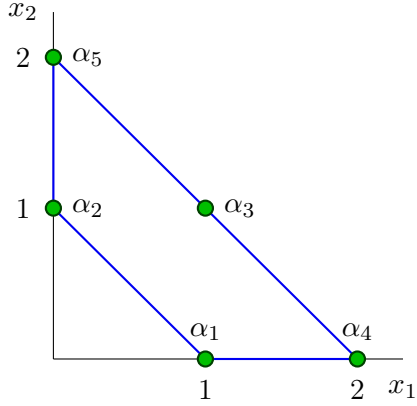
\begin{figure}
    \centering
    \begin{tikzpicture}%
	[scale=2.000000,
	back/.style={loosely dotted, thin},
	edge/.style={color=blue!95!black, thick},
	facet/.style={fill=blue!95!black,fill opacity=0.8},
	vertex/.style={inner sep=2pt,circle,draw=green!25!black,fill=green!75!black,thick}]
%
%
%% This TikZ-picture was produce with Sagemath version 9.3
%% with the command: ._tikz_2d and parameters:
%% scale = 1
%% edge_color = blue!95!black
%% facet_color = blue!95!black
%% opacity = 0.8
%% vertex_color = green
%% axis = False

\draw (0,0) -- (0,2.3);
\draw (0,0) -- (2.3,0);
\draw (-0.2,2.3) node          {$x_2$};
\draw (-0.2,1) node          {1};
\draw (-0.2,2) node          {2};
\draw (2.3,-0.2) node          {$x_1$};
\draw (1,-0.2) node          {1};
\draw (2,-0.2) node          {2};
%% Coordinate of the vertices:
%%
\coordinate (0.00000, 1.00000) at (0.00000, 1.00000);
\coordinate (0.00000, 2.00000) at (0.00000, 2.00000);
\coordinate (1.00000, 0.00000) at (1.00000, 0.00000);
\coordinate (2.00000, 0.00000) at (2.00000, 0.00000);
%%
%%
%% Drawing the interior
%%
%\fill[facet] (2.00000, 0.00000) -- (0.00000, 2.00000) -- (0.00000, 1.00000) -- (1.00000, 0.00000) -- cycle {};
%%
%%
%% Drawing edges
%%
\draw[edge] (0.00000, 1.00000) -- (0.00000, 2.00000);
\draw[edge] (0.00000, 1.00000) -- (1.00000, 0.00000);
\draw[edge] (0.00000, 2.00000) -- (2.00000, 0.00000);
\draw[edge] (1.00000, 0.00000) -- (2.00000, 0.00000);
%%
%%
%% Drawing the vertices
%%
\node[vertex,label=0:\(\alpha_3\)] at (1.00000, 1.00000)     {};
\node[vertex,label=0:\(\alpha_2\)] at (0.00000, 1.00000)     {};
\node[vertex,label=0:\(\alpha_5\)] at (0.00000, 2.00000)     {};
\node[vertex,label=\(\alpha_1\)] at (1.00000, 0.00000)     {};
\node[vertex,label=\(\alpha_4\)] at (2.00000, 0.00000)     {};
\end{tikzpicture}
    \caption{The Newton polytope of the Lee-Pomeransky polynomial $\G$ of the bubble. The edges $(a_1a_4)$ and $(a_2a_5)$ correspond to its tadpole or $n=1$ subgraphs, and the edges $(a_1a_2)$ and $(a_4a_5)$ are the Newton polytopes of the first and second Symanzik polynomials $\mathcal{U}$ and $\mathcal{F}$, respectively. The former does not contribute to $\widetilde{E_A}$.}
    \label{fig: newton polytope bubble}
\end{figure}

The Lee-Pomeransky polynomial for the latter in generic kinematics is
\begin{equation}\label{eq:bubbleExample}
    \G=x_1+x_2+(m_1^2+m_2^2-p^2)x_1x_2+m_1^2x_1^2+m_2^2x_2^2\,,
\end{equation}
so the $A$-matrix is given by
\begin{equation}
    A=\begin{pmatrix}
        1&1&1&1&1\\
        1&0&1&2&0\\
        0&1&1&0&2
    \end{pmatrix}.
\end{equation}
As in Section \ref{sec:FeynmanGGKZ} we let $\alpha_i$ denote the exponents in the $\G$-polynomial. With this we can write the reduced principal $A$-determinant as
\begin{align}
    \widetilde{E_A}(\G)&=\Delta_{\alpha_4}\Delta_{\alpha_5}\Delta_{\alpha_4\alpha_5}\Delta_{\alpha_1\alpha_2\alpha_4\alpha_5}\nonumber\\
    &=m_1^2m_2^2(p^4+m_1^4+m_2^4-2p^2m_1^2-2p^2m_2^2-2m_1^2m_2^2)p^2\,,
\end{align}
where we momentarily index the discriminants with the vertices in the corresponding face. The final expression as a polynomial of the variables $m_1^2,m_2^2,p^2$ indeed agrees with the general expression~\eqref{eq:1LoopEA} in terms of the minors of the modified Cayley matrix,
\begin{equation}
    \cY=\begin{pmatrix}
    0&1&1\\
    1&2m_1^2&m_1^2+m_2^2-p^2\\
    1 & m_1^2+m_2^2-p^2 & 2m_2^2
    \end{pmatrix}.
\end{equation}
Applying our general formulas for the alphabet, eqs.~\eqref{eq: letter one-contracted}-\eqref{eq: rational letter} to this case, we obtain five letters; one for each possible tadpole subgraph and three for the bubble itself. In even space-time dimension the letters are:
\begin{align}
    W_1&=\frac{\cY\begin{bmatrix}
        3\\3
    \end{bmatrix}}{\cY\begin{bmatrix}
        1&3\\1&3
    \end{bmatrix}}=\frac{-1}{2m_1^2},\qquad W_2=\frac{\cY\begin{bmatrix}
        2\\2
    \end{bmatrix}}{\cY\begin{bmatrix}
        1&2\\1&2
    \end{bmatrix}}=\frac{-1}{2m_2^2},\qquad W_{12}=\frac{\cY\begin{bmatrix}
    \cdot\\ \cdot
    \end{bmatrix}}{\cY\begin{bmatrix}
    1\\ 1
    \end{bmatrix}}=\frac{2p^2}{\lambda(p^2,m_1^2,m_2^2)},\nonumber\\
    W_{(1)2}&=\frac{\cY\begin{bmatrix}
        2\\1
    \end{bmatrix}-\sqrt{\cY\begin{bmatrix}
        2\\ 2
    \end{bmatrix}\cY\begin{bmatrix}
        1\\ 1
    \end{bmatrix}}}{\cY\begin{bmatrix}
        2\\1
    \end{bmatrix}+\sqrt{\cY\begin{bmatrix}
        2\\ 2
    \end{bmatrix}\cY\begin{bmatrix}
        1\\ 1
    \end{bmatrix}}}=\frac{-m_1^2 + m_2^2 + p^2-\sqrt{\lambda(p^2,m_1^2,m_2^2)}}{-m_1^2 + m_2^2 + p^2+\sqrt{\lambda(p^2,m_1^2,m_2^2)}},\\
   W_{1(2)}&=\frac{\cY\begin{bmatrix}
       3\\ 1
   \end{bmatrix}-\sqrt{\cY\begin{bmatrix}
       3\\3
   \end{bmatrix}\cY\begin{bmatrix}
       1\\ 1
   \end{bmatrix}}}{\cY\begin{bmatrix}
       3\\ 1
   \end{bmatrix}+\sqrt{\cY\begin{bmatrix}
       3\\3
   \end{bmatrix}\cY\begin{bmatrix}
       1\\ 1
   \end{bmatrix}}}=\frac{-m_1^2 + m_2^2 - p^2-\sqrt{\lambda(p^2,m_1^2,m_2^2)}}{-m_1^2 + m_2^2 - p^2+\sqrt{\lambda(p^2,m_1^2,m_2^2)}},\nonumber
\end{align}
where we recall that $\lambda$ denotes the K\"all\'en function\,,
\begin{equation}
    \lambda(p^2,m_1^2,m_2^2)=p^4+m_1^4+m_2^4-2p^2m_1^2-2p^2m_2^2-2m_1^2m_2^2.
\end{equation}
While the dependence of the bubble on few kinematic variables does not leave much room for nontrivial limits, we have checked that in all codimension-1 limits where a momentum or mass squared vanishes, or two of them are set equal to each other, the above alphabet reduces to harmonic polylogarithms~\cite{Remiddi:1999ew} as expected.

Apart from the bubble alphabet, let us also present the canonical differential equation for the corresponding master integrals. In this case, the choice of basis \eqref{eq:can-int-notation} specializes to
\begin{equation}
    \begin{aligned}
        \mathcal{J}_1&=\frac{\epsilon\,\mathcal{I}_1^{(2)}}{j_1}, &       \mathcal{J}_2&=\frac{\epsilon\,\mathcal{I}_2^{(2)}}{j_2}, &
        \mathcal{J}_{12}&=\frac{\epsilon\,\mathcal{I}_{12}^{(2)}}{j_{12}},
    \end{aligned}
\end{equation}
with the leading singularities being
\begin{equation}
    \begin{aligned}
        j_1^{-1}&=\sqrt{-\cY\begin{bmatrix}
        3\\ 3
    \end{bmatrix}}=1, &       j_2^{-1}&=\sqrt{-\cY\begin{bmatrix}
        2\\ 2
    \end{bmatrix}}=1, &
        j_{12}^{-1}&=\sqrt{-\cY\begin{bmatrix}
        1\\ 1
    \end{bmatrix}}=\sqrt{\lambda(p^2,m_1^2,m_2^2)}.
    \end{aligned}
\end{equation}
Using the expression for the canonical differential equation in eq.~\eqref{eq: CDE D0+E even}, the matrix \eqref{eq:can-ansatz} is given by
\begin{equation}
    \widetilde{M}=\left(
\begin{array}{ccc}
 w_1 & 0 & 0 \\
 0 & w_2 & 0 \\
 -w_{\text{1(2)}} & w_{\text{(1)2}} & w_{12} \\
\end{array}
\right),
\end{equation}
where $w=\log W$.

At odd space-time dimension of loop momentum $D_0$, the rational letters $W_1,\ W_2$ and $W_{12}$ are the same as in the even $D_0$ case, whereas the square-root letters  become
\begin{align}
    W_{(1)2}&=\frac{\cY\begin{bmatrix}
        2\\ 1
    \end{bmatrix}-\sqrt{-\cY\begin{bmatrix}
        \cdot\\ \cdot
    \end{bmatrix}\cY\begin{bmatrix}
        1&2\\ 1&2
    \end{bmatrix}}}{\cY\begin{bmatrix}
        2\\ 1
    \end{bmatrix}+\sqrt{-\cY\begin{bmatrix}
        \cdot\\ \cdot
    \end{bmatrix}\cY\begin{bmatrix}
        1&2\\ 1&2
    \end{bmatrix}}}=\frac{-m_1^2+m_2^2+p^2-\sqrt{4p^2m_2^2}}{-m_1^2+m_2^2+p^2+\sqrt{4p^2m_2^2}}\\
    W_{1(2)}&=\frac{\cY\begin{bmatrix}
        3\\ 1
    \end{bmatrix}-\sqrt{-\cY\begin{bmatrix}
        \cdot\\ \cdot
    \end{bmatrix}\cY\begin{bmatrix}
        1&3\\ 1&3
    \end{bmatrix}}}{\cY\begin{bmatrix}
        3\\ 1
    \end{bmatrix}+\sqrt{-\cY\begin{bmatrix}
        \cdot\\ \cdot
    \end{bmatrix}\cY\begin{bmatrix}
        1&3\\ 1&3
    \end{bmatrix}}}=\frac{-m_1^2+m_2^2-p^2-\sqrt{4p^2m_1^2}}{-m_1^2+m_2^2-p^2+\sqrt{4p^2m_1^2}}
\end{align}

\begin{figure}
    \centering
    \includegraphics[width=0.5\textwidth]{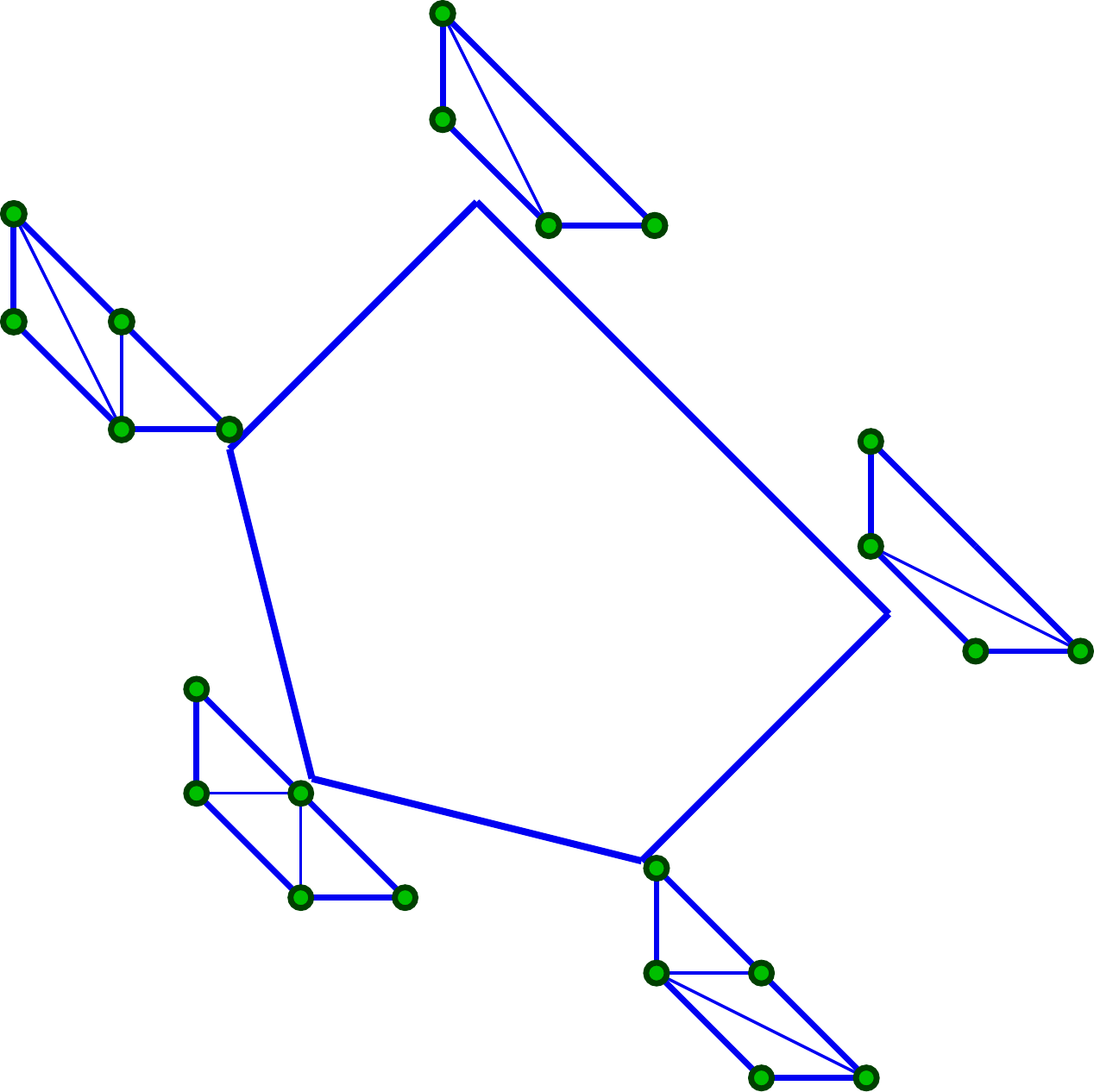}
    \caption{Secondary polytope and regular triangulations of $\newt(\G)$ for the 1-loop bubble graph, with $\G$ as in \eqref{eq:bubbleExample}. It is isomorphic to the $A_2$ cluster polytope.
}
    \label{fig: second polytope bubble}
\end{figure}
For the basis of master integrals and leading singularities we obtain
\begin{equation}
    \begin{aligned}
        \mathcal{J}_1&=\frac{\mathcal{I}_1^{(1)}}{j_1}, &       \mathcal{J}_2&=\frac{\mathcal{I}_2^{(1)}}{j_2}, &
        \mathcal{J}_{12}&=\frac{\epsilon\mathcal{I}_{12}^{(3)}}{j_{12}},
    \end{aligned}
\end{equation}
\begin{equation}
    \begin{aligned}
        j_1^{-1}&=\sqrt{\frac{\cY\begin{bmatrix}
        1 & 3\\ 1 & 3
    \end{bmatrix}}{2}}=\sqrt{m_1^2}, &       j_2^{-1}&=\sqrt{\frac{\cY\begin{bmatrix}
        1& 2\\ 1& 2
    \end{bmatrix}}{2}}=\sqrt{m_2^2}, &
        j_{12}^{-1}&=\sqrt{-2\cY\begin{bmatrix}
        \cdot\\ \cdot
    \end{bmatrix}}=2\sqrt{p_1^2}.
    \end{aligned}
\end{equation}
Putting this together in the differential equation matrix we get
\begin{equation}
    \widetilde{M}=\left(
\begin{array}{ccc}
 w_1 & 0 & 0 \\
 0 & w_2 & 0 \\
 -w_{\text{1(2)}} & w_{\text{(1)2}} & w_{12} \\
\end{array}
\right).
\end{equation}
\begin{remark}
The Newton polytope of $\G$ for the massive bubble is a trapezoid (see Figure \ref{fig: newton polytope bubble}) with five regular triangulations, meaning that the secondary polytope is a pentagon, see Figure \ref{fig: second polytope bubble}. At the same time, the alphabet of the bubble graph may be expressed in terms of the variables of the $A_2$ cluster algebra, whose cluster polytope is also a pentagon. We find this match striking, even though the Newton polytope of $\G$ for the $n$-point graph is $n$-dimensional, and hence the correspondence with cluster algebras, which typically triangulate two-dimensional surfaces, it not as straightforward for $n>2$.
\end{remark}
%%%%%%%%%%%%%%%%%%%%%%%%%%%%%%%%%%%%%%%%%%%%%%%%%%%%%%%%%%%%%%%%%%%%%%%%%%
\subsection{Triangle graphs}
In our next example, we consider the $n=3$ or triangle Feynman graph illustrated in the diagram below. 
\begin{equation*}
    \centering
    \begin{tikzpicture}[baseline=-\the\dimexpr\fontdimen22\textfont2\relax]
\begin{feynman}
\vertex (e){\(p_1\)};
\vertex [right = of e] (a);
\vertex [below right = of a] (b);
\vertex [above right = of a] (c);
\vertex [below right = of b] (e1){\(p_2\)};
\vertex [above right = of c] (e2){\(p_3\)};
\diagram* {
	(e) -- [fermion] (a) -- [edge label'=\(x_2\)] (b) -- [edge label'=\(x_3\)] (c) --[edge label'=\(x_1\)] (a),
	(e1) -- [fermion] (b),
	(e2) -- [fermion] (c),
};
\end{feynman}
\end{tikzpicture}
\end{equation*}

The two Symanzik polynomials of this graph with generic kinematics are
\begin{align}
    \U=&x_1+x_2+x_3,\nonumber\\ \F=&(m_1^2+m_2^2-p_1^2)x_1x_2+(m_1^2+m_2^3-p_3^2)x_1x_3+(m_2^2+m_3^2-p_2^2)x_2x_3\nonumber\\
    &+m_1^2x_1^2+m_2^2x_2^2+m_3^2x_3^2
\end{align}
which with $\G=\U+\F$ gives the $A$-matrix:
\begin{equation*}
    A=\begin{pmatrix}
       1&1&1&1&1&1&1&1&1\\
       1&0&0&1&1&0&2&0&0\\
       0&1&0&1&0&1&0&2&0\\
       0&0&1&0&1&1&0&0&2
    \end{pmatrix}.
\end{equation*}
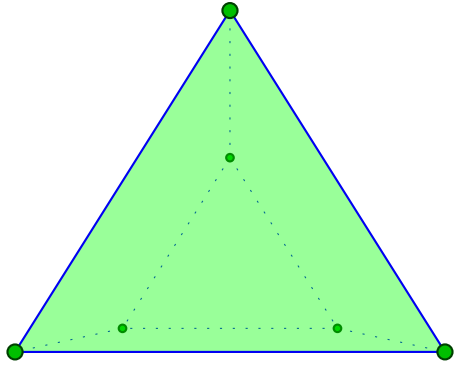
\begin{figure}[h]
    \centering
    \begin{tikzpicture}%
	[x={(-0.707068cm, -0.156292cm)},
	y={(0.707146cm, -0.156185cm)},
	z={(-0.000065cm, 0.975284cm)},
	scale=2.000000,
	back/.style={loosely dotted, thin},
	edge/.style={color=blue!95!black, thick},
	facet/.style={fill=green,fill opacity=0.400000},
	vertexout/.style={inner sep=2pt,circle,draw=green!25!black,fill=green!75!black,thick},
    vertexback/.style={inner sep=1pt,circle,draw=green!25!black,fill=green!75!black,thick}]
%
%
%% This TikZ-picture was produce with Sagemath version 9.3
%% with the command: ._tikz_3d_in_3d and parameters:
%% view = [-0.250300000000000, -0.604300000000000, -0.756500000000000]
%% angle = 145.200000000000
%% scale = 2
%% edge_color = blue!95!black
%% facet_color = green
%% opacity = 0.400000000000000
%% vertex_color = green
%% axis = False

%% Coordinate of the vertices:
%%
\coordinate (0.00000, 0.00000, 1.00000) at (0.00000, 0.00000, 1.00000);
\coordinate (0.00000, 0.00000, 2.00000) at (0.00000, 0.00000, 2.00000);
\coordinate (0.00000, 1.00000, 0.00000) at (0.00000, 1.00000, 0.00000);
\coordinate (0.00000, 2.00000, 0.00000) at (0.00000, 2.00000, 0.00000);
\coordinate (1.00000, 0.00000, 0.00000) at (1.00000, 0.00000, 0.00000);
\coordinate (2.00000, 0.00000, 0.00000) at (2.00000, 0.00000, 0.00000);
%%
%%
%% Drawing edges in the back
%%
\draw[edge,back] (0.00000, 0.00000, 1.00000) -- (0.00000, 1.00000, 0.00000);
\draw[edge,back] (0.00000, 1.00000, 0.00000) -- (0.00000, 2.00000, 0.00000);
\draw[edge,back] (0.00000, 1.00000, 0.00000) -- (1.00000, 0.00000, 0.00000);
\draw[edge,back] (0.00000, 0.00000, 1.00000) -- (0.00000, 0.00000, 2.00000);
\draw[edge,back] (0.00000, 0.00000, 1.00000) -- (1.00000, 0.00000, 0.00000); 
\draw[edge,back] (1.00000, 0.00000, 0.00000) -- (2.00000, 0.00000, 0.00000);
%%
%%
%% Drawing vertices in the back
%%
\node[vertexback] at (0.00000, 1.00000, 0.00000)     {};
\node[vertexback] at (0.00000, 0.00000, 1.00000)     {};
\node[vertexback] at (1.00000, 0.00000, 0.00000)     {};
%%
%%
%% Drawing the facets
%%
\fill[facet] (2.00000, 0.00000, 0.00000) -- (0.00000, 0.00000, 2.00000) -- (0.00000, 2.00000, 0.00000) -- cycle {};
%\fill[facet] (2.00000, 0.00000, 0.00000) -- (0.00000, 0.00000, 2.00000) -- (0.00000, 2.00000, 0.00000) -- cycle {};
%%
%%
%% Drawing edges in the front
%%

\draw[edge] (0.00000, 0.00000, 2.00000) -- (0.00000, 2.00000, 0.00000);
\draw[edge] (0.00000, 0.00000, 2.00000) -- (2.00000, 0.00000, 0.00000);
\draw[edge] (0.00000, 2.00000, 0.00000) -- (2.00000, 0.00000, 0.00000);

%%
%%
%% Drawing the vertices in the front
%%

\node[vertexout] at (0.00000, 0.00000, 2.00000)     {};
\node[vertexout] at (0.00000, 2.00000, 0.00000)     {};

\node[vertexout] at (2.00000, 0.00000, 0.00000)     {};
\end{tikzpicture}
    \caption{The Newton polytope of the Lee-Pomeransky polynomial $\G$ of the massive triangle, with two parallel triangular faces, and three trapezoidal faces from bubble subgraphs.}
    \label{fig: newton polytope triangle}
\end{figure}The Newton polytope of $\G$ is displayed in Figure \ref{fig: newton polytope triangle} where there are three facets looking like Figure \ref{fig: newton polytope bubble}. These facets are obtained from $\newt(\G)$ by intersecting it with one of the coordinate hyperplanes $x_i=0$. This of course corresponds to contracting edge $i$ in the underlying Feynman graph, giving us a bubble graph. The discriminant factors making up the principal $A$-determinant $E_A$ as described in \eqref{eq:EaFact} are:
\begin{align*}
    \Delta_A&=p_1^4+p_2^4+p_3^4-2p_1^2p_2^2-2p_1^2p_3^2-2p_2^2p_3^2,\\
    \Delta_{A\cap x_1=0}&=-p_2^2,\\
    \Delta_{A\cap x_2=0}&=-p_3^2,\\
    \Delta_{A\cap x_3=0}&=-p_1^2,\\
    \Delta_{A\cap\newt(\F)}&=-m_1^2 m_2^2 p_1^2 + m_1^2 m_3^2 p_1^2 + m_2^2 m_3^2 p_1^2 - m_3^4 p_1^2 - 
 m_3^2 p_1^4 - m_1^4 p_2^2 + m_1^2 m_2^2 p_2^2 + m_1^2 m_3^2 p_2^2\\
 &- 
 m_2^2 m_3^2 p_2^2 + m_1^2 p_1^2 p_2^2 + m_3^2 p_1^2 p_2^2 - m_1^2 p_2^4 + 
 m_1^2 m_2^2 p_3^2 - m_2^4 p_3^2 - m_1^2 m_3^2 p_3^2 + m_2^2 m_3^2 p_3^2\\
 &+ 
 m_2^2 p_1^2 p_3^2 + m_3^2 p_1^2 p_3^2 + m_1^2 p_2^2 p_3^2 + m_2^2 p_2^2 p_3^2 - 
 p_1^2 p_2^2 p_3^2 - m_2^2 p_3^4,\\
    \Delta_{A\cap\newt(\F)\cap x_1=0}&=p_2^4+m_2^4+m_3^2-2p_2^2m_2^2-2p_2^2m_3^2-2m_2^2m_3^2,\\
    \Delta_{A\cap\newt(\F)\cap x_2=0}&=p_3^4+m_1^4+m_3^4-2p_3^2m_1^2-2p_3^2m_3^2-2m_1^2m_3^2,\\
    \Delta_{A\cap\newt(\F)\cap x_3=0}
    &=p_1^4+m_1^4+m_2^4-2p_1^2m_1^2-2p_1^2m_2^2-2m_1^2m_2^2,\\
    \prod_{\mathrm{vertices}}\Delta_v&=m_1^2m_2^2m_3^2.
\end{align*}
Again, all these discriminants, calculated directly from the GKZ approach, correspond to minors of the modified Cayley matrix
\begin{equation}
    \cY=\begin{pmatrix}
    0&1&1&1\\
    1&2m_1^2&(m_1^2+m_2^2-p_1^2)&(m_1^2+m_2^3-p_3^2)\\
    1&(m_1^2+m_2^2-p_1^2)&2m_2^2&(m_2^2+m_3^2-p_2^2)\\
    1&(m_1^2+m_2^3-p_3^2)&(m_2^2+m_3^2-p_2^2)&2m_3^2
    \end{pmatrix}\,,
\end{equation}
and up to overall factors we have the identification
\begin{alignat*}{2}
&\Delta_A\to\cY\begin{bmatrix}
    \cdot\\ \cdot
\end{bmatrix},\qquad &&\Delta_{A\cap x_i=0}\to\cY\begin{bmatrix}
    i+1\\i+1
\end{bmatrix},\\
&\Delta_{A\cap\newt(\F)}\to\cY\begin{bmatrix}
    1\\ 1
\end{bmatrix},\qquad 
&&\Delta_{A\cap\newt(\F)\cap x_i=0}\to\cY\begin{bmatrix}
    1&i+1\\ 1&i+1
\end{bmatrix}\,, 
\end{alignat*}
whereas the masses correspond the diagonal elements of $\cY$. As predicted by eq.~\eqref{eq:LetterCounts}, for $D_0$ even there are 18 multiplicatively independent letters, out of which 12 correspond to bubble subgraphs and are thus obtained by relabeling the formulas of the previous subsection. The remaining six letters are
\begin{align}
    W_{123}&=\frac{\cY\begin{bmatrix}
        \cdot\\ \cdot
    \end{bmatrix}}{\begin{bmatrix}
        1\\ 1
    \end{bmatrix}},\label{eq:W123}\\
    W_{(1)23}&=\frac{\cY\begin{bmatrix}
        2\\ 1
    \end{bmatrix}-\sqrt{-\cY\begin{bmatrix}
        \cdot\\ \cdot
    \end{bmatrix}\cY\begin{bmatrix}
        1&2\\1&2
    \end{bmatrix}}}{\cY\begin{bmatrix}
        2\\ 1
    \end{bmatrix}+\sqrt{-\cY\begin{bmatrix}
        \cdot\\ \cdot
    \end{bmatrix}\cY\begin{bmatrix}
        1&2\\1&2
    \end{bmatrix}}}\qquad \mathrm{plus\ cyclic}\ (1)\to(2)\to(3)\,,\\
    W_{(1)(2)3}&=\frac{\cY\begin{bmatrix}
        2\\3
    \end{bmatrix}-\sqrt{\cY\begin{bmatrix}
        \cdot\\ \cdot
    \end{bmatrix}}}{\cY\begin{bmatrix}
        2\\3
    \end{bmatrix}+\sqrt{\cY\begin{bmatrix}
        \cdot\\ \cdot
    \end{bmatrix}}}\qquad \mathrm{plus\ cyclic}\ (1)(2)\to(2)(3)\,.
\end{align}
In this case the basis integrals of eq.~\eqref{eq:can-int-notation} read
\begin{equation}
    \begin{aligned}
        \mathcal{J}_1&=\frac{\epsilon\,\mathcal{I}_1^{(2)}}{j_1}, &       \mathcal{J}_2&=\frac{\epsilon\,\mathcal{I}_2^{(2)}}{j_2}, & \mathcal{J}_3&=\frac{\epsilon\,\mathcal{I}_3^{(2)}}{j_3}, \\
        \mathcal{J}_{12}&=\frac{\epsilon\,\mathcal{I}_{12}^{(2)}}{j_{12}}, & \mathcal{J}_{13}&=\frac{\epsilon\,\mathcal{I}_{13}^{(2)}}{j_{13}}, &
        \mathcal{J}_{23}&=\frac{\epsilon\,\mathcal{I}_{23}^{(2)}}{j_{23}}, \\
        \mathcal{J}_{123}&=\frac{\epsilon^2\mathcal{I}_{123}^{(4)}}{j_{123}}, 
    \end{aligned}
\end{equation}
with leading singularities
\begin{equation}
    \begin{aligned}
        j_1^{-1}&=\sqrt{-\cY\begin{bmatrix}
        3&4\\ 3&4
    \end{bmatrix}}=1, &       j_2^{-1}&=\sqrt{-\cY\begin{bmatrix}
        2&4\\ 2&4
    \end{bmatrix}}=1, \\
    j_3^{-1}&=\sqrt{-\cY\begin{bmatrix}
        2&3\\ 2&3
    \end{bmatrix}}=1, &
        j_{12}^{-1}&=\sqrt{-\cY\begin{bmatrix}
        1 & 4\\ 1&4
    \end{bmatrix}}=\sqrt{\lambda(p_1^2,m_1^2,m_2^2)},\\
    j_{13}^{-1}&=\sqrt{-\cY\begin{bmatrix}
        1 & 3\\ 1&3
    \end{bmatrix}}=\sqrt{\lambda(p_3^2,m_1^2,m_3^2)},&
    j_{23}^{-1}&=\sqrt{-\cY\begin{bmatrix}
        1 & 2\\ 1&2
    \end{bmatrix}}=\sqrt{\lambda(p_2^2,m_2^2,m_3^2)},\\
    j_{123}^{-1}&=2\sqrt{\cY\begin{bmatrix}
        \cdot\\ \cdot
    \end{bmatrix}}=2\sqrt{\lambda(p_1^2,p_2^2,p_3^2)}.
    \end{aligned}
\end{equation}
Putting it all together we get the differential equation matrix
\begin{equation}
\widetilde{M}=\left(
\begin{array}{ccccccc}
 w_1 & 0 & 0 & 0 & 0 & 0 & 0 \\
 0 & w_2 & 0 & 0 & 0 & 0 & 0 \\
 0 & 0 & w_3 & 0 & 0 & 0 & 0 \\
 -w_{\text{1(2)}} & w_{\text{(1)2}} & 0 & w_{12} & 0 & 0 & 0 \\
 -w_{\text{1(3)}} & 0 & w_{\text{(1)3}} & 0 & w_{13} & 0 & 0 \\
 0 & -w_{\text{2(3)}} & w_{\text{(2)3}} & 0 & 0 & w_{23} & 0 \\
 -w_{\text{1(2)(3)}} & w_{\text{(1)(2)3}}+w_{\text{1(2)(3)}} & -w_{\text{(1)(2)3}} & -w_{\text{12(3)}} & w_{\text{1(2)3}} & -w_{\text{(1)23}} & w_{123} \\
\end{array}
\right).
\end{equation}

For the case of odd $D_0$ there are now 19 letters, where again 12 of them are obtained from the odd $D_0$ bubble by relabeling. Out of the remaining seven, the rational letter $W_{123}$ is the same as in the even $D_0$ case~\eqref{eq:W123}, and the rest are
\begin{align}
    W_{(1)23}&=\frac{\cY\begin{bmatrix}
        2\\ 1
    \end{bmatrix}-\sqrt{\cY\begin{bmatrix}
        2\\2
    \end{bmatrix}\cY\begin{bmatrix}
        1\\1
    \end{bmatrix}}}{\cY\begin{bmatrix}
        2\\ 1
    \end{bmatrix}+\sqrt{\cY\begin{bmatrix}
        2\\2
    \end{bmatrix}\cY\begin{bmatrix}
        1\\1
    \end{bmatrix}}},\qquad \mathrm{plus\ cyclic}\ (1)\to(2)\to(3),\\
    W_{(1)(2)3}&=\frac{\cY\begin{bmatrix}
        1&2\\
        1&3
    \end{bmatrix}-\sqrt{-\cY\begin{bmatrix}
        1\\1
    \end{bmatrix}\cY\begin{bmatrix}
        1&2&3\\
        1&2&3
    \end{bmatrix}}}{\cY\begin{bmatrix}
        1&2\\
        1&3
    \end{bmatrix}+\sqrt{-\cY\begin{bmatrix}
        1\\1
    \end{bmatrix}\cY\begin{bmatrix}
        1&2&3\\
        1&2&3
    \end{bmatrix}}},\qquad \mathrm{plus\ cyclic}\ (1)(2)\to(1)(3)\to(2)(3).
\end{align}
The basis integrals now become,
\begin{equation}\label{eq:Basis3Dtriangle}
    \begin{aligned}
        \mathcal{J}_1&=\frac{\mathcal{I}_1^{(1)}}{j_1}, &       \mathcal{J}_2&=\frac{\mathcal{I}_2^{(1)}}{j_2}, & \mathcal{J}_3&=\frac{\mathcal{I}_3^{(1)}}{j_3}, \\
        \mathcal{J}_{12}&=\frac{\epsilon\,\mathcal{I}_{12}^{(3)}}{j_{12}}, & \mathcal{J}_{13}&=\frac{\epsilon\,\mathcal{I}_{13}^{(3)}}{j_{13}}, &
        \mathcal{J}_{23}&=\frac{\epsilon\,\mathcal{I}_{23}^{(3)}}{j_{23}}, \\
        \mathcal{J}_{123}&=\frac{\epsilon\,\mathcal{I}_{123}^{(3)}}{j_{123}}, 
    \end{aligned}
\end{equation}
where the leading singularities are explicitly given by
\begin{equation}
    \begin{aligned}
        j_1^{-1}&=\sqrt{\frac{\cY\begin{bmatrix}
        1&3&4\\ 1&3&4
    \end{bmatrix}}{2}}=\sqrt{m_1^2}, &       j_2^{-1}&=\sqrt{\frac{\cY\begin{bmatrix}
        1&2&4\\ 1&2&4
    \end{bmatrix}}{2}}=\sqrt{m_2^2}, \\
    j_3^{-1}&=\sqrt{\frac{\cY\begin{bmatrix}
        1&2&3\\ 1&2&3
    \end{bmatrix}}{2}}=\sqrt{m_3^2}, &
        j_{12}^{-1}&=\sqrt{-2\cY\begin{bmatrix}
         4\\ 4
    \end{bmatrix}}=2\sqrt{p_1^2},\\
    j_{13}^{-1}&=\sqrt{-2\cY\begin{bmatrix}
         3\\ 3
    \end{bmatrix}}=2\sqrt{p_2^2},&
    j_{23}^{-1}&=\sqrt{-2\cY\begin{bmatrix}
         2\\ 2
    \end{bmatrix}}=2\sqrt{p_3^2},\\
    j_{123}^{-1}&=\sqrt{-2\cY\begin{bmatrix}
        1\\ 1
    \end{bmatrix}}=2\sqrt{-\Delta_{A\cap\newt(\F)}},
    \end{aligned}
\end{equation}
and putting it all together we get the differential equation matrix
\begin{equation}\label{eq:3DtriangleDEMatrix}
\widetilde{M}=\left(
\begin{array}{ccccccc}
 w_1 & 0 & 0 & 0 & 0 & 0 & 0 \\
 0 & w_2 & 0 & 0 & 0 & 0 & 0 \\
 0 & 0 & w_3 & 0 & 0 & 0 & 0 \\
 -w_{\text{1(2)}} & w_{\text{(1)2}} & 0 & w_{12} & 0 & 0 & 0 \\
 -w_{\text{1(3)}} & 0 & w_{\text{(1)3}} & 0 & w_{13} & 0 & 0 \\
 0 & -w_{\text{2(3)}} & w_{\text{(2)3}} & 0 & 0 & w_{23} & 0 \\
 w_{\text{1(2)(3)}} & -w_{\text{(1)2(3)}} & w_{\text{(1)(2)3}} & w_{\text{12(3)}} & -w_{\text{1(2)3}} & w_{\text{(1)23}} & w_{123} \\
\end{array}
\right).
\end{equation}
 As an independent check of our results, we may also compare the finite part of the $D_0=3$ triangle integral $\mathcal{I}_{123}^{(3)}$, first computed in~\cite{Nickel:1978ds}, with our prediction for its symbol based on eqs.~\eqref{eq:Basis3Dtriangle}-\eqref{eq:3DtriangleDEMatrix}, together with eq.~\eqref{eq:symbol}. In particular, our prediction reads,
 \begin{equation}
 \mathcal{I}_{123}^{(3)}\propto\frac{1}{\sqrt{-Y_3}}\log\frac{W_{(1)(2)3}W_{1(2)(3)}}{W_{(1)2(3)}}+\mathcal{O}(\epsilon)\,,
 \end{equation}
 where $Y_3$ denotes the triangle Cayley matrix with elements as shown in eq.~\eqref{eq:CayleyElements}, and $n=3$. Taking into account that in~\cite{Nickel:1978ds} the Cayley matrix has been defined with an extra $1/2$ overall factor, as well as rescaled to become dimensionless by dividing with the masses associated to each row and column, we indeed find agreement.
 
%%%%%%%%%%%%%%%%%%%%%%%%%%%%%%%%%%%%%%%%%%%%%%%%%%%%%%%%%%%%%%%%%%%%%%%%%%%%%%%%%%%%%%%%%%%%%%
\subsection{Box graphs}\label{subsec:Boxes}
In our final example with full kinematic dependence, we present the alphabet for the box graph illustrated below. 
\begin{equation*}
    \centering
     \begin{tikzpicture}[baseline=-\the\dimexpr\fontdimen22\textfont2\relax]
    \begin{feynman}
    \vertex (a);
    \vertex [right = of a] (b);
    \vertex [below = of b] (c);
    \vertex [below = of a] (d);
    \vertex [above left = of a] (x){\(p_1\)};
    \vertex [above right = of b] (y){\(p_4\)};
    \vertex [below right = of c] (z){\(p_3\)};
    \vertex [below left = of d] (w){\(p_2\)};
    \diagram*{
        (x)--[fermion](a), (y)--[fermion](b), (z)--[fermion](c), (w)--[fermion](d), (a) -- [edge label=\(x_1\)] (b) -- [edge label=\(x_4\)] (c) --[edge label=\(x_3\)] (d) -- [edge label=\(x_2\)](a),
    };
    \end{feynman}
    \end{tikzpicture}
\end{equation*}
The two Symanzik polynomials of the box graph with generic massive kinematics are
\begin{align*}
    \U=&x_1+x_2+x_3+x_4,\\
    \F=&(m_1^2+m_2^2-p_1^2)x_1x_2+(m_1^2+m_3^2-s)x_1x_3+(m_1^2+m_4^3-p_4^2)x_1x_4\nonumber\\
       &+(m_2^2+m_3^2-p_2^2)x_2x_3+(m_2^2+m_4^2-t)x_2x_4+(m_3^2+m_4^2-p_3^2)x_3x_4\nonumber\\
       &+m_1^2x_1^2+m_2^2x_2^2+m_3^2x_3^2+m_4^2x_4^2
\end{align*}
and the modified Cayley matrix is given by
\begin{equation}
    \cY=\begin{pmatrix}
        0& 1& 1& 1& 1\\
        1& 2 m_1^2& m_1^2 + m_2^2 - p_1^2& m_1^2 + m_3^2 - s& m_1^2 + m_4^2 - p_4^2\\
  1& m_1^2 + m_2^2 - p_1^2& 2 m_2^2& m_2^2 + m_3^2 - p_2^2& m_2^2 + m_4^2 - t\\
  1& m_1^2 + m_3^2 - s& 
  m_2^2 + m_3^2 - p_2^2& 2 m_3^2& m_3^2 + m_4^2 - p_3^2\\ 
  1&  m_1^2 + m_4^2 - p_4^2& m_2^2 + m_4^2 - t& m_3^2 + m_4^2 - p_3^2& 2 m_4^2.
    \end{pmatrix}
\end{equation}
The symbol alphabet with generic massive kinematics contains 57 letters for $D_0$ even and 61 for $D_0$ odd. These letters are to large to show here but are provided in the auxiliary \ma\ file.

The $\widetilde{M}$ matrices for the case of even and odd $D_0$ are given in \eqref{eq:cde-even} and \eqref{eq:cde-odd}, respectively. Compared to the bubble and triangle examples, the box is the first case where we see that a graph only depends on its subgraphs with at most two legs removed. This is evident in the vanishing of the box-tadpole elements of the differential equation matrices $\widetilde{M}_{15,1},\ldots,\widetilde{M}_{15,4}$.
\newpage
\begin{landscape}
\scriptsize    \begin{equation}\label{eq:cde-even}
        \hspace*{-2cm}\left(
\begin{array}{ccccccccccccccc}
 w_1 & 0 & 0 & 0 & 0 & 0 & 0 & 0 & 0 & 0 & 0 & 0 & 0 & 0 & 0 \\
 0 & w_2 & 0 & 0 & 0 & 0 & 0 & 0 & 0 & 0 & 0 & 0 & 0 & 0 & 0 \\
 0 & 0 & w_3 & 0 & 0 & 0 & 0 & 0 & 0 & 0 & 0 & 0 & 0 & 0 & 0 \\
 0 & 0 & 0 & w_4 & 0 & 0 & 0 & 0 & 0 & 0 & 0 & 0 & 0 & 0 & 0 \\
 -w_{\text{1(2)}} & w_{\text{(1)2}} & 0 & 0 & w_{12} & 0 & 0 & 0 & 0 & 0 & 0 & 0 & 0 & 0 & 0 \\
 -w_{\text{1(3)}} & 0 & w_{\text{(1)3}} & 0 & 0 & w_{13} & 0 & 0 & 0 & 0 & 0 & 0 & 0 & 0 & 0 \\
 -w_{\text{1(4)}} & 0 & 0 & w_{\text{(1)4}} & 0 & 0 & w_{14} & 0 & 0 & 0 & 0 & 0 & 0 & 0 & 0 \\
 0 & -w_{\text{2(3)}} & w_{\text{(2)3}} & 0 & 0 & 0 & 0 & w_{23} & 0 & 0 & 0 & 0 & 0 & 0 & 0 \\
 0 & -w_{\text{2(4)}} & 0 & w_{\text{(2)4}} & 0 & 0 & 0 & 0 & w_{24} & 0 & 0 & 0 & 0 & 0 & 0 \\
 0 & 0 & -w_{\text{3(4)}} & w_{\text{(3)4}} & 0 & 0 & 0 & 0 & 0 & w_{34} & 0 & 0 & 0 & 0 & 0 \\
 -w_{\text{1(2)(3)}} & w_{\text{(1)(2)3}}+w_{\text{1(2)(3)}} & -w_{\text{(1)(2)3}} & 0 & -w_{\text{12(3)}} & w_{\text{1(2)3}} & 0 & -w_{\text{(1)23}} & 0 & 0 & w_{123} & 0 & 0 & 0 & 0 \\
 -w_{\text{1(2)(4)}} & w_{\text{(1)(2)4}}+w_{\text{1(2)(4)}} & 0 & -w_{\text{(1)(2)4}} & -w_{\text{12(4)}} & 0 & w_{\text{1(2)4}} & 0 & -w_{\text{(1)24}} & 0 & 0 & w_{124} & 0 & 0 & 0 \\
 -w_{\text{1(3)(4)}} & 0 & w_{\text{(1)(3)4}}+w_{\text{1(3)(4)}} & -w_{\text{(1)(3)4}} & 0 & -w_{\text{13(4)}} & w_{\text{1(3)4}} & 0 & 0 & -w_{\text{(1)34}} & 0 & 0 & w_{134} & 0 & 0 \\
 0 & -w_{\text{2(3)(4)}} & w_{\text{(2)(3)4}}+w_{\text{2(3)(4)}} & -w_{\text{(2)(3)4}} & 0 & 0 & 0 & -w_{\text{23(4)}} & w_{\text{2(3)4}} & -w_{\text{(2)34}} & 0 & 0 & 0 & w_{234} & 0 \\
 0 & 0 & 0 & 0 & -w_{\text{12(3)(4)}} & w_{\text{1(2)3(4)}} & -w_{\text{1(2)(3)4}} & -w_{\text{(1)23(4)}} & w_{\text{(1)2(3)4}} & -w_{\text{(1)(2)34}} & w_{\text{123(4)}} & -w_{\text{12(3)4}} & w_{\text{1(2)34}} & -w_{\text{(1)234}} & w_{1234} \\
\end{array}
\right)
    \end{equation}\normalsize
%\end{landscape}
%\newpage
%\begin{landscape}
 \scriptsize   \begin{equation}\label{eq:cde-odd}
        \hspace*{-2.5cm}\left(
\begin{array}{ccccccccccccccc}
 w_1 & 0 & 0 & 0 & 0 & 0 & 0 & 0 & 0 & 0 & 0 & 0 & 0 & 0 & 0 \\
 0 & w_2 & 0 & 0 & 0 & 0 & 0 & 0 & 0 & 0 & 0 & 0 & 0 & 0 & 0 \\
 0 & 0 & w_3 & 0 & 0 & 0 & 0 & 0 & 0 & 0 & 0 & 0 & 0 & 0 & 0 \\
 0 & 0 & 0 & w_4 & 0 & 0 & 0 & 0 & 0 & 0 & 0 & 0 & 0 & 0 & 0 \\
 -w_{\text{1(2)}} & w_{\text{(1)2}} & 0 & 0 & w_{12} & 0 & 0 & 0 & 0 & 0 & 0 & 0 & 0 & 0 & 0 \\
 -w_{\text{1(3)}} & 0 & w_{\text{(1)3}} & 0 & 0 & w_{13} & 0 & 0 & 0 & 0 & 0 & 0 & 0 & 0 & 0 \\
 -w_{\text{1(4)}} & 0 & 0 & w_{\text{(1)4}} & 0 & 0 & w_{14} & 0 & 0 & 0 & 0 & 0 & 0 & 0 & 0 \\
 0 & -w_{\text{2(3)}} & w_{\text{(2)3}} & 0 & 0 & 0 & 0 & w_{23} & 0 & 0 & 0 & 0 & 0 & 0 & 0 \\
 0 & -w_{\text{2(4)}} & 0 & w_{\text{(2)4}} & 0 & 0 & 0 & 0 & w_{24} & 0 & 0 & 0 & 0 & 0 & 0 \\
 0 & 0 & -w_{\text{3(4)}} & w_{\text{(3)4}} & 0 & 0 & 0 & 0 & 0 & w_{34} & 0 & 0 & 0 & 0 & 0 \\
 w_{\text{1(2)(3)}} & -w_{\text{(1)2(3)}} & w_{\text{(1)(2)3}} & 0 & w_{\text{12(3)}} & -w_{\text{1(2)3}} & 0 & w_{\text{(1)23}} & 0 & 0 & w_{123} & 0 & 0 & 0 & 0 \\
 w_{\text{1(2)(4)}} & -w_{\text{(1)2(4)}} & 0 & w_{\text{(1)(2)4}} & w_{\text{12(4)}} & 0 & -w_{\text{1(2)4}} & 0 & w_{\text{(1)24}} & 0 & 0 & w_{124} & 0 & 0 & 0 \\
 w_{\text{1(3)(4)}} & 0 & -w_{\text{(1)3(4)}} & w_{\text{(1)(3)4}} & 0 & w_{\text{13(4)}} & -w_{\text{1(3)4}} & 0 & 0 & w_{\text{(1)34}} & 0 & 0 & w_{134} & 0 & 0 \\
 0 & w_{\text{2(3)(4)}} & -w_{\text{(2)3(4)}} & w_{\text{(2)(3)4}} & 0 & 0 & 0 & w_{\text{23(4)}} & -w_{\text{2(3)4}} & w_{\text{(2)34}} & 0 & 0 & 0 & w_{234} & 0 \\
 0 & 0 & 0 & 0 & -w_{\text{12(3)(4)}} & w_{\text{1(2)3(4)}} & -w_{\text{1(2)(3)4}} & -w_{\text{(1)23(4)}} & w_{\text{(1)2(3)4}} & -w_{\text{(1)(2)34}} & w_{\text{123(4)}} & -w_{\text{12(3)4}} & w_{\text{1(2)34}} & -w_{\text{(1)234}} & w_{1234} \\
\end{array}
\right)
    \end{equation}\normalsize
\end{landscape}
\paragraph{Massless off-shell box.}
In the limit $m_1,\ldots,m_4\to 0$ the symbol alphabet simplifies from 57 letters to 25 letters for even $D_0$. We have obtained these letters both from our limiting procedure and from the canonical differential equation directly as an independent check. Our limiting procedure only generates a spanning set of letters, using the provided \ma\ code one obtains 30 letters. By the discussion in Remark \ref{remark: 2 kallen triangle} about letters containing the K\"all\'en function, these 30 letters can be reduced to 25 independent letters. The reduced principal $A$-determinant in this case contains 12 factors:
\begin{align*}
    \Delta_A&=p_1^4 p_3^2 - p_1^2 p_2^2 p_3^2 + p_1^2 p_3^4 - p_1^2 p_2^2 p_4^2 + p_2^4 p_4^2 - 
 p_1^2 p_3^2 p_4^2 - p_2^2 p_3^2 p_4^2 + p_2^2 p_4^4 + p_1^2 p_2^2 s - 
 p_1^2 p_3^2 s \\
 &- p_2^2 p_4^2 s + p_3^2 p_4^2 s - p_1^2 p_3^2 t + p_2^2 p_3^2 t +
  p_1^2 p_4^2 t - p_2^2 p_4^2 t - p_1^2 s t - p_2^2 s t - p_3^2 s t - 
 p_4^2 s t + s^2 t + s t^2,\\
 \Delta_{A\cap\newt(\F)}&=p_1^4 p_3^4 - 2 p_1^2 p_2^2 p_3^2 p_4^2 + p_2^4 p_4^4 - 2 p_1^2 p_3^2 s t - 
 2 p_2^2 p_4^2 s t + s^2 t^2,\\
 \Delta_{A\cap x_1=0}&=p_2^4 - 2 p_2^2 p_3^2 + p_3^4 - 2 p_2^2 t - 2 p_3^2 t + t^2,\\
 \Delta_{A\cap x_2=0}&=p_3^4 - 2 p_3^2 p_4^2 + p_4^4 - 2 p_3^2 s - 2 p_4^2 s + s^2,\\
 \Delta_{A\cap x_3=0}&=p_1^4 - 2 p_1^2 p_4^2 + p_4^4 - 2 p_1^2 t - 2 p_4^2 t + t^2,\\
 \Delta_{A\cap x_4=0}&=p_1^4 - 2 p_1^2 p_2^2 + p_2^4 - 2 p_1^2 s - 2 p_2^2 s + s^2,\\
 \prod_{\mathrm{vertices}}\Delta_v&=stp_1^2p_2^2p_3^2p_4^2.
\end{align*}
\paragraph{Three off-shell legs.} If we in addition to taking the masses to zero also impose $p_4^2\to 0$ the symbol alphabet reduces to 18 letters for $D_0$ even. The reduced principal $A$-determinant contains eleven factors in this case
\begin{align*}
    \Delta_A&= p_1^2sp_2^2-p_1^2st+s^2t-sp_2^2t+st^2+p_1^4p_3^2-p_1^2sp_3^2-p_1^2p_2^2p_3^2-p_1^2tp_3^2-stp_3^2+p_2^2tp_3^2+p_1^2p_3^4,\\
    \Delta_{A\cap x_1=0}&=p_2^4+p_3^4+t^2-2p_2^2p_3^2-2p_2^2t-2p_3^2t,\\
    \Delta_{A\cap x_4=0}&=p_1^4+p_2^4+s^2-2p_1^2p_2^2-2p_1^2s-2p_2^2s,\\
    \Delta_{A\cap x_2=0}&=p_3^2-s,\\
    \Delta_{A\cap x_3=0}&=p_1^2-t,\\
    \Delta_{\Gamma}&=p_1^2p_3^2-st,\\
    \prod_{\mathrm{vertices}}\Delta_v&=stp_1^2p_2^2p_3^2.
\end{align*}
We have now reached a level of reduced kinematics such that the identification between discriminants, determinants and subgraphs discussed in Section \ref{subsec:OneLoopPricADet} breaks down. The face $\Gamma$ of $\newt(\G)$ is given by
\begin{equation}
    \Gamma=\conv\begin{pmatrix}
        1&0&1&0\\
        0&1&1&0\\
        1&0&0&1\\
        0&1&0&1
    \end{pmatrix}\,,    
\end{equation}
which should correspond to the box topology since all rows contain non-zero elements. This face actually corresponds to the box with only two off-shell external legs positioned at opposite corners, which also appeared as an example in the introduction. This is clearly not a subgraph of the box with three off-shell legs.

The mathematical origin of this subtlety is described in the following example.
\begin{example}\label{ex:discriminant_not_determinant}
Let $\F=c_1x_1x_3+c_2x_2x_4+c_3x_1x_2+c_4x_2x_3+c_5x_3x_4$ where $(c_1,c_2,c_3,c_4,c_5)=(-s,-t,-p_1^2,-p_2^2,-p_3^2)$, this is the $\F$-polynomial of the box with all internal masses zero and three external massive legs. The coefficient matrix of the Jacobian, $\jac(\F)$, is just the Cayley matrix 
\begin{equation}\label{ex:detY}
    Y=\begin{pmatrix}
    0&c_3&c_1&0\\
    c_3&0&c_4&c_2\\
    c_1&c_4&0&c_5\\
    0&c_2&c_5&0
    \end{pmatrix},\qquad \det(Y)=(c_1c_2-c_3c_5)^2\,,
\end{equation}
whose determinant is clearly reducible. This means that we expect the distinction between punctured affine space and algebraic torus to matter. To correctly calculate the disciminant we would work in the algebraic torus and compute\small\begin{align*}
    &\overline{\left\lbrace c \in \CC^5 \;|\; \frac{\partial \F}{\partial x_1}= \cdots =\frac{\partial \F}{\partial x_4}=0 {\; \rm has \; a\; solution\; for\; }x\in (\CC^*)^4 \right\rbrace} = \{ c\in \CC^5 \;|\; c_4=c_1c_2-c_3c_5=0\}\,,
\end{align*}\normalsize
which has codimension two, meaning that $\Delta_{A\cap\newt\F}(\F)=1$ per definition.
On the other hand, from linear algebra, we know that the determinant in \eqref{ex:detY} corresponds to working in the punctured affine space, where we obtain  \small\begin{align*}
    &\overline{\left\lbrace c \in \CC^5 \;|\; \frac{\partial \F}{\partial x_1}= \cdots =\frac{\partial \F}{\partial x_4}=0 {\; \rm has \; a\; solution\; for\; }x\in \CC^4 \setminus\{\mathbf{0}\}\right\rbrace} = \{ c\in \CC^5 \;|\; (c_1c_2-c_3c_5)^2=0\};
\end{align*}\normalsize note that the defining polynomial $(c_1c_2-c_3c_5)^2$ is  {\bf not} the desired discriminant.  
Hence working over the torus rather than the punctured affine space is essential in this example (precisely because $\det(Y)$ is reducible).
\end{example}
Despite these subtleties in the kinematic limits of individual discriminants, in Section \ref{sec: limiting procedure} we have provided strong evidence that the entire principal $A$-determinant does remain well-defined in these limits. In practice, therefore, they do not matter.
%%%%%%%%%%%%%%%%%%%%%%%%%%%%%%%%%%%%%%%%%%%%%%%%%%%%%%%%%%%%%%%%%%%%%%%%%%%%%%%%%%%%%%%%%%%%%%
\subsection{Pentagon graphs and beyond}
The generic $n=5$ or pentagon graph shown below depends on 15 dimensionfull variables, and as we have mentioned, its principal $A$-determinant consists of 57 different Cayley and Gram determinants, one of each associated to the graph itself (leading Landau singularities of type I and II), and the rest to its subgraphs. For simplicity we will restrict the discussion to the case of even dimension $D_0$ of loop momenta, but the entire analysis may of course be repeated also for the odd case. By the process of refactorizing these in pairs as described in subsection \ref{subsec:SymbolLettersFormula}, we obtain a total of 166 letters, out of which 16 are genuinely new, and the rest may be obtained by relabeling the letters presented in the previous subsections, as they are associated to subgraphs with $n<5$. As functions of the masses and Mandelstam invariants, these letters in total contain 26 square roots. Instead of presenting lengthy formulas, we refer the reader to the ancillary file for this new result.
 \begin{equation*}
     \centering
     \begin{tikzpicture}[baseline=-\the\dimexpr\fontdimen22\textfont2\relax]
     \begin{feynman}
     \vertex (a);
     \vertex [below left = 1cm  of a] (b);
     \vertex [below left = 0.8cm of b] (c);
     \vertex [below = 1.9cm of b] (d);
     \vertex [below right = 1cm of a] (g);
     \vertex [below right = 0.8cm of g] (f);
     \vertex [below = 1.9cm of g] (e);
     %Vertices for external momenta
     \vertex [above = of a](p1){\(p_1\)};
     \vertex [left = of c](p2){\(p_2\)};
     \vertex [below left = of d](p3){\(p_3\)};
     \vertex [below right = of e](p4){\(p_4\)};
     \vertex [right = of f](p5){\(p_5\)};
     \diagram*{
         (a)-- [edge label'=\(x_2\)](b)--(c)-- [edge label'=\(x_3\)](d)-- [edge label'=\(x_4\)](e)-- [edge label'=\(x_5\)](f)--(g)-- [edge label'=\(x_1\)](a),
         (p1)--[fermion](a),(p2)--[fermion](c),(p3)--[fermion](d),(p4)--[fermion](e),(p5)--[fermion](f),
     };
     \end{feynman}
  \end{tikzpicture}
    \label{fig:poly feynman}
 \end{equation*}

Based on the evidence presented in subsection~\ref{sec: limiting procedure}, we expect that any other pentagon, where some of its masses or momenta have been identified or set to zero, may be obtained from the generic one by the limiting procedure we have described. In what follows, we will apply it to obtain further new results, as well as to check it against previously computed alphabets.

In particular, we will consider limits where all internal masses, as well as some of the external momenta have been set to zero. Let us start with the case of three offshell external legs, which thus now depends on 8 dimensionful variables. We distinguish two cases, based on whether the all three offshell legs are adjacent (`hard') or not (`easy'), and after taking the corresponding limits of the generic pentagon and eliminating multiplicative relations between letters, we arrive at 57 letters containing 7 square roots, and 54 letters containing 5 square roots, respectively. As far as we are aware of, these alphabets have not appeared in the literature before.

Next, we may continue the limiting process to similarly obtain pentagons where two external legs are offshell. Choosing for example the `easy' configuration where the offshell legs are not adjacent, for this 7-variable alphabet we obtain 40 letters depending on 3 square roots, and we have checked that it is indeed equivalent to the one previously computed in~\cite{He:2022tph}. Moving on to send another momentum-squared to zero, we land on the 6-variable alphabet of the massless pentagon with one offshell leg, which consists of 30 letters containing 2 square roots. We have also compared our alphabet to the result for the latter reported in~\cite{Abreu:2020jxa}, again finding perfect agreement.

Finally, let us also briefly discuss the $n=6$ or hexagon case.
As we have mentioned at the end of subsection~\ref{subsec:SymbolLettersFormula}, this integral requires us to be in $n>d=5$ dimensions of external kinematics for all distinct Mandelstam invariants to be algebraically independent, and hence also for the symbol letters to be multiplicatively independent. We will in particular be restricting our attention to the limit where all masses and momenta squared are set to zero, and to the letters exclusively associated to the hexagon graph, and not to its subgraphs. In this limit, the 15 square-root letters of the second type, eq.~\eqref{eq: letter two-contracted} reduce to 3 multiplicatively ones, which now only depend on one square root, whereas all 6 square-root letters of the first type, eq.~\eqref{eq: letter two-contracted}, remain multiplicatively independent. Together with the rational letter~\eqref{eq: rational letter}, we thus in total have 10 genuinely hexagon letters, and also here we establish their equivalence to their earlier determination in~\cite{Henn:2022ydo}. These checks further solidify the evidence provided in subsection~\ref{sec: limiting procedure} on the well-defined nature of the limiting process for principal $A$-determinants and symbol alphabets, and also support the correctness of the new limiting results we have obtained.

%%%%%%%%%%%%%%%%%%%%%%%%%%%%%%%%%%%%%%%%%%%%%%%%%%%%%%%%%%%%%%%%%%%%%%%%%%%%%%%%%%%%%%%%%%%%%%%%%%%
\section{Normality, Cohen-Macaulay, and generalized permutohedra}\label{sec: normality}

{In this section we study several mathematical properties of Feynman integrals. In subsection~\ref{subsec:normalityOneLoop}, we rigorously prove that the {\em Cohen-Macaulay} property holds for  a larger collection of one-loop Feynman integrals\footnote{More precisely, this is a property of the toric ideal associated to the Feynman integral, as defined in \eqref{eq:toricIdeal}.} than was previously known; for one loop integrals this generalizes previous results of \cite{walther2022feynman,TH22}. %This proof uses the GKZ framework, see subsection~\ref{section:LeePomRep}. 
As discussed in the introduction, and detailed further below, the physical meaning of the Cohen-Macaulay property is that the number of master integrals for a given topology is independent of the spacetime dimension and of the generalized propagator powers. 

In subsection~\ref{sec:GP}, we prove that the Newton polytope of the second Symanzik polynomial, as defined in Section~\ref{section:LeePomRep}, is a {\em generalized  permutohedron} (GP) for a new class of Feynman integrals of arbitrary loop order. { The practical utility of this property, is that it facilitates new methods for fast Monte Carlo evaluation of Feynman integrals \cite{Borinsky:2020rqs,Borinsky:2023jdv}.}

The Cohen-Macaulay and GP property are also related to numerous other important properties in the context of toric geometry and the study of the polytopes associated to toric varieties. To better orient the reader, we summarize a selection of these properties and their relations in Figure~\ref{fig:NormCMGPEtc}. Note that it is in particular the stronger property of {\em normality}, which implies Cohen-Macaulay, that we will prove in subsection~\ref{subsec:normalityOneLoop}. It would be very interesting to understand any additional physical implications these properties have for Feynman integrals, for example with respect to their ultraviolet or infrared divergences, however we will not attempt this here.

In the rest of this preamble, for the sake of completeness, we briefly recall the definitions of the properties summarized in Figure~\ref{fig:NormCMGPEtc}, and further elaborate on the implications of the Cohen-Macaulay property for GKZ-hypergeometric systems and the associated Feynman integrals. Since, in this section, we are aiming for mathematical rigor, its content will inescapably be technical in nature. However, at the beginning of each subsection we will point the reader to the main results, and explain their physical significance.}

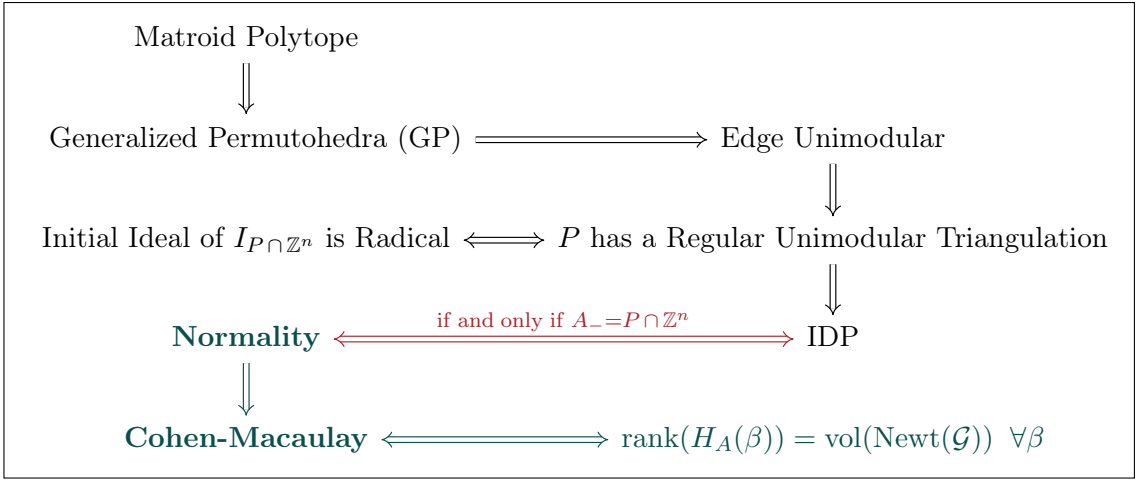
\begin{figure}[h]
%\fbox{\parbox[c]{\textwidth}
\begin{tikzpicture}[mybox/.style={draw, inner sep=5pt}]
\node[mybox] (box){%
\begin{tikzcd}[arrows=Rightarrow]
 \text{Matroid Polytope}\arrow[d]  \\
 \text{~~Generalized Permutohedra (GP)} \arrow[r]  & \text{Edge Unimodular}\arrow[d] \\
  \text{Initial Ideal of $I_{P\,\cap \,\ZZ^n}$ is Radical} \arrow[r,Leftrightarrow] & P \text{~has a Regular Unimodular Triangulation}\arrow[d] \\
\color{darkTeal} \text{\bf Normality}  \arrow[r,Leftrightarrow,"\mathrm{if \; and \; only\; if\;}   A_-=P\,\cap\,\ZZ^n",brickRed]\arrow[d,darkTeal] &
  \text{IDP}   \\ \text{\bf \color{darkTeal}Cohen-Macaulay} \arrow[r,Leftrightarrow,darkTeal] 
 & \color{darkTeal}\rank(H_A(\beta))=\vol(\newt(\G))\;\;\forall\beta 
\end{tikzcd}};
\end{tikzpicture}
\caption{The diagram above considers the relationships between various properties of the hypergeometric system $H_A(\beta)$ arising from the Lee-Pomeransky polynomial $\G$ of a Feynman graph $G$, with associated semi-group $\NN A$ and toric ideal $I_A$ and the associated polytope $P=\conv(A_-)$. The properties in {\bf\color{darkTeal} bold teal colored text} are properties of the toric ideal $I_A$, the semi-group ring $\CC[\NN A]$ and the hypergeometic system $H_A(\beta)$; the black plain text properties are properties of the associated polytope. Definitions of a matroid polytope, the Integer Decomposition Property (IDP), and edge unimodularity are given in subsection \ref{sec:GP}. Note that the equivalence (in {\color{brickRed}red}) between Normality and IDP holds if and only if $ A_-=P\,\cap\,\ZZ^n$.  }\label{fig:NormCMGPEtc}
\end{figure}

Generalized hypergeometric systems { $H_A(\beta)$, as defined by the matrix in \eqref{eq:AMatrix} and vector in \eqref{eq:BetaVector}}, have many nice combinatorial and analytic properties. The dimension of the solutions space of $H_A(\beta)$, also referred to as the rank of $H_A(\beta)$, is what physicists would call the number of {master integrals}. %~\cite{Agostini:2022cgv}. 
More precisely, for most one-loop integrals the number of master integrals and the rank of $H_A(\beta)$ are exactly the same, for higher loop orders or for special kinematics the rank of the hypergeometric system $H_A(\beta)$ 
 only provides an upper bound on the number of master integrals. For generic $\beta$ it is a classical result by GKZ that the dimension of the solution space is given by the volume of the Newton polytope of $\G$, i.e. $\mathrm{rank}(H_A(\beta))=\vol(\newt(\G))=\vol(\conv(A_-))$. However, for actual physical calculations the vector $\beta$ is non-generic, for example the choice $\beta=(-D/2,-1,\ldots,-1)^T$ {of space-time dimension and generalized propagator powers}  is often used. In order for the equality between the rank and volume of $\newt(\G)$ to hold for every $\beta$ we need $I_A$ to be {\em Cohen-Macaulay}. Note that here, and in all other instances, by {\em volume} we mean {\em normalized volume}, that is we use the convention that  the standard simplex in $\RR^n$ has volume 1; that is our volume is $n!$ multiplied by the usual Euclidean volume in $\RR^n$. 

For our purposes we will define the Cohen-Macaulay property of toric ideals $I_A$ in terms of hypergeometric systems; by doing so we are employing a deep result of \cite{matusevich2005homological}. For a definition of the Cohen-Macaulay property for arbitrary polynomial ideals we refer the reader to the book \cite{bruns1998cohen}. We say the ideal $I_A$ in $\CC[\partial]$ and the semigroup ring $\CC[\NN A]\cong \CC[\partial]/I_A$ are {\em Cohen-Macaulay} if the associated hypergeometric system $H_A(\beta)$ is such that   $\rank(H_A(\beta))=\vol(\conv(A_-))$ for all $\beta$.
Hence in particular we have the following equivalence
\begin{equation}
    \rank(H_A(\beta))=\vol(\newt(\G))\ \forall\beta\Longleftrightarrow I_A\ \mathrm{is\ Cohen-Macaulay}.\label{eq:CMConstRankAllBeta}
\end{equation}
In subsection \ref{subsec:normalityOneLoop} we prove that for one-loop graphs the semi-group $\NN A$ is {\em normal} for almost every kinematic setup. We say the semi-group $\NN A$, and the associated semi-group ring $\CC[\NN A]$, are {\em normal} if $$\NN A=\ZZ A \cap \RR_{\geq 0} A.$$ Showing that $\NN A$ is normal will in turn imply that the associated toric ideal $I_A$ is Cohen-Macaulay by a result of \cite{hochster1972rings}; hence in particular the relation \eqref{eq:CMConstRankAllBeta} holds.

In the one loop case, with all internal and external masses non-zero and different we get the simple formula for the number of master integrals:
\begin{equation}
    \#(\mathrm{master\ integrals})=\vol(\newt(\G))=2^{{n}}-1.\label{eq:numMasterIntegrals1Loop}
\end{equation}
This formula has been quoted before, see e.g. \cite{Mizera:2021icv}.

The properties above, namely normality and the Cohen-Macaulay property, are properties of the semi-group $\NN A$ and the ring $\CC[\NN A]$, on the other hand the generalized  permutohedra property is a property of a polytope $P$. A polytope $P\subset\RR^n$ is a {\em generalized permutohedra (GP)} if and only if every edge is parallel to $\mathbf{e}_i-\mathbf{e}_j$ for some $i,j\in\{1,\ldots,n\}$, where the ${\bf e}_\ell$ denotes the standard basis vectors in $\RR^n$. When all lattice points in $P=\conv(A_-)$ are contained in $A_-$ then properties of the polytope, such as the GP property, have relations with those of the semi-group ring $\CC[\NN A]$; in particular in this case GP implies normality. Assuming $A_-=P\cap \ZZ^n$, another property which implies that $\CC[\NN A]$ is normal (and hence Cohen-Macaulay) is if there is some monomial order such that the {\em initial ideal} of $I_{A_-}$ is radical; more precisely  \cite[Corollary 8.9]{sturmfels1996grobner} tells us that this is equivalent to $P$ having a {\em regular unimodular triangulation}, which in turn implies normality. We now briefly recall some of these definitions. A regular triangulation of a polytope $P\subset \RR^n$ is called unimodular if it consists only of simplices with volume $1$  (recall our convention that a standard simplex has volume 1);  
for a definition of a regular triangulation see \cite[\S5.1]{ziegler2012lectures}. For a polynomial ideal $I$ with {\em Gr\"obner basis} $\{g_1, \dots, g_r \}$ in a polynomial ring with some monomial order $<$ the {\em initial ideal} is the ideal $in_<(I):=\langle in_<(g_1),\dots, in_<(g_r) \rangle$ where $in_<(f)$ is the monomial of a polynomial $f$ which is largest with respect to the ordering $<$, \cite[\S1.2]{michalek2021invitation}. We say an ideal $I$ in a polynomial ring $R$ is {\em radical} if $I=\sqrt{I}$ where  $\sqrt{I}:=\{ f\in R\;|\; f^\ell \in I \;{\rm for \; some\;} \ell \in \NN\}$. 
{We remind the reader that the key properties defined in the last several paragraphs, and their relations, may be found in Figure \ref{fig:NormCMGPEtc}.}

\subsection{Normality and Cohen-Macaulay for one-loop Feynman graphs}
\label{subsec:normalityOneLoop}

% In particular, we show that the semi-group $\NN A$ is normal if and only if we have that \begin{equation}
%     p(F_{ij})^2-m_i^2-m_j^2= (p_i+\cdots + p_{j-1})^2-m_i^2-m_j^2\neq 0\label{eq:CMCondExplain}
% \end{equation} for all edges $i,j$ where both $m_i\neq0$ and $m_j\neq0$. Hence this condition is sufficient for the Cohen-Macaulay property to hold, see Figure \ref{fig:NormCMGPEtc}. Note that if only one of the masses are non-zero, cancellation is allowed and the Cohen-Macaulay property (and even more strongly, normality) will hold. We also wish to highlight that all individual terms in \eqref{eq:CMCondExplain} are allowed to be zero, it is only when both masses are non-zero that care must be taken.

{ Throughout this subsection we assume that the Feynman graph $G$ under consideration is a one-loop Feynman graph, and employ the notations introduced in Section \ref{sec:FeynmanGGKZ} with Symanzik polynomials $\U,\F$ giving $\G=\U+\F$ and  $A= \{1\}\times A_-= \{1\}\times{\rm Supp}(\G)$ as above, see e.g.~\eqref{eq:AMatrix}. Note that taking $\G_h=\U x_0 +\F$ we have that the matrix ${\rm Supp}(\G_h)$ is obtained from $A$ by elementary row operations; hence without loss of generality we may (and will) assume ${\rm Supp}(\G_h)=A$.

The main result of this subsection is Theorem \ref{thm:Normal_1loop}, which proves that the semi-group $\NN A$ is normal if and only if we have that \begin{equation}
    p(F_{ij})^2-m_i^2-m_j^2= (p_i+\cdots + p_{j-1})^2-m_i^2-m_j^2\neq 0\label{eq:CMCondExplain}
\end{equation} for all edges $i,j$ where both $m_i\neq0$ and $m_j\neq0$; that is, the only cases where normality does not hold are when either all three terms on the right of the above equation are nonzero and the entire right-hand side vanishes, or when $p(F_{ij})^2=0$ and the two masses are equal to each other\footnote{We do not currently have a physical justification for why this is the case, but it would be interesting to address this in the future.}. Hence this condition is sufficient for the Cohen-Macaulay property to hold, see Figure \ref{fig:NormCMGPEtc}. Note that if only one of the masses are non-zero, cancellation is allowed and the Cohen-Macaulay property (and even more strongly, normality) will hold. We also wish to highlight that all individual terms in \eqref{eq:CMCondExplain} are allowed to be zero; this includes the case where all internal propagators are massless for any external kinematics, as illustrated in Corollary \ref{corollary_massless}.}

For Feynman graphs $G$ of any loop order, but under the assumption of generic (i.e.~non-zero) momenta, so that { the momentum flow $p(F)$ between the two connected components of any two-forest $F$ of $G$ is nonzero, the Cohen-Macaulay property has been proven} in the fully massive case  by \cite{TH22} and for {generic (but some times zero)} masses by \cite{walther2022feynman}.

Our proof in the one loop case is founded on a result of \cite{ohsugi1998normal}; we begin with the relevant definitions. To a  graph we may associate a matrix which catalogs which vertices in a graph are joined by an edge; note that in the discussion which follows this will be a {\em different } graph than the Feynman graph $G$. 

\begin{definition}[Edge Matrix] Let $H=(E,V)$ be a finite connected graph with vertex set $V=\{0,\ldots,d\}$. If $e=\{i,j\}$ is an edge of $H$ joining vertices $i$ and $j$ we define $\rho(e)\in\RR^{d+1}$ by $\rho(e)=\mathbf{e}_i+\mathbf{e}_j$ where $\mathbf{e}_i$ is the $i$th unit vector in $\RR^{d+1}$. Let $M$ be the matrix whose columns correspond to the finite set $\{\rho(e):\ e\in E\}$, then $M$ is called the \emph{edge matrix} of $H$ and the convex hull of $M$ is called the \emph{edge polytope}.
\label{def:Edge_Matrix}
\end{definition}

A result of \cite{ohsugi1998normal} will tells us that if we can associate a certain graph $H$ to $A$ then the semi-group $\NN A$ is normal, hence establishing the desired result. This result is based on verifying the following condition for a graph $H$.

\begin{definition}[Odd cycle condition]\label{def:oddCycle}
A cycle in a graph is called \emph{minimal} if it has no chord and it is said to be \emph{odd} if it is a cycle with odd length. A graph $H$ satisfies the \emph{odd cycle condition} if for two arbitrary minimal odd cycles $C$ and $C'$, either $C$ and $C'$ have a common vertex or there is an edge connecting a vertex of $C$ with a vertex of $C'$.

\end{definition}
We may now state a  result of \cite{ohsugi1998normal} which we will apply in Theorem \ref{thm:Normal_1loop} below. 
\begin{theorem}[Corollary 2.3, \cite{ohsugi1998normal}]\label{thm: odd-cycle}
 Let $B$ be the edge matrix of a graph $H$. The semi-group $\NN B$ is normal if and only if the graph $H$ satisfies the odd cycle condition.   
\end{theorem}

{ Using the above result we now prove the main result of this Section, which gives a precise condition for when the semi-group $\NN A$ is normal for the matrix $A$ associated to a one loop Feynman graph; recall that normality of $\NN A$ implies the Cohen-Macaulay property holds, see Figure \ref{fig:NormCMGPEtc}.}% which tells us the ideal $I_A$ has the Cohen-Macaulay property and \eqref{eq:CMConstRankAllBeta} holds for almost all one-loop Feynman graphs. 
% {In particular, we show that the semi-group $\NN A$ is normal if and only if we have that \begin{equation}
%     p(F_{ij})^2-m_i^2-m_j^2\neq 0\label{eq:CMCondExplain}
% \end{equation} for all edges $i,j$ where both $m_i\neq0$ and $m_j\neq0$. Hence this condition is sufficient for the Cohen-Macaulay property to hold, see Figure \ref{fig:NormCMGPEtc}. Note that if only one of the masses are non-zero, cancellation is allowed and the Cohen-Macaulay property (and even more strongly, normality) will hold. We also wish to highlight that all individual terms in \eqref{eq:CMCondExplain} are allowed to be zero, it is only when both masses are non-zero that care must be taken.}\ycomment{I (believe I) have summarized this statement at the beginning of the section 5 preamble, do we want to keep both instances or just one?}
% \begin{theorem}\label{thm:Normal_1loop}Let $\mathcal{G}_h=\mathcal{U}x_0+\mathcal{F}$ be the Symanzik polynomial of a one-loop Feynman diagram $G$, {and assume \eqref{eq:CMCondExplain} holds for every pair of non-zero masses.} Then the matrix $A=\supp(\mathcal{G}_h)$ is the edge matrix of a graph $H$ satisfying the odd cycle condition. Hence, in particular, the semi-group $\NN A$ is normal. 
% \end{theorem}
\begin{theorem}\label{thm:Normal_1loop}Let $\mathcal{G}_h=\mathcal{U}x_0+\mathcal{F}$ be the Symanzik polynomial of a one-loop Feynman diagram $G$. {Then the matrix $A=\supp(\mathcal{G}_h)$ is the edge matrix of a graph $H$ satisfying the odd cycle condition if and only if we have that \begin{equation}
    p(F_{ij})^2-m_i^2-m_j^2\neq 0\label{eq:CMCondExplain}
\end{equation} for all edges $i,j$ where both $m_i\neq0$ and $m_j\neq0$. Hence, in particular, the semi-group $\NN A$ is normal if and only if \eqref{eq:CMCondExplain} holds for all edges $i,j$ where both $m_i\neq0$ and $m_j\neq0$. }
\end{theorem}
\begin{proof}
First consider the structure of the matrix $A=\supp(\mathcal{G}_h)$ for our one loop Feynman diagram $G$. Since we have exactly one loop in the diagram $G$ then the exponent vectors of $\U$, which correspond to the spanning trees of $G$, are obtained by removing exactly one edge from the loop in $G$, and the exponent records this removed edge giving an identity matrix obtained from the monomials of $\U$. In other words, this means that the exponents of $\U$ are the indicator vectors of the bases of the co-graphic matroid of $M^*(G)$ of $G$. The matrix obtained from $\U x_0$ is then an identity matrix with a row of ones added on top:
\begin{equation*}
    \begin{pmatrix}
    \begin{matrix}
    1&\cdots & 1
    \end{matrix}\\
     \mathbf{1}_{E}
    \end{pmatrix}
\end{equation*}
where $\mathbf{1}_E$ is the $E\times E$-identity matrix.

Hence the matrix $A$ is the edge matrix (as in Definition \ref{def:Edge_Matrix}) of a graph $H$. The part of $A$ arising from $\U$ gives $E+1$ vertices of $H$ and exactly one edge between the central vertex (corresponding to the exponent of $x_0$ in $\U x_0$). We now construct the graph $H$ in three steps; the first step adds the vertices and edges arising from   $\U x_0$, this step is illustrated in Figure \ref{fig:U_graph}.

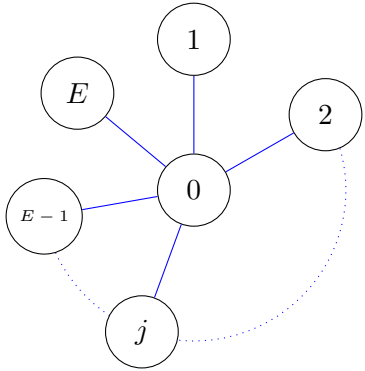
\begin{figure}[h]
    \centering
       \begin{tikzpicture}
\draw[blue,dotted] (250:2) arc (250:390:20mm);
   \draw[blue,dotted] (190:2) arc (190:250:20mm);
    \node[state,fill=white] (center) at (0,0) {$0$};
    \node[state,fill=white] (1) at (90 :2cm) {$1$};
    \draw[blue] (1) -- (center);
      \node[state,fill=white] (2) at (30 :2cm) {$2$};
    \draw[blue] (2) -- (center);
          \node[state,fill=white] (j) at (250 :2cm) {$j$};
    \draw[blue] (j) -- (center);
        \node[state,fill=white] (Em1) at (190 :2cm) {\tiny$E-1$};
    \draw[blue] (Em1) -- (center);
            \node[state,fill=white] (E) at (140 :2cm) {$E$};
    \draw[blue] (E) -- (center);
%\foreach \phi in {1,...,6}{
 %   \node[state,fill=white] (v_\phi) at (360/8 * \phi:2cm) {$v_\phi$};
  %       \draw[blue] (v_\phi) -- (center);
   %   }
   \end{tikzpicture}
       \caption{The part of $H$ associated to  $\U x_0$. The vertex labels in the graph denote the variable subscript. }
    \label{fig:U_graph}
\end{figure}
Now consider the columns of $A$ arising from exponents of $\F_0$. The polynomial $\F_0$ has no monomials which contain $x_0$, hence the first entry of all columns of $A$ arising from $\F_0$ is 0. In our one loop diagram, to obtain a 2-forest we remove 2 edges, and the monomials in $\F_0$ record these two edges which have been removed, it follows these columns contain exactly two 1s. Hence the columns of $A$ arising from $\F_0$ contain two ones and a zero in the first entry. This will lead to connected edges between pairs of vertices on the circular arc. Each such connection will yield a new minimal odd three cycle, as illustrated in Figure \ref{fig:F0_graph}. It follows that the graph $H$ obtained by this addition will satisfy the odd-cycle condition (Definition \ref{def:oddCycle}).

\begin{figure}[h]
    \centering

       \begin{tikzpicture}
\draw[blue,dotted] (250:2) arc (250:390:20mm);
   \draw[blue,dotted] (190:2) arc (190:250:20mm);
    \node[state,fill=white] (center) at (0,0) {$0$};
    \node[state,fill=white] (1) at (90 :2cm) {$1$};
    \draw[blue] (1) -- (center);
      \node[state,fill=white] (2) at (30 :2cm) {$2$};
    \draw[blue] (2) -- (center);
          \node[state,fill=white] (i) at (320 :2cm) {$i$};
    \draw[blue] (i) -- (center);
              \node[state,fill=white] (j) at (250 :2cm) {$j$};
    \draw[blue] (j) -- (center);
        \draw[blue] (i) -- (j);
        \node[state,fill=white] (Em1) at (190 :2cm) {\tiny$E-1$};
    \draw[blue] (Em1) -- (center);
            \node[state,fill=white] (E) at (140 :2cm) {$E$};
    \draw[blue] (E) -- (center);
%\foreach \phi in {1,...,6}{
 %   \node[state,fill=white] (v_\phi) at (360/8 * \phi:2cm) {$v_\phi$};
  %       \draw[blue] (v_\phi) -- (center);
   %   }
   \end{tikzpicture}
       \caption{The part of $H$ associated to  $\U x_0$ with edges added corresponding to monomials of $\F_0$. One such edge is illustrated connected the vertices $i$ and $j$ below.}    \label{fig:F0_graph}
\end{figure}
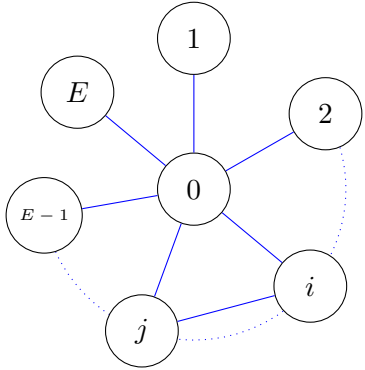
Now we consider the consequences of adding massive particles, that is edges with an associated mass in the Feynman graph $G$. In the polynomial $\mathcal{G}_h$ the addition of a massive edge in $G$ corresponds to the following product of polynomials
\begin{equation}\label{eq:massiveTerm}
    (x_1+x_2+\cdots+x_E)\cdot m_j^2x_j;
\end{equation}
this will contain the monomials $x_ix_j,\ i\neq j$ and $x_j^2$. {The square term corresponds to adding a loop at vertex $j$ of $H$. This odd-cycle is connected to vertex $0$ by a simple edge and thus at most separated by one edge from every odd-cycle corresponding to terms $x_ix_j,\ i\neq j$. Assume now that we have two internal masses, $m_i\neq0$ and $m_j\neq0$, then the loops they create at vertex $i$ resp. $j$ have to be connected by an edge for the odd-cycle condition to hold. This is true if and only if $\eqref{eq:CMCondExplain}$ holds, i.e.~if and only if the corresponding term in $\F_0+\F_m$ is non-vanishing.} 
%These corresponds to connecting the vertex $j$ to every other vertex in the graph $H$ and adding a loop at $j$. This ensures that the edge matrix $A$ contains all the lattice points in the edge polytope $\mathrm{conv}(A)$ and additionally that the graph obtained by adding the edges corresponding to this operation to the part of $H$ associated to $\U x_0$ satisfies the odd-cycle condition. 

\begin{figure}[h]
    \centering

       \begin{tikzpicture}
\draw[blue,dotted] (250:2) arc (250:390:20mm);
   \draw[blue,dotted] (190:2) arc (190:250:20mm);
    \node[state,fill=white] (center) at (0,0) {$0$};
    \node[state,fill=white] (1) at (90 :2cm) {$1$};
    \draw[blue] (1) -- (center);
      \node[state,fill=white] (2) at (10 :2cm) {$2$};
    \draw[blue] (2) -- (center);
          \node[state,fill=white] (i) at (320 :2cm) {$i$};
    \draw[blue] (i) -- (center);
              \node[state,fill=white] (j) at (235 :2cm) {$j$};
    \draw[blue] (j) -- (center);
\draw[blue] (j) to [out=180,in=210,looseness=11] (j);
        \draw[blue] (i) -- (j);
         \draw[blue] (1) -- (j);
         \draw[blue] (2) -- (j);
         \draw[blue] (E) -- (j);
          \draw[blue] (Em1) -- (j);
        \node[state,fill=white] (Em1) at (190 :2cm) {\tiny$E-1$};
    \draw[blue] (Em1) -- (center);
            \node[state,fill=white] (E) at (140 :2cm) {$E$};
    \draw[blue] (E) -- (center);
%\foreach \phi in {1,...,6}{
 %   \node[state,fill=white] (v_\phi) at (360/8 * \phi:2cm) {$v_\phi$};
  %       \draw[blue] (v_\phi) -- (center);
   %   }
   \end{tikzpicture}
       \caption{The part of $H$ associated to $\U x_0$ with edges added from Equation (\ref{eq:massiveTerm}). The term $x_j^2$ correspond to the self loop and the other terms connect $j$ to every other vertex in $H$ satisfying \eqref{eq:CMCondExplain}.}
    \label{fig:mass_graph}
\end{figure}
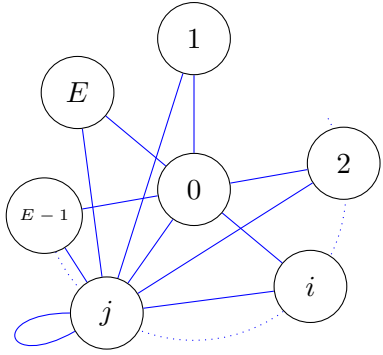
%All together, to obtain the graph $H$, we begin with the graph in Figure \ref{fig:U_graph} and add edges either of the type illustrated in Figure \ref{fig:U_graph} or in Figure \ref{fig:mass_graph}. In either case the resulting graph satisfies the odd-cycle condition. Note that we cannot have double edges as this would correspond to a repeated column in the matrix $A$, which is simply a monomial in $\mathcal{G}_h$ with a coefficient other than 1, however the Newton polytope does not record coefficients of monomials, but only their exponents. 

Since the graph $H$ satisfies the odd-cycle condition it has edge matrix $A$, the algebra $\NN A$
%$k[A]=k[y_1,\ldots,y_n]/I_A$ 
is normal by Theorem \ref{thm: odd-cycle}.
\end{proof}
We now illustrate this result on an example which is guaranteed to be normal, and hence Cohen-Macaulay, by the result of Theorem \ref{thm:Normal_1loop} but for which the earlier results of \cite{TH22} and \cite{walther2022feynman} do not apply. 

\begin{example}\label{ex: new normal}
The arguments in \cite{TH22,walther2022feynman} rest on the assumption that for every proper subset $V'\subset~V_\mathrm{ext}$ we have $\left(\sum_{v\in V'}p_v\right)^2\neq 0$. For on-shell massless diagrams this assumption fails, e.g. since $p_v^2=0$ for every $v\in V_\mathrm{ext}$. This means for example that the on-shell massless box-diagram with homogeneous Lee-Pomeransky polynomial 
    \begin{equation}
        \G_h=x_0(x_1+x_2+x_3+x_4)-sx_1x_3-tx_2x_4
    \end{equation}
is not covered by any of the previous results but is still normal due to Theorem \ref{thm:Normal_1loop}. The edge graph associated to $\G_h$ is shown in Figure \ref{fig: odd-cycle example} which clearly satisfy the odd-cycle condition implying, by Theorem \ref{thm:Normal_1loop}, that the semi-group $\NN A$ is normal, and hence the Cohen-Macaulay property holds (see also Figure \ref{fig:NormCMGPEtc}). {The number of master integrals can thus be calculated simply as $\mathrm{vol}(\newt(\G_h))=3$ which corresponds to the box integral itself along with the $s$- and $t$-channel bubble integrals.}
    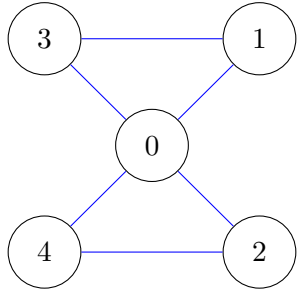
\begin{figure}[h]
    \centering
       \begin{tikzpicture}
    \node[state,fill=white] (center) at (0,0) {$0$};
    \node[state,fill=white] (1) at (45 :2cm) {$1$};
    \draw[blue] (1) -- (center);
      \node[state,fill=white] (3) at (135 :2cm) {$3$};
    \draw[blue] (3) -- (center);
          \node[state,fill=white] (4) at (225 :2cm) {$4$};
    \draw[blue] (4) -- (center);
        \node[state,fill=white] (2) at (315 :2cm) {$2$};
    \draw[blue] (2) -- (center);
    \draw[blue] (1) -- (3);
    \draw[blue] (2) -- (4);
   \end{tikzpicture}
       \caption{The edge graph in Example \ref{ex: new normal}; as before the vertex labels denote variable subscripts.}
    \label{fig: odd-cycle example}
\end{figure}
\end{example}
We especially note the following corollary of Theorem \ref{thm:Normal_1loop}.
\begin{corollary}\label{corollary_massless}
    If $G$ is a one-loop Feynman graph with $m_e=0$ for all edges $e\in E$, i.e. $\F_m=0$, or at most one edge has a non-zero mass. Then $\NN A$ is normal for all possible external kinematics.
\end{corollary}
Even though we have primarily discussed the Cohen-Macaulay property, the result above as well as the papers \cite{TH22,walther2022feynman} focus on proving normality of $\NN A$ which is a stronger criteria, see Figure \ref{fig:NormCMGPEtc}. Going to special kinematics one may find Feynman integrals with the Cohen-Macaulay property but which are not normal as well as two-loop integrals where even the Cohen-Macaulay property fails \cite{walther2022feynman}.    
%%%%%%%%%%%%%%%%%%%%%%%%%%%%%%%%%%%%%%%%%%%%%%%%%%%%%%%%%%%%%%%%%%%%%%%%%%%%
\subsection{Generalized permutohedra}\label{sec:GP}
 We again employ the notations introduced in Section \ref{sec:FeynmanGGKZ} with Symanzik polynomials $\U,\F$ giving $\G=\U+\F$. { The main contributions of this subsection are Proposition~\ref{prop: massive GP} and Theorem~\ref{thm: massive path}. Reinterpreting earlier results in the literature, in essence they demonstrate that the polytope $\newt(\F)$\footnote{It is well-known that $\newt(\U)$ is always GP, see the discussion after Proposition \ref{prop: IDP}} of a Feynman graph of arbitrary order is a {generalized  permutohedron (GP)} if:
 1) all internal propagators are massive for any external kinematics in the former case; 2) every vertex can be connected to an external vertex by a path of propagators that are all massive, and the graph is one-particle and one-vertex irreducible, with all external momenta offshell/massive in the latter case}. { This enlarges the class of integrals previously known to be GP~\cite{Schultka:2018nrs}, as we will also review in what follows.}
 
 The \emph{permutohedron} is a classical polytope with many special properties, for example, it is a simple zonotope, the monotone path polytope of a cube \cite{ziegler2012lectures} and the secondary polytope of a triangular prism $\Delta(1,2)\times\Delta(1,n)$, see \cite{gelfand2008discriminants}. More recently the \emph{generalized permutohedron} (GP) was introduced by Postnikov \cite{Postnikov2005}, and it was shown by Aguiar and Ardila that these polytopes are universal combinatorial representatives for a vast class of Hopf monoids \cite{Aguiar2017}. {In the physical context generalized permutohedra have facilitated new methods for fast Monte Carlo evaluation of Feynman integrals \cite{Borinsky:2020rqs,Borinsky:2023jdv}.}

 Simply put, a generalized permutohedron is any polytope whose normal fan is a coarsening of a permutohedron's normal fan. In the context of establishing normality of a polytope or semi-group, the following classification is more useful.
\begin{theorem}[{\cite[Theorem 12.3]{Aguiar2017}}]\label{thm: GP}
    A polytope $P\subset\RR^n$ is a generalized permutohedra if and only if every edge is parallel to $\mathbf{e}_i-\mathbf{e}_j$ for some $i,j\in\{1,\ldots,n\}$.
\end{theorem}
\noindent A {\em matroid polytope} is the convex hull of the indicator vectors of all bases of matroid; note these vectors have entries 0 or 1 only. We will say that a polytope has the GP property if Theorem \ref{thm: GP} is satisfied. This especially means that every matroid polytope \cite{GGMS,GelfandSerganova} has the GP property and that every polytope with the GP property is \emph{edge-unimodular}, i.e., the matrix of edge-directions is unimodular.

As used in this paper, normality of a set of lattice points $A=\{1\}\times A_-\subset\ZZ^{n+1}$ is a property of the semi-group $\NN A$ while the GP property is associated to a polytope. Even if $P=\conv(A_-)$, there is a priori no direct connection between the two properties, however, if $A_-$ is the full set of lattice points in $P$, i.e.,
\begin{equation}\label{eq: AlatticePoints}
    A_-=P\cap\ZZ^n
\end{equation}
then GP implies normality. This is because every polytope with the GP property also has the \emph{integer decomposition property} (IDP):
\begin{definition}[Integer decomposition property]
    A polytope $P\subset\RR^n$ is said to have the \emph{integer decomposition property} (IDP) if for every $k\in\ZZ_{>0}$ it satisfies
    \begin{equation}
        kP\cap\ZZ^n=P\cap\ZZ^n+(k-1)P\cap\ZZ^n.
    \end{equation}
\end{definition}
By a result of Howard (see, \cite{Howard2007}, cf. \cite{OWR}) every edge-unimodular polytope, and therefore especially every GP, has the IDP property. The significance of the IDP property in our setting is that it is  equivalent to $A$ being normal if $A_-=P\cap\ZZ^n$.

\begin{proposition}\label{prop: IDP}
Let $P\subset\RR^n$ be a polytope and $A_-=P\cap\ZZ^n$, then $P$ has the IDP if and only if $A=\{1\}\times A_-$ is normal.
\end{proposition}
\begin{proof}
    Assume $P$ has the IDP, then for every integer $k>0$ we have that $a\in kP\cap\ZZ^n$ implies there exists $a_1,\ldots,a_k\in A_-$ such that $a=a_1+\cdots+a_k$. By the construction of $A$ we may choose a basis for $\ZZ A$ such that the first coordinate is $1$ and the remaining entries are a basis for $\ZZ A_-$. An arbitrary point in $\ZZ A\cap\RR_{\ge 0}A\subset\ZZ^{n+1}$ has the form $(k,a)$ for some integer $k>0$ and where $a\in \ZZ A_-$, but $A_-=P\cap \ZZ^n$, so $a\in kP$ and the IDP implies that $a=a_1+\cdots+a_k$ for some $a_1,\cdots,a_k\in A_-$ and thus $(k,a)=(1,a_1)+\cdots+(1,a_k)$. Therefore normality is proven.

    Now, assume that $\NN A$ is normal, that is, we can write every element $(k,a)\in\RR_{\ge 0}A\cap\ZZ A$ as $(k,a)=(1,a_1)+\cdots+(1,a_k)$ for $a_1,\ldots,a_k\in A_-$; this directly implies the IDP.
\end{proof}

By definition $\newt(\U)$ is the matroid polytope of the dual matroid to the Feynman graph, meaning that not only does $\newt(\U)$ have the GP property but also $\supp(\U)=\newt(\U)\cap\ZZ^{{n}}$ so the semi-group generated by $\supp(\U)$ is normal. 

Properties connected to the $\F$-polynomial are much more intricate as they are always dependent on the kinematic setup and not just on the underlying graph, however, some general statements are known. For example, Schultka proved in \cite{Schultka:2018nrs} that $\newt(\F)$ is a GP in the Euclidean regime with generic kinematics. Here we provide a slight generalization:
\begin{proposition}\label{prop: massive GP}
    Assume $m_e\neq 0$ for all $e\in E$, then $\newt(\F)$ is a GP for all possible choices of external kinematics.
\end{proposition}
The proof of this statement is contained in \cite{TH22} but not in reference to the GP property. We state it here for completeness.
\begin{proof}
    When all internal masses are non-zero all vertices of $\newt(F)$ must always come from $\U\cdot\sum_{e\in E}m_e^2x_e\subset\F$, i.e.
    \begin{equation}
        \newt(\F)=\newt\left(\U\cdot\sum_{e\in E}m_e^2x_e\right)=\newt\left(\U\right)+\newt\left(\sum_{e\in E}m_e^2x_e\right).
    \end{equation}
    Since $\newt\left(\sum_{e\in E}m_e^2x_e\right)$ is just the standard simplex $\Delta(1,{n})$, which is a GP, and $\newt(\U)$ is a GP, this means that $\newt(\F)$ is a GP since the GP property is closed under Minkowski addition \cite{Aguiar2017}. 
\end{proof}

Another case is contained in \cite{TH22}, assume all internal masses are zero, i.e. $m_e=0$ for all $e\in E$, and that every two-forest of $G$ comes with a non-zero coefficient. This last assumption means that $V=V_{\mathrm{ext}}$ and that $p(V')^2\neq 0$ for all $V'\subset V$ where $p(V')=\sum_{v\in V'}p_v$. For this setup $\newt(\F)$ is a matroid polytope and hence a GP.

This result was generalized by Walther in \cite{walther2022feynman} where he managed to remove the assumption $V_{\mathrm{ext}}=V$. He showed that $\newt(\F)$ is a matroid polytope, and hence GP, if $m_e=0$ for all $e\in E$ and $p(V')^2\neq 0$ for all $V'\subset V_{\mathrm{ext}}$. The assumption placed on the external momenta essentially means that they behave as Euclidean vectors in the sense that (i) $p_v^2\neq 0$ for all $v\in V_{\mathrm{ext}}$
and (ii) $(p_v+p_u)^2\neq0$ for all $u,v\in V_{\mathrm{ext}}$. Neither of these two assumptions are true in the general Minkowski setting.

As long as the equality
\begin{equation}
    \newt(\F)=\newt\left(\U\cdot\sum_{e\in E}m_e^2x_e\right)
\end{equation}
holds for a Feynman graph (with at least one $m_e\neq 0$), it is clear that $\newt(\F)$ is a GP by the arguments in the proof of Proposition \ref{prop: massive GP}. This can be rephrased to the statement, that a sufficient condition for $\newt(\F)$ to be a GP is that $\supp(\F_0)\subseteq\supp(\F_m)$. Since the terms in $\F_0$ come from the two-forests of the graph, one way this can be true is if $\F_m$ contains terms from all two-forests. {This is guaranteed if every vertex of $G$ is connected to an external vertex by a massive path, i.e. a path of consecutive edges all with non-zero mass.}

\begin{theorem}[Theorem 4.5 in \cite{walther2022feynman}]\label{thm: massive path} 
Let $G$ be a one-particle irreducible and one-vertex irreducible Feynman graph such that no cancellation between $\F_0$ and $\F_m$ occurs and $p(V')\neq 0$ for all $V'\subset V_{\mathrm{ext}}$. Then every term in $\F_0$ also appears in $\F_m$ (i.e. $\supp(\F_0)\subseteq\supp(\F_m)$) if and only if every $v\in V$ has a massive path to an external vertex $v'\in V_{\mathrm{ext}}$. 
\end{theorem}
\noindent {A direct consequence is that $\newt(\F)$ is a GP for every Feynman graph satisfying Proposition \ref{prop: massive GP} or Theorem~\ref{thm: massive path}. In addition to the GP property being a desirable property, see for example the discussion at the beginning of this section, GP also implies the Cohen-Macaulay property holds, see Figure \ref{fig:NormCMGPEtc}, and hence that the number of master integrals may be calculated from the volume of the associated polytope independent of generalized propagator powers and space time dimension.}
%%%%%%%%%%%%%%%%%%%%%%%%%%%%%%%%%%%%%%%%%%%%%%%%%%%%%%%%%%%%%%%%%%%%%%%%%%%%%%%%%%%%%%%%%%%%%

\section{Conclusions and outlook}\label{sec: outlook}
In this paper we have recast the problem of determining the symbol alphabet of a polylogarithmic Feynman integral as the question of factorizing its principal $A$-determinant, which encodes its Landau singularities and may be obtained independently of the standard procedure of its analytic evaluation. We have primarily studied one-loop Feynman integrals. Our main results are the formulas for their symbol alphabet \eqref{eq: letter one-contracted}-\eqref{eq: rational letter} and canonical differential equations \eqref{eq: CDE D0+E even}-\eqref{eq: CDE D0+E odd} together with a \ma\ code for their automatic evaluation, as well as the proof that normality, and hence the Cohen-Macaulay property, holds in Theorem \ref{thm:Normal_1loop}. These results are complimented with the limiting procedure for specialized kinematics in subsection \ref{sec: limiting procedure}, also implemented in \ma, and a discussion of the generalized permutohedron in subsection \ref{sec:GP}.

While the main focus of this paper is on one-loop graphs we are also optimistic that the approach to obtaining the symbol letters via the principal $A$-determinant described in this note will also apply in more generality. 
To this end we conclude with a simple example of a two-loop graph where the principal $A$-determinant gives the symbol alphabet. We consider the slashed box with two different choices for one off-shell leg.
\begin{equation*}
    \centering
     \begin{tikzpicture}[baseline=-\the\dimexpr\fontdimen22\textfont2\relax]
    \begin{feynman}
    \vertex (a);
    \vertex [right = of a] (b);
    \vertex [below = of b] (c);
    \vertex [below = of a] (d);
    \vertex [above left = of a] (x){\(p_1\)};
    \vertex [above right = of b] (y){\(p_4\)};
    \vertex [below right = of c] (z){\(p_3\)};
    \vertex [below left = of d] (w){\(p_2\)};
    \diagram*{
        (x)--[fermion](a), (y)--[fermion](b), (z)--[fermion](c), (w)--[fermion](d), (a) -- [edge label=\(x_1\)] (b) -- [edge label=\(x_4\)] (c) --[edge label=\(x_3\)] (d) -- [edge label=\(x_2\)](a), (d)--[edge label=\(x_5\)] (b)
    };
    \end{feynman}
    \end{tikzpicture}
\end{equation*}
The first Symanzik polynomial is independent of the kinematics and is therefore the same in all the following different cases:
\begin{equation}
    \U=x_1x_3+x_1x_4+x_1x_5+x_2x_3+x_2x_4+x_2x_5+x_3x_5+x_4x_5.
\end{equation}

The simplest one-mass case is the off-shell leg being connected to the internal diagonal, e.g. $p_2^2\neq 0$ while $p_1^2=p_3^2=p_4^2=0$. This gives the $\F$-polynomial
\begin{equation}
    \F=-sx_1x_3x_5-tx_2x_4x_5-p_2^2x_2x_3x_5.
\end{equation}
The Newton polytope $\newt(\F)$ is not a generalized permutohedron (GP) and $\newt(\G)$ is not edge-unimodular, however, $I_A$ has a radical initial ideal and therefore $\newt(\G)$ has a regular unimodular triangulation so $\NN A$ is normal. At one-loop and with generic kinematics this means that the number of master integrals equals the volume of $\newt(\G)$. This is no longer true at two loops, since Gale duality is no longer enough to fix all coefficients in $\U$ to be one, as required for Feynman integrals. The physically interesting case now is a restriction ideal of $H_A(\beta)$.  In physical variables (and with all coefficients in $\U$ one) we obtain that the reduced principal $A$-determinant is:
\begin{equation*}
    \widetilde{E_A}(\G)=(p_2^2-s-t)(p_2^2-s)(p_2^2-t)stp_2^2.
\end{equation*}
With the three variables $z_1=s/p_2^2,\ z_2=t/p_2^2$ and $z_3$ satisfying $z_1+z_2+z_3=1$ we get
\begin{equation}
    \widetilde{E_A}(\G)\propto z_3(1-z_1)(1-z_2)z_1z_2.
\end{equation}
These five factors constitute all but one letter in the symbol alphabet for two-dimensional harmonic polylogarithms~\cite{Gehrmann:2000zt}, known to be the appropriate class of functions for describing all four-point two-loop master integrals with one offshell leg, and hence also the slashed box integrals discussed here. 

The other one-mass configuration has $\F$-polynomial
\begin{equation*}
    \F=-sx_1x_3x_5-tx_2x_4x_5-p_1^2x_1x_2(x_3+x_4+x_5).
\end{equation*}
Again $\newt(\F)$ is not GP and $\newt(\G)$ is not edge-unimodular, however, $I_A$ has a radical initial ideal and therefore $\newt(\G)$ has a regular unimodular triangulation. So $\NN A$ is normal, where $A$ is the support of $\G$,
\begin{equation}
    A=\begin{pmatrix}
        1&1&1&1&1&1&1&1&1&1&1&1&1\\
        1&1&1&0&0&0&0&0&1&0&1&1&1\\
        0&0&0&1&1&1&0&0&0&1&1&1&1\\
        1&0&0&1&0&0&1&0&1&0&1&0&0\\
        0&1&0&0&1&0&0&1&0&1&0&1&0\\
        0&0&1&0&0&1&1&1&1&1&0&0&1
    \end{pmatrix}.
\end{equation}
As before we use the physically relevant setup and work in physical variables. The reduced principal $A$-determinant is:
\begin{equation*}
    \widetilde{E_A}(\G)=(p_1^2-t)(p_1^2-s)(p_1^2-s-t)(s+t)stp_1^2.
\end{equation*}
With the same change of variables as before (but with $p_2$ taking the role of $p_1$) we get
\begin{equation}
    \widetilde{E_A}(\G)\propto (1-z_2)(1-z_1)z_3(1-z_3)z_1z_2,
\end{equation}
which is the {\em full symbol alphabet for the two-dimensional harmonic polylogarithms}.

At this level the discriminants which needed to be calculated are quite sizeable. For example, not only does the full $A$-discriminant in this case have degree 14, but the discriminant corresponding to the face
\begin{equation}
    A\cap\Gamma=\begin{pmatrix}
        1&1&1&1&1&1&1&1&1&1\\
        1&1&0&0&0&0&1&0&1&1\\
        0&0&1&1&0&0&0&1&1&1\\
        1&0&1&0&1&0&1&0&1&0\\
        0&1&0&1&0&1&0&1&0&1\\
        0&0&0&0&1&1&1&1&0&0
    \end{pmatrix}
\end{equation}
has degree 20.
Even though computational complexity grows quickly at higher loops, and we have the fact that the physically interesting ideals will be restriction ideals of $H_A(\beta)$, this nontrivial two-loop example of the two-loop slashed box integral with one offshell leg not only shows that its principal $A$-determinant may still be computed directly; but also that it yields the full symbol alphabet of the integral in question.
%%%%%%%%%%%%%%%%%%%%%%%%%%%%%%%%%%%%%%%%%%%%%%%%%%%%%%%%%%%%%%%%%%%%%%%%%%%%%%%%%%%%%%%%%%%%%%%%%%%

\acknowledgments

We have benefited from stimulating discussions with Ekta Chaubey, Einan Gardi, Sebastian Mizera, Ben Page and Simone Zoia. GP and FT acknowledge
support from the Deutsche Forschungsgemeinschaft under Germany’s Excellence Strategy – EXC 2121 “Quantum Universe” – 390833306.
While preparing this work Martin Helmer was partially supported by the Air Force Office of Scientific Research under award
number FA9550-22-1-0462.

\bibliography{Ref}% Produces the bibliography via BibTeX.
\end{document}